\documentclass[journal,11pt,onecolumn,draftclsnofoot,]{IEEEtran}
\usepackage{xcolor}
\usepackage{amsmath, amsthm}
\usepackage{amsfonts}
\usepackage{mathtools}
\usepackage{algorithmic}
\usepackage{algorithm}
\DeclarePairedDelimiterX{\norm}[1]{\lVert}{\rVert}{#1}
\usepackage{graphicx}
\usepackage{subcaption}
\usepackage{dsfont}
\usepackage{bbm}
\usepackage{cite}
\usepackage{romannum}
\usepackage{enumitem}

\newtheorem{assumption}{Assumption}
\newtheorem{theorem}{Theorem}
\newtheorem{corollary}{Corollary}
\newtheorem{lemma}{Lemma}

\begin{document}
\pagenumbering{arabic}
\title{Asymptotic Analysis of RLS-based Digital Precoder with Limited PAPR in Massive MIMO}

\author{Xiuxiu~Ma, Abla~Kammoun,~\IEEEmembership{Member,~IEEE}, Ayed~M.~Alrashdi,~\IEEEmembership{Member,~IEEE}, Tarig~Ballal,~\IEEEmembership{Member,~IEEE},
Tareq~Y.~Al-Naffouri,~\IEEEmembership{Senior Member,~IEEE}
and~Mohamed-Slim~Alouini,~\IEEEmembership{Fellow,~IEEE}
\thanks{X. Ma, A. Kammoun, T. Ballal,T. Y.  Al-Naffouri and M. Alouini are with the Division of Computer, Electrical and Mathematical Science \& Engineering, King Abdullah University of Science and Technology (KAUST), Thuwal, KSA. A. M. Alrashdi is with the Department of Electrical Engineering, College of Engineering, University of Ha'il, P.O. Box 2440, Ha'il, 81441, Saudi Arabia. E-mails: (\{xiuxiu.ma; abla.kammoun; tarig.ahmed;tareq.alnaffouri;slim.alouini\}@kaust.edu.sa; am.alrashdi@uoh.edu.sa)}}



\maketitle

\begin{abstract}
This paper focuses on the performance analysis of a class of limited peak-to-average power
ratio (PAPR) precoders for downlink multi-user massive multiple-input multiple-output
(MIMO) systems. Contrary to conventional precoding approaches based on simple linear
precoders such as maximum ratio transmission (MRT) and regularized zero-forcing (RZF),
the precoders in this paper are obtained by solving a convex optimization problem. To be
specific, these precoders are designed so that the power of each precoded symbol entry is restricted, and the PAPR at each antenna is tunable. By using the Convex Gaussian
Min-max Theorem (CGMT), we analytically characterize the empirical distribution of the
precoded vector and the joint empirical distribution between the distortion and the intended
symbol vector. This allows us to study the performance of these precoders in terms of per-antenna power, per-user distortion power, signal-to-noise and distortion ratio (SINAD),
and bit error probability. We show that for this class of precoders, there is an optimal
transmit per-antenna power that maximizes the system performance in terms of SINAD
and bit error probability.
\end{abstract}


\section{Introduction}
\IEEEPARstart{M}{assive} multiple-input multiple-output (MIMO) systems are recognized among the  key enabling technologies for next-generation communication systems\cite{intro1,intro2, intro3}.  
However, there are still major implementation issues to address for massive MIMO systems to be a reality. First, the number of antennas that can be supported is limited by the transceiver's form factor. In practice, this issue can be handled by moving the operating frequency to mmWave frequency bands~\cite{intro_mm}.  Second, it  requires equipping each antenna with a dedicated radio frequency (RF) chain, which allows the pass-band communication signals to be processed in the base-band~\cite{l4}, thereby leading to a prohibitively high cost and power consumption, calling into question the practicality of such systems. One possible solution is to reduce the number of RF chains by employing hybrid-precoding \cite{r1,intro4}. However, the work in \cite{refpapr} shows that the power consumption of some hybrid precoding architectures still scales with the number of antennas, and proposed as a solution a linear and highly energy efficient reflect-array and transmit-array antennas scheme. Aside from power consumption, it is of interest to control the power of each RF chain allowing for cheap system modules. In light of this observation, the work in \cite{papr} proposed a precoder with fewer RF chains that constrains the power of each signal entry to be below a certain threshold.

The proposed precoder in \cite{papr} reminds the conventional regularized zero-forcing precoder (RZF)~\cite{intro_rzf}, in that it builds on the regularized least squares (RLS) method to minimize a penalty of the residue sum of squares (RSS). The main difference with RZF is that it constrains the absolute value of the precoded vector to not exceed a certain threshold. Referred to as the RLS-based precoder with limited peak-to-average power ratio (PAPR), the precoder in \cite{papr} allows for achieving the two sought-for goals, that is a lower number of RF chains together with a limited PAPR, making it possible to use inexpensive power amplifiers.

\subsection{Contributions and related works}
\noindent{\bf Performance analysis of non-linear precoders.} In this paper, we carry out a rigorous, asymptotic characterizaton of the performance of multi-user downlink transmission when an RLS precoder with limited PAPR is employed. More precisely, we study the asymptotic behavior of the per-antenna power, per-user distortion power, signal-to-noise and distortion ratio (SINAD), and bit error probability when the number of antennas and the number of served users grow large at the same pace. A similar problem has been recently studied in \cite{glse,8100647} where asymptotic expressions for the distortion error and a lower bound on the achievable rate have been derived. Compared to these works, our contribution differs as follows.  On a methodological level, while the works in \cite{glse,8100647} are based on the \emph{non-rigorous} replica method, the main tool in our work is the recently developed Convex Gaussian Min-max Theorem (CGMT) \cite{cgmt} framework. Using the CGMT, our analysis goes beyond the performance metrics studied in \cite{glse,8100647}. Particularly, we assume BPSK modulation while  \cite{glse,8100647} rely on Gaussian signaling. Furthermore, we derive accurate characterization of the joint distribution between the transmitted symbol vector and the distortion error. This characterization allows us to analyze the bit error  probability and a tight approximation of the SINAD. On an operational level, we derive several insights from our analysis by studying  the obtained asymptotic expressions in different regimes describing small numbers of served users  or small/large values of the power control parameter. Particularly, we show that the performance of the RLS precoder with limited PAPR is not always better when the transmit per-antenna power increases, for higher transmit power may also imply higher distortion power. In other words, there is an optimal per-antenna transmit power that maximizes the performance in terms of SINAD and bit error probability. It  can be achieved by properly setting the power control parameter.

\noindent{\bf Convex Gaussian Min-max Theorem.} The main ingredient of the proof of our main results is the Convex Gaussian Min-max Theorem (CGMT). This framework has been initiated by Stojnic \cite{stojnic} before being  formally developed  in \cite{cgmt} and \cite{mesti}. It has been applied to characterize the asymptotic behavior of  convex-optimization-based estimators with application to high-dimensional regression problems as well as binary classification problems. In this line, the work in \cite{mesti} applied the CGMT to quantify the performances of several estimators including the Least Absolute Shrinkage and  Selection Operator (LASSO). As far as wireless communications are concerned, the CGMT has been applied to characterize the performance of non-explicit decoders. In this context, under the assumption of real Gaussian channels, the CGMT was used to derive closed-form approximations of the bit error probability of convex-optimization-based decoders termed box relaxation decoders under Binary Phase Shift Keying (BPSK) signaling~\cite{linearC,linearJ} as well as M-ary Pulse Amplitude Modulation (M-PAM) signaling~\cite{nonlinear}. All these works have focused on the design of non-linear algorithms from the decoder perspective. The application of the CGMT for the design of non-linear precoders has not, to the best of our knowledge, been studied, which motivates our work. 
 
\subsection{Paper Organization}
In Section~\ref{sec2}, we introduce the system model and formulate the problem. Then, in Section~\ref{sec3} we state our main results characterizing the statistics of the precoded vector and the distortion, based on which we get the asymptotic behavior of the studied specifications and performance metrics, namely, the transmit per-antenna power, the per-user distortion power, the received SINAD, and the bit error probability.  In Section~\ref{sec4}, numerical simulations are provided to confirm the accuracy of our results before concluding the paper in Section~\ref{sec9}. For the readers' convenience, the proofs are deferred to Section~\ref{sec5}-\ref{sec8}.

\subsection{Notations}
For simplicity, we make use of the following notations onwards. 

Our work is to characterize the behaviors of the RLS-based precoder with limited PAPR, in the large dimensional regime where $m,n\rightarrow\infty$ with a fixed ratio $\delta=m/n$, and to keep the notations short we simply write $n\rightarrow\infty$. We say that an event $\xi$ holds with probability approaching $1$ (\emph{w.p.a.1}) if $\mathrm{lim}_{n\rightarrow\infty}\mathbb{P}[\xi]=1$. If a sequence of random variables $X_n$ converges to a constant $X$, we write $X_n\stackrel{P}{\rightarrow} X$.  
 For any vector $\mathbf{x}$, we use $x_i$ or $[{\bf x}]_i$ to denote its $i$-th element. We also use the notation  $\|\cdot\|$ to denote the Euclidean norm, and the notation $(x)_{+}$ to denote $\max(x,0)$. We write $f(x)$ as $O(g(x))$ if there are constants $M$ such that $|f(x)|\leq Mg(x)$ for all $x$ going to the limiting value in the analysis.

The empirical distribution of a vector ${\bf t}\in\mathbb{R}^{m}$ is given by $\frac{1}{m}\sum_{i=1}^{m}\boldsymbol{\delta}_{t_i}$ where $\boldsymbol{\delta}_{t_i}$ is the Dirac delta mass at $t_i$. 
 For $q\in\mathbb{N}$,  a function $f:\mathbb{R}^q\to \mathbb{R}$ is said to be pseudo-Lipschitz of order $k$ if for all ${\bf x}$ and ${\bf y}$ in $\mathbb{R}^{q}$, $\left|f({\bf x})-f({\bf y})\right|\leq C(1+\|{\bf x}\|^{k-1}+\|{\bf y}\|^{k-1})\|{\bf x}-{\bf y}\|$. The Wasserstein$-k$ distance~\cite{wd} between two measures $\mu$ and $\nu$ is defined as $W_k(\mu,\nu)=\left(\inf_{\rho}\mathbb{E}_{(X,Y)\sim \rho}|X-Y|^{k}\right)^{\frac{1}{k}}$ where the infimum is over all random variables $(X,Y)$ such that $X\sim \mu$ and $Y\sim \nu$ marginally. A sequence of probability distributions $\nu_p$ converges in $W_k$ to $\nu$ if $W_k(\nu_p,\nu)\to 0$ as $p\to\infty.$ An equivalent definition of the convergence in $W_k$ is that, for any $f$ pseudo-Lipschitz of order $k$, $\lim_{p\to \infty} \mathbb{E}[f(X_p)] =\mathbb{E}[f(X)]$ where the expectation is with respect $X_p\sim\nu_p$ and $X\sim \nu$.

\section{System model and problem formulation}
\label{sec2}
Consider a conventional multiuser downlink, slow narrow band transmission between a base station equipped with $n$ transmit antennas and $m$ single antenna user terminals. The precoding scheme is a function that maps the user information symbols, collected in $\mathbf{s}=[s_1,s_2,...,s_m]^{T}=\{{\pm}1\}^m$ and assumed to be drawn uniformly from the BPSK constellation, into an $n$-dimensional signal $\mathbf{x}=[x_1,x_2,...,x_n]^{T}$. Since the signal here is BPSK, we assume a real wireless channel and additive noise.
Letting ${\bf h}_k$ denote the channel vector between the base station and user $k$, the received signal at the $k$-th user writes as
\begin{equation}
{y}_k={\bf h}_k^{T}{\bf x}+z_k,
\label{yk}
\end{equation}
where ${ z}_k$ is the additive noise, assumed to follow a Gaussian distribution with mean zero and variance $\sigma^2$. Stacking the received signals into a vector ${\bf y}=\left[y_1,\cdots,y_m\right]^{T}$ yields 
\begin{equation*}
{\bf y}={\bf H}{\bf x}+{\bf z},
\end{equation*}
where ${\bf z}=\left[z_1,\cdots,z_m\right]^T$ and ${\bf H}=\left[{\bf h}_1,\cdots,{\bf h}_m\right]^{T}$. 

The main goal of precoding is to remove the effect of the channel by minimizing the error between the channel-distorted received vector ${\bf Hx}$ and the information vector ${\bf s}$. To meet this requirement, the non-linear least squares precoder proposed in \cite{8100647} is formulated as the solution to the following regularized least squares problem:
\begin{equation}
\hat{\bf x}=\arg\min_{{\bf x}\in\mathbb{X}^{n}} \|{\bf Hx}-\sqrt{\rho}{\bf s}\|^2+ \lambda\|{\bf x}\|^2, \label{eq:nonlse}
\end{equation}
where ${\rho}$ is a positive power control factor, $\lambda$ is a positive regularization parameter and $\mathbb{X}$ being a predefined set containing admissible values for the precoded signal. 
The formulation in \eqref{eq:nonlse} defines a whole class of precoded vectors for different choices of the set $\mathbb{X}$ and parameter $\lambda$. For example, if  $\mathbb{X}=\mathbb{R}$ and $\lambda>0$, we obtain the RZF precoding given by 
\begin{equation*}
\hat{\bf x}_{\rm RZF}= \sqrt{\rho}\left({\bf H}^{T}{\bf H}+\lambda {\bf I}_n\right)^{-1}{\bf H}^{T}{\bf s},
\end{equation*}
which for $\lambda=0$ reduces to the zero-forcing (ZF) precoding (assuming $m>n$)
\begin{equation*}
\hat{\bf x}_{\rm ZF} = \sqrt{\rho}\left({\bf H}^{T}{\bf H}\right)^{-1} {\bf H}^{T}{\bf s}.
\end{equation*}
When $\mathbb{X}=[-\sqrt{P},\sqrt{P}]$, the precoder in \eqref{eq:nonlse} does not admit
a closed-form expression, characterizing the performance of the precoder is a challenging task. In this paper, we aim to study its performance in the large dimensional regime in which the number of antennas $n$ and the number of users $m$ grow large at the same pace. More formally, in our analysis, we rely on the following assumptions. 
\begin{assumption}
The number of antennas $n$ and the number of users $m$ grow to infinity at a fixed ratio $\delta:=\frac{m}{n}$. 
 \label{ass:regime}
\end{assumption}
\begin{assumption}
The channel matrix ${\bf H}$ has independent and identically distributed Gaussian entries with zero mean  and a variance equal to $\frac{1}{n}$.
\label{ass:statistic} 
\end{assumption}
We are interested in characterizing the performance of the precoder in \eqref{eq:nonlse} with respect to the following  specifications and performance metrics. 

\noindent{\bf Per-antenna power:} We define the per-antenna transmit power as 
\begin{equation}
P_b:=\frac{\|\hat{\bf x}\|^2}{n}.
\label{Pb_def}
\end{equation}

\noindent{\bf Per-user distortion error power:} By expressing the received signal at user $k$ as
\begin{equation}
y_k=\sqrt{\rho}{s_k} + {\bf h}_{k}^{T}\hat{\bf x}- \sqrt{\rho}{s_k} + z_k,
\label{yk1}
\end{equation}
the distortion error observed by user $k$ is represented by the quantity ${\bf h}_k^{T}\hat{\bf x}-\sqrt{\rho}s_k$. We define the per-user distortion error power as 
\begin{equation}
P_d:= \frac{\|{\bf H\hat{\bf x}}-\sqrt{\rho}{\bf s}\|^2}{m}.
\label{thePd}
\end{equation}

\noindent{\bf Average per-user SINAD:}
From \eqref{yk1}, we can easily see that the ${\rm SINAD}$ at user $k$ is given by
\begin{equation*}
{\rm SINAD}_k= \frac{\rho }{\mathbb{E}_{s_k}\left|{\bf h}_k^{T}\hat{\bf x}-\sqrt{\rho}s_k\right|^2 +\sigma^2}.
\end{equation*}
We define the average per user SINAD as:
\begin{equation}
\overline{\rm SINAD}=\frac{1}{m}\sum_{k=1}^m \mathbb{E}\left[{\rm SINAD}_k\right].
\label{eq:SINR_average}
\end{equation}

\noindent{\bf Average per-user SINAD upper bound and lower bound:} 
From Jensen's inequality, we can easily check that the expected value of the SINAD at user $k$ can be upper bounded and lower bounded as
\begin{equation}
\mathbb{E}[{\rm SINAD}_k] \leq \mathbb{E} \left[\frac{\rho}{\left|{\bf h}_k^{T}\hat{\bf x}-\sqrt{\rho}s_k\right|^2+\sigma^2}\right], \label{eq:average}
\end{equation}
\begin{equation}
\mathbb{E}[{\rm SINAD}_{k}]\geq  \frac{\rho}{\mathbb{E}\left[\left|{\bf h}_k^{T}\hat{\bf x}-\sqrt{\rho}s_k\right|^2\right]+\sigma^2}. \label{eq:average_lb}
\end{equation}
From \eqref{eq:average} and \eqref{eq:average_lb}, we can prove that the following quantities define upper and lower bounds for the average per-user SINAD:
\begin{align}
	{\rm SINAD}_{\rm up}&=\frac{1}{m}\sum_{k=1}^m \mathbb{E}\left[\frac{\rho}{|{\bf h}_k^{T}\hat{\bf x}-\sqrt{\rho}s_k|^2+\sigma^2}\right], \label{eq:up}\\
	{\rm SINAD}_{\rm lb}&=\frac{\rho}{\mathbb{E}\left[\frac{1}{m}\sum_{k=1}^m |{\bf h}_k^{T}\hat{\bf x}_k-\sqrt{\rho}s_k|^2\right]+\sigma^2}.\label{eq:lb}
\end{align}
 
In practice, we predict  the SINAD lower bound in \eqref{eq:lb} to provide a tight approximation for the SINAD. Indeed, under Assumption \ref{ass:statistic}, all users experience the same channel statistics, we expect that
\begin{equation*}
\mathbb{E}_{s_k}\left|{\bf h}_k^{T}\hat{\bf x}-\sqrt{\rho}s_k\right|^2
\end{equation*}
is asymptotically close to $ \frac{1}{m} \sum_{k=1}^m \left|{\bf h}_k^{T}\hat{\bf x}-\sqrt{\rho}s_k\right|^2$.  This latter term should  converge to its expectation $ \mathbb{E}\left[\frac{1}{m} \sum_{k=1}^m \left|{\bf h}_k^{T}\hat{\bf x}-\sqrt{\rho}s_k\right|^2\right]$, and hence substituting it by its expectation leads to ${\rm SINAD}_{\rm lb}$. Therefore, in Section \ref{sec4}, we compare the empirical SINAD with our approximation for ${\rm SINAD}_{\rm lb}$ and validate its accuracy.  

\noindent{\bf Bit error rate and bit error probability:}
The bit error rate (BER) is defined as 
\begin{equation}
\mathrm{BER}:=\frac{1}{m}\sum_{i=1}^{m}{\bf 1}_{\{\mathrm{sign}(y_i)\neq s_i\}},
\label{def_ber}
\end{equation}
where ${\bf 1}_{\{\cdot\}}$ denotes the indicator function. Another related quantity of interest is the bit error probability $P_e$, which is defined as the expectation of the BER, i.e.,
\begin{equation}
P_e:=\mathbb{E}[\mathrm{BER}]=\frac{1}{m}\sum_{i=1}^{m}\mathbb{P}[\mathrm{sign}(y_i)\neq s_i].
\label{def_bep}
\end{equation}


Although both the RZF and ZF precoding admit a closed-form expression, the PAPR at each RF chain is not restricted. Taking the definition in \cite{papr}, the per-antenna PAPR is defined as follows:
\begin{equation}
{\mathbf{ PAPR}}_i:=\left(\frac{1}{J}\sum_{j=1}^{J}|x_i(j)|^2\right)^{-1}\max_{j\in[J]}|x_i(j)|^2,
\label{papr}
\end{equation}
where ${\bf x}(j)$ is the $j$-th realization of the transmit vector, and $J$ is the number of samples. We claim that
\begin{equation}
\lim_{J\to\infty}\mathbb{E}\left[\frac{1}{J}\sum_{j=1}^{J}|x_i(j)|^2\right]=\lim_{n\to\infty}\mathbb{E}\left[\frac{1}{n}\sum_{i=1}^{n}|x_i(j)|^2\right], 
\label{papr2}
\end{equation}
	which means the average PAPR over antennas equals that over time. The above result, verified by simulations, follows because the channel statistics over all antennas are the same, so there is no reason that one antenna experiences more power than any other one. As illustrated later in the following sections, the limiting of the per-antenna power represented by the right-hand side of \eqref{papr2} can be tuned to any target value by properly choosing the power control parameter. Moreover, by construction, the entries of the precoded vector are in the set  $\mathbb{X}=\left[-\sqrt{P},\sqrt{P}\right]$, where $P$ is carefully chosen so that the peak value of the precoded vector is restricted.  
	As a consequence, the studied precoder achieves a PAPR that is less than $P/\mathcal{E}$ where $\mathcal{E}$ is the target power value.
\eqref{eq:nonlse} will be thus referred to as the RLS-based precoder with limited PAPR and for simplicity termed as limited PAPR-RLS precoder.

\section{Main results} 
\label{sec3}
\subsection{Distributional characterization of the precoded vector and the distortion error}
A major result of our study is the theoretical characterization of the empirical distributions of the elements of the precoded vector $\hat{\bf x}$ and the 
joint empirical distribution of the distortion error vector given by 
\begin{equation*}
\hat{\bf e}=\left({\bf H}\hat{\bf x}-\sqrt{\rho}{\bf s}\right)
\end{equation*}
and  the transmitted symbol ${\bf s}$. As shown next, both of these distributional characterizations will be instrumental in sharply characterizing the convergences of the specifications and performance metrics introduced in the previous section. 

\begin{theorem}[Distributional characterization of the precoded vector]
	Consider the following max-min optimization problem:
	\begin{equation}
		\overline{\phi}=\max_{\beta\geq 0}\min_{\tau\geq 0}  \frac{\tau\beta\delta}{2} +\frac{\rho\beta}{2\tau} -\frac{\beta^2}{4} +Y(\beta,\tau),
		\label{eq:sec}
	\end{equation}
	where
	\begin{equation*}
	Y(\beta,\tau)=\frac{\beta}{\alpha} \left(\mathbb{E}_{H\sim\mathcal{N}(0,1)}\left[(H-\sqrt{P}\alpha)^2{ \bf 1}_{\left\{H\geq \sqrt{P}\alpha\right\}}\right]-\frac{1}{2}\right),
	\end{equation*}
	with $\alpha=1/\tau+2\lambda/\beta$.

    \begin{enumerate}[label=(\roman*)]
		\item The optimization problem in \eqref{eq:sec} admits a unique finite saddle-point $(\beta^\star,\tau^\star)$ if and only if $\lambda >0$ or $\lambda=0$ and $\delta>1$. 
		\item When $\lambda=0$ and $\delta>1$, the saddle point $(\beta^\star,\tau^\star)$ of \eqref{eq:sec} is given by
		\begin{align}
			\tau^\star&=\arg\min_{\tau\geq 0} \left(\frac{\tau\delta}{2}+\frac{\rho}{2\tau}+\tilde{Y}(\tau)\right)_{+}^2,\\
			\beta^\star&= \left({\tau^\star\delta}+\frac{\rho}{\tau^\star}+2\tilde{Y}(\tau^\star)\right)_{+},
		\end{align}
	where
	\begin{equation*}
	\tilde{Y}(\tau):=\tau\left(\mathbb{E}\left[(H-\frac{\sqrt{P}}{\tau})^2{\bf 1}_{\{H\geq \frac{\sqrt{P}}{\tau}\}}\right]-\frac{1}{2}\right).
	\end{equation*}
	Moreover, $\overline{\phi}$ reduces to
	\begin{equation*}
	\overline{\phi}=\left(\frac{\tau\delta}{2}+\frac{\rho}{2\tau}+\tilde{Y}(\tau)\right)_{+}^2.
	\end{equation*}
		\item Let $\hat{\bf x}$ be the solution of \eqref{eq:nonlse}, and consider its associated empirical density function
		\begin{equation*}
		\hat{\mu}(\hat{\bf x}):=\frac{1}{n}\sum_{i=1}^n \boldsymbol{\delta}_{\hat{\bf x}_i}.
		\end{equation*}
		Further, let the function $\theta:\mathbb{R}\to [-\sqrt{P},\sqrt{P}]$,
		\begin{equation*}
		\theta(\gamma):=\left\{
		\begin{array}{ll}
			-\sqrt{P} &  \mathrm{if } \gamma\leq -\sqrt{P}\alpha^{\star}\\
			\frac{\gamma}{\alpha^{\star}} &  \mathrm{if } -\sqrt{P}\alpha^{\star}\leq \gamma \leq \sqrt{P}\alpha^\star\\
			\sqrt{P} & \mathrm{if }  \gamma \geq \sqrt{P}\alpha^\star
		\end{array},
		\right.
		\end{equation*}
		where $\alpha^\star=1/{\tau^\star}+2\lambda/\beta^\star$. Assume either $\lambda>0$ or $\lambda=0$ and $\delta>1$. 
		Then, under Assumption \ref{ass:regime} and Assumption \ref{ass:statistic}, for any pseudo-Lipschitz function $f$ of order $k$, it holds that
		\begin{equation*}
		\frac{1}{n}\sum_{i=1}^n f(\hat{\bf x}_i) \stackrel{P}{\rightarrow} \mathbb{E}_{H}\left[f(\theta(H))
		\right],
		\end{equation*}
		where $H\sim\mathcal{N}(0,1)$. Particularly, the empirical density function $\hat{\mu}(\hat{\bf x})$ converges in Wasserstein$-k$ distance  to $\theta(H)$. 
		\end{enumerate}
	\label{th:previous}
\end{theorem}
\begin{proof}
	See Section \ref{proof_th:previous}.
\end{proof}

\begin{theorem}[Distributional characterization of the distortion]
Consider the setting of Theorem \ref{th:previous}. 
 Let $f:\mathbb{R}^2\to \mathbb{R}$ be a pseudo-Lipschitz function of order $2$ and let $\alpha^\star=1/\tau^\star+2\lambda/\beta^\star$. Assume either $\lambda>0$ or $\lambda=0$ and $\delta>1$. 
Then, under Assumption \ref{ass:regime} and Assumption \ref{ass:statistic}, the following convergence holds true:
\begin{equation}
\begin{split}
&\frac{1}{m}\sum_{i=1}^m f(\left[\hat{\bf e}\right]_i,s_i) \\\stackrel{P}{\rightarrow} &\mathbb{E}_{H,S}\left[f\left(\frac{\beta^\star}{2}\frac{\sqrt{(\tau^\star)^2\delta-\rho}H-\sqrt{\rho} S}{\tau^\star\delta},S\right)\right],
\end{split}
\end{equation}
where $H$ is a standard normal scalar variable and $S$  a discrete binary variable taking $1$ and $-1$ with equal probabilities. 
 Equivalently, letting 
 \begin{equation*}
 \hat{\mu}({\bf e},{\bf s}):= \frac{1}{m}\sum_{i=1}^m \boldsymbol{\delta}_{([{\bf e}]_i,s_i)},
 \end{equation*}
 then $\hat{\mu}({\bf e},{\bf s})$ converges in Wasserstein$-2$ distance to the distribution of $(\frac{\beta^\star}{2}\frac{\sqrt{(\tau^\star)^2\delta-\rho}H-\sqrt{\rho} S}{\tau^\star\delta},S)$. 
\label{th:distortion}
\end{theorem}
\begin{proof}
	See Section \ref{proof_th:distortion}.
\end{proof}

\subsection{Characterizations of specifications and performance metrics}
As an application of Theorem \ref{th:previous} and Theorem \ref{th:distortion}, we derive closed-form approximations for the specifications and performance metrics defined in Section \ref{sec2}:
\begin{corollary}[Convergence of the average SINAD upper and lower bounds]
	Consider Assumption \ref{ass:regime} and \ref{ass:statistic}, then $\overline{\rm SINAD}_{\rm up}$ converges to: 
	\begin{equation}
	\begin{split}
	&\overline{\rm SINAD}_{\rm up}\to\overline{\rm SINAD}_{\rm up}^\star\\:=& \mathbb{E}_{H,S}\left[\frac{\rho}{\frac{(\beta^\star)^2}{4}\frac{\left((\sqrt{(\tau^\star)^2\delta-\rho})H-\sqrt{\rho}S\right)^2}{(\tau^\star\delta)^2}+\sigma^2}\right] ,
	\label{eq:upper_b}
	\end{split}
	\end{equation}
	and $\overline{\rm SINAD}_{lb}$ converges to:
	\begin{equation}
	\overline{\rm SINAD}_{\rm lb}\stackrel{P}{\rightarrow} \overline{\rm SINAD}_{\rm lb}^\star:=\frac{\rho}{\frac{(\beta^\star)^2}{4\delta}+\sigma^2}. \label{eq:lower_bound_l}
	\end{equation}
\label{cor:lbub}
\end{corollary}
\begin{proof}
	Function $x\mapsto \frac{\rho}{x^2+\sigma^2}$ is a Lipschitz function. Applying Theorem \ref{th:distortion} yields:
	\begin{equation}
	\begin{split}
	&\frac{1}{m}\sum_{k=1}^m\frac{\rho}{|e_k|^2+\sigma^2}\\\stackrel{P}{\rightarrow} &\mathbb{E}_{H,S}\left[\frac{\rho}{\frac{(\beta^\star)^2}{4}\frac{\left((\sqrt{(\tau^\star)^2\delta-\rho})H-\sqrt{\rho}S\right)^2}{(\tau^\star\delta)^2}+\sigma^2}\right].
	\end{split}
	\end{equation}
	Finally, since $x\mapsto  \frac{\rho}{x^2+\sigma^2}$ is bounded by $\frac{\rho}{\sigma^2}$, the convergence in \eqref{eq:upper_b} follows from the dominated convergence theorem. 
	To prove \eqref{eq:lower_bound_l}, we use the fact that $x\mapsto x^2$ is a pseudo-Lipschitz function of order $2$. Hence, we may again use Theorem \ref{th:distortion} to obtain 
	\begin{equation*}
	\frac{1}{m}\sum_{k=1}^{m}|e_k|^2 \stackrel{P}{\rightarrow} \frac{(\beta^\star)^2}{4\delta}.
	\end{equation*}
To prove the convergence in \eqref{eq:lower_bound_l}, it suffices to check that $\frac{1}{m}\sum_{k=1}^{m}|e_k|^2 $ is bounded. Indeed, if this is true then one can in a similar way as before use the dominated convergence theorem to prove the convergence of the expectation of $\frac{1}{m}\sum_{k=1}^{m}|e_k|^2$ to  its probability limit. Using the fact that $\hat{\bf x}$ minimizes the cost in \eqref{eq:nonlse}, the following inequality holds:
	\begin{equation*}
	\frac{1}{m}\|{\bf H}\hat{\bf x}-\sqrt{\rho}{\bf s}\|^2+\frac{\lambda}{m}\|\hat{\bf x}\|^2 \leq \frac{1}{m}\|\sqrt{\rho}{\bf s}\|^2
\end{equation*}
	and hence, 
	\begin{equation*}
	\frac{1}{m}\|{\bf H}\hat{\bf x}-\sqrt{\rho}{\bf s}\|^2\leq \rho.
	\end{equation*}
Recalling that $\frac{1}{m}\sum_{k=1}^m|e_k|^2= 	\frac{1}{m}\|{\bf H}\hat{\bf x}-\sqrt{\rho}{\bf s}\|^2$ we establish that $\frac{1}{m}\sum_{k=1}^m|e_k|^2$ is bounded. 
\end{proof}
\begin{corollary}[Convergence of the per-antenna power and the per-user distortion error power]
	Under the setting of Theorem \ref{th:previous}, the per-antenna and the per-user distortion error power satisfy the following convergences:
	\begin{equation}
	P_b\stackrel{P}{\rightarrow} P_b^\star:= \delta(\tau^\star)^2-\rho \label{eq:rev},
	\end{equation}
	and 
	\begin{equation}
	P_d \stackrel{P}{\rightarrow} P_d^\star:= \frac{(\beta^\star)^2}{4\delta}.
	\end{equation}
	\label{cor:dist_pb}
\end{corollary}
\begin{proof}
Note that $P_d=\frac{1}{m}\sum_{k=1}^m|e_k|^2$. The convergence of $P_d$ to its probability limit has been established in the proof of Corollary \ref{cor:lbub}. The convergence of $P_b$ to the limit in \eqref{eq:rev} follows directly by applying Theorem \ref{th:previous} along with the first-order optimality condition for the variable $\tau$.
\end{proof}
Corollary \ref{cor:dist_pb} allows us to provide an interpretation of the parameters $\tau^\star$ and $\beta^\star$. From the convergences stated in this Corollary, it appears that $\tau^\star$ is related to how much power is devoted to the precoded vector $\hat{\bf x}$, while $\beta^\star$ allows for quantifying the amount of distortion experienced by the PAPR precoder. The control factor $\rho$ can always be adjusted to fix the power $P_b^\star$ to a given feasible value. However, this would lead to varying the coefficient $\beta^\star$ which determines the distortion level. More details on the role of the control factor $\rho$ on the performance will be given in this section and in section \ref{sec4}. 
\begin{corollary}[Convergence of the bit error probability]
Under the setting of Theorem \ref{th:previous},	the bit error probability defined in \eqref{def_bep} converges to
\begin{equation*}
P_e{\to} P_e^\star:=Q\left(\frac{\sqrt{\rho}-\frac{\beta^\star \sqrt{\rho}}{2\tau^\star \delta}}{\sqrt{\frac{(\beta^\star)^2}{4}\frac{(\tau^\star)^2\delta-\rho}{(\tau^\star)^2\delta^2}+\sigma^2}}\right).
\end{equation*}
\end{corollary}
\begin{proof}
 The symbol $s_k$ is decoded erroneously if 
\begin{equation*}
{\bf h}_k^{T}\hat{\bf x}-\sqrt{\rho} +z_k \leq -\sqrt{\rho}
\end{equation*}
when $s_k=1$ and
\begin{equation*}
{\bf h}_k^{T}\hat{\bf x}+\sqrt{\rho} +z_k \geq \sqrt{\rho}
\end{equation*}
when $s_k=-1$. So 
\begin{equation}
\begin{split}
	P_e&= \frac{1}{2}\mathbb{P}\left[{\bf h}_k^{T}\hat{\bf x}-\sqrt{\rho}s_k+z_k\leq -\sqrt{\rho}\ |\ s_k=1\right]\\
	&+ \frac{1}{2}\mathbb{P}\left[{\bf h}_k^{T}\hat{\bf x}-\sqrt{\rho}s_k+z_k\geq \sqrt{\rho}|s_k=-1\right].
	\end{split}
	\end{equation}
Using Theorem \ref{th:distortion} along with the Portemanteau Lemma \cite{van-der-vaart}, we prove that
\begin{align}
&P_e\to \nonumber\\&\frac{1}{2}\mathbb{P}\left[\frac{\beta^\star}{2}\frac{\sqrt{(\tau^\star)^2\delta-\rho}H-\sqrt{\rho}S}{\tau^\star\delta}+\sigma Z \leq -\sqrt{\rho} | S=1\right]\nonumber\\
+&\frac{1}{2}\mathbb{P}\left[\frac{\beta^\star}{2}\frac{\sqrt{(\tau^\star)^2\delta-\rho}H-\sqrt{\rho}S}{\tau^\star\delta}+\sigma Z \geq \sqrt{\rho}|S=-1\right]\label{eq:ll}\\
=&Q\left(\frac{\sqrt{\rho}-\frac{\beta^\star \sqrt{\rho}}{2\tau^\star \delta}}{\sqrt{\frac{(\beta^\star)^2}{4}\frac{(\tau^\star)^2\delta-\rho}{(\tau^\star)^2\delta^2}+\sigma^2}}\right).\label{eq:ll2}
\end{align}
where in \eqref{eq:ll} $Z$ follows a standard normal distribution with mean zero and variance $1$. 
\end{proof}

It is important to note that although a BPSK modulation is assumed, \eqref{eq:ll2} is different from the asymptotic bit error probability ${P_e}=Q(\sqrt{2{\rm SNR}})$. The reason lies in the fact that the latter relation holds in the case of additive Gaussian noise that is independent of the transmitted symbols. In our case, we have not only noise but also the distortion $\hat{\bf e}$ which is correlated with the transmitted symbols, as evidenced by Theorem \ref{th:distortion}.

\subsection{Special cases: RZF and ZF precoding ($P\to\infty$)}
 The analysis of the RZF and ZF precoding in multi-user downlink systems has been the focus of several studies in the literature. Among these studies, we cite  the work in  \cite{wagner} which considered this problem with sophisticated channel models involving different correlations across users.
 However, to the best of our knowledge, none of the existing works studied the bit error probability approximation (all the focus being on the asymptotic characterization of the SINAD). 
  In the sequel, we show that by taking $P\to\infty$ in the asymptotic expressions of Theorem \ref{th:previous}, we can simplify the expression \eqref{eq:lower_bound_l} to reach the same results for the asymptotic SINAD performance obtained in the literature. Additionally, we obtain new asymptotic approximations for the bit error probability. For the sake of scientific rigor, since our proofs in Theorem \ref{th:previous} and Theorem \ref{th:distortion} relied on the assumption of finite values of $P$, we do not claim the convergence in probability of the specifications and performance metrics to the  limits of their asymptotic equivalents when $P\to\infty$, although we believe this to be the case. A rigorous proof of the convergence would require us to re-consider the case where $P\to\infty$ separately. However, we do not provide such a proof since
  the analysis of the RZF or the ZF can be conducted using tools from random matrix theory~\cite{abla_rmt} and is thus less worthy of consideration. 

\begin{theorem}[$P\to\infty$ and $\lambda>0$]
	For a given value of $P$, denote by $\tau^\star(P)$ and $\beta^\star(P)$ the solutions to the max-min problem in \eqref{eq:sec}. Assume $\lambda>0$, then as $P\to\infty$, the following convergences hold true:
	\begin{align}
			\lim_{P\to\infty} \tau^\star(P)&=\frac{\sqrt{\rho}}{\sqrt{\delta-\frac{1}{(1+\lambda s^\star)^2}}},\\
		\lim_{P\to\infty} \beta^\star(P)&=\frac{2\sqrt{\rho}}{s^\star\sqrt{\delta-\frac{1}{(1+\lambda s^\star)^2}}},
	\end{align}
where $s^\star$ is given by:
\begin{equation}
s^\star=\frac{\sqrt{(\delta-\lambda -1)^2+4\delta\lambda}-\delta+\lambda+1}{2\delta \lambda}.
\label{s_star}
\end{equation}
Particularly, in this regime, the asymptotic values of the per-antenna power, the distortion power, the SINAD lower bound $\overline{\rm SINAD}_{{\rm lb}}^{\star}$, and the bit error probability converge to
\begin{align}
	\lim_{P\to\infty}P_b^\star&=\frac{\rho}{
	\delta(1+\lambda s^\star)^2-1},\label{rzf_pb}\\
	\lim_{P\to\infty}P_d^\star&=\frac{\rho}{(s^\star)^2\delta(
\delta-\frac{1}{(1+\lambda s^\star)^2})},\label{rzf_pd}\\
	\lim_{P\to\infty}\overline{\rm SINAD}_{{\rm lb}}^\star&=\frac{(s^\star)^2\delta(\delta-\frac{1}{(1+\lambda s^\star)^2})}{1+\frac{\sigma^2}{\rho}(s^\star)^2\delta(\delta-\frac{1}{(1+\lambda s^\star)^2})},\label{rzf_snr}\\
	\lim_{P\to\infty}P_e^\star&=Q\left(\frac{\sqrt{\rho}(\delta s^\star-1)}{\sqrt{\sigma^2\delta^2(s^\star)^2+\frac{\rho}{(1+\lambda s^\star)^2-1}}}\right).\label{rzf_pe}
\end{align}
\label{Theo_rzf}
\end{theorem}
\begin{proof}
See Section \ref{Proof_rzf}.
\end{proof}

\begin{theorem}[$P\to\infty$, $\lambda=0$ and $\delta>1$]
	For a given value of $P$, denote by $\tau^\star(P)$ and $\beta^\star(P)$ the solutions to the max-min problem in \eqref{eq:sec}. Assume $\lambda=0$ and $\delta>1$. Then as $P\to\infty$, 
		\begin{align}
				\lim_{P\to\infty} \tau^\star(P)&=\sqrt{\frac{\rho}{(\delta-1)}},\label{zft}\\
		\lim_{P\to\infty} \beta^\star(P)&=2\sqrt{\rho(\delta-1)}.\label{zfb}
		\end{align}
Particularly, in this regime, the asymptotic values of the per-antenna power, the distortion power, the SINAD lower bound $\overline{\rm SINAD}_{{\rm lb}}^{\star}$, and the bit error probability converge to
\begin{align}
	\lim_{P\to\infty}P_b^\star&= \frac{\rho}{\delta-1},\label{eq:zf_pb}\\
	\lim_{P\to\infty}P_d^\star&=\rho(1-\frac{1}{\delta}),\\
		\lim_{P\to\infty}\overline{\rm SINAD}_{{\rm lb}}^\star&=\frac{\rho}{\rho(1-\frac{1}{\delta})+\sigma^2},\\
	\lim_{P\to\infty}P_e^\star&=Q\left(\frac{\sqrt{\rho}}{\sqrt{\rho(\delta-1)+\sigma^2\delta^2}}\right). \label{eq:Pe_zf}
\end{align}
\end{theorem} 
 \begin{proof}
 	By setting $\lambda=0$ in equation \eqref{1st_condition} of the proof of Theorem \ref{Theo_rzf}, we directly obtain \eqref{zft} and \eqref{zfb}, from which the approximations in \eqref{eq:zf_pb}-\eqref{eq:Pe_zf} follow easily. 
\end{proof}

\subsection{Limiting cases}
The expressions derived so far are useful to characterize the performance of the  limited PAPR precoder in terms of the design parameters, that is the ratio of $m$ to $n$, and the power control parameter $\rho$. To gain more insight into the impact of these parameters on  the performance of the limited PAPR precoder, next we study the following limiting cases. 
\begin{theorem}[The number of users much smaller than the number of antennas $(m\ll n)$]
	For a given value of $\delta$, denote by $\tau^\star(\delta)$ and $\beta^\star(\delta)$ the solutions to the max-min problem in \eqref{eq:sec}. Assume $\lambda>0$.  Then as $\delta\to 0$, the following convergences hold true:
	\begin{align}
		\lim_{\delta\to 0}\frac{\tau^\star(\delta)}{\sqrt{\frac{\rho}{\delta}}}&= 1\label{tau_delta_zero},\\
		\lim_{\delta\to 0}\frac{\beta^\star(\delta)(1+\frac{1}{\lambda})}{2\sqrt{\rho\delta}}& =1.  \label{beta_delta_zero}
		\end{align}
	Particularly, in this regime,  the asymptotic values for the per-antenna power, the distortion power, the SINAD lower bound $\overline{\rm SINAD}_{\rm lb}^{\star}$, and the bit error probability converge to 
\begin{align}
	\lim_{\delta \to 0}P_b^\star&=0 \label{eq:P_b_delta_zero},\\
	\lim_{\delta\to 0}P_d^\star &= \frac{\rho}{(1+\frac{1}{\lambda})^2},\\
	\lim_{\delta\to 0}\overline{\rm SINAD}_{\rm lb}^\star &= \frac{\rho}{\frac{\rho}{(1+\frac{1}{\lambda})^2}+\sigma^2},  \\
	\lim_{\delta\to 0}P_e^\star&= Q(\frac{\sqrt{\rho}}{(\lambda+1)\sigma}).
	\label{eq:SINR_delta_zero}
\end{align}
\label{th:as_1}
\end{theorem}
\begin{proof}
	See Section \ref{ass:as_1}.
\end{proof}
 
\begin{theorem}[The number of antennas much smaller than the the number of users $(m\gg n)$]
For a given value of $\delta$, denote by $\tau^\star(\delta)$ and $\beta^\star(\delta)$ the solutions to the max-min problem in \eqref{eq:sec}. Assume $\lambda\geq0$.  Then as $\delta\to \infty$, the following convergences hold true:
\begin{align}
	&\lim_{\delta\to \infty}\frac{\tau^\star(\delta)}{\sqrt{\frac{\rho}{\delta}}}= 1 \label{tau_delta_inf},\\
	& \lim_{\delta\to\infty}\frac{\beta^\star(\delta)}{2\sqrt{\rho\delta}}= 1. \label{beta_delta_inf}
\end{align}
Particularly, in this regime,  the asymptotic values for the per-antenna power, the distortion power, the SINAD lower bound $\overline{\rm SINAD}_{\rm lb}^{\star}$, and the bit error probability converge to
\begin{align}
	\lim_{\delta\to\infty} P_b^\star&=0 \label{P_b_delta_inf},\\
	\lim_{\delta\to \infty}P_d^\star &= \rho \label{P_d_delta_inf},\\
		\lim_{\delta\to \infty}\overline{\rm SINAD}_{\rm lb}^\star &= \frac{\rho}{\rho+\sigma^2} \label{SINR_delta_inf},\\
	\lim_{\delta\to \infty}P_e^\star&=\frac{1}{2}. \label{P_e_delta_inf} 
\end{align}
\label{th:as_2}
\end{theorem}
\begin{proof}
	See Section \ref{ass:as_2}.
\end{proof}

\begin{figure*} 
	\centering
	\includegraphics[width=6 in]{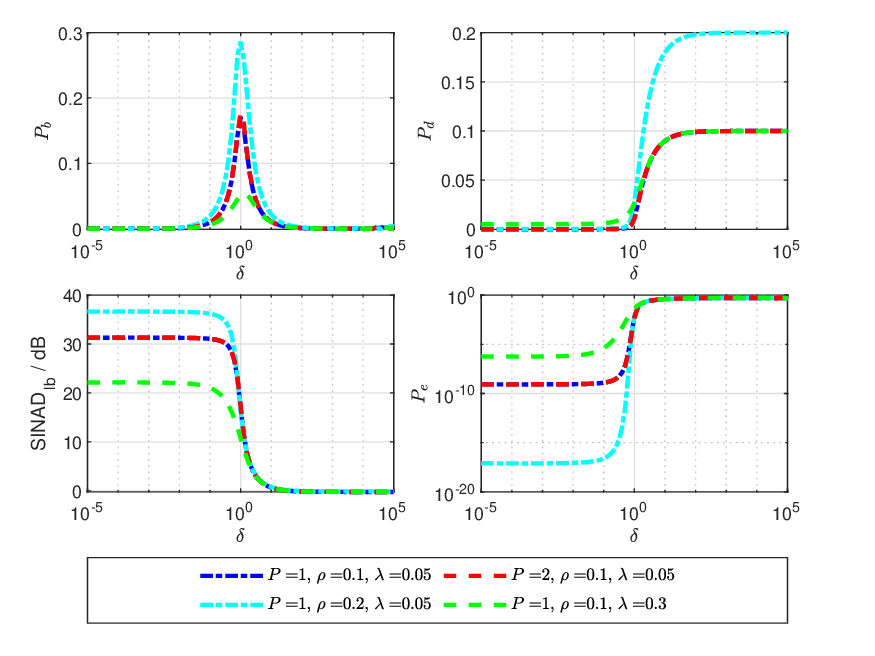}
	\caption{The theoretical  per-antenna power, distortion power, SINAD lower bound and bit error probability versus $\delta$, for $\sigma=0.05$ and different parameter combinations.}
	\label{fig4}
\end{figure*}

Theorem \ref{th:as_1} and Theorem \ref{th:as_2} allow us to understand the behavior of the limited PAPR precoder when the number of available antennas largely exceeds the number of users or vice versa. As an important  remark, we note that, interestingly, in both cases, all performance metrics do not asymptotically depend on $P$. In other words, considering the regimes where $\delta\to0$, or $\delta\to\infty$,  regardless of the maximum power at each antenna, the performance is almost the same. However, the results depend on $\lambda$ when $\delta\to 0$ and do not depend on $\lambda$ when $\delta\to\infty$. This behavior can be attributed to the fact that the limited PAPR  precoder becomes close to the RZF precoder when $\delta\to 0$ and to the ZF precoder when $\delta\to\infty$. Below, we provide arguments supporting these claims. 

\noindent{\bf Case 1 ($\delta\to 0$).} In this case, the  limited PAPR precoder solving \eqref{eq:nonlse} becomes close to the RZF precoder if the latter satisfies the per-antenna constraint. It turns out that when $\delta\to 0$, the Frobenius norm of ${\bf H}^{T}$ scales as $\sqrt{m}$ and so does the  per antenna power of the RZF \footnote{Here we used the fact that $\|\hat{\bf x}_{\rm RZF}\|\leq \frac{1}{\lambda}\|{\bf H}^{T}{\bf s}\|$ and that $\|{\bf H}^{T}{\bf s}\|$ can be approximated by $\sqrt{{\rm tr} ({\bf H}^{T}{\bf H})}$ with high probability when $m\to\infty$}, which leads to the RZF becoming asymptotically feasible with respect to the optimization problem \eqref{eq:nonlse}.   However, in this case, it is advisable to select a small regularization parameter. The reason lies in that a sufficient number of degrees of freedom are available to find ${\bf x}$ such that ${\|{\bf Hx}-\sqrt{\rho}{\bf s}\|^2}$ is as small as desired. Using a large regularization parameter will thus increase the bias, thereby  deteriorating the performance.

\noindent{\bf Case 2 ($\delta\to \infty$).}  In this case, we argue that the limited PAPR precoder solving \eqref{eq:nonlse} becomes close  to the ZF precoder.  This explains why the performances do not depend on either $P$ or $\lambda$. 

Towards this goal, it will suffice to show that
\begin{enumerate}[label=(\roman*)]
\item The per-antenna power of the ZF precoder converges  to zero as $\delta\to\infty$, and thus the ZF precoder becomes  feasible with respect to the optimization problem in \eqref{eq:nonlse}. 
\item  Denoting the optimal cost in \eqref{eq:nonlse} by $\hat{G}$, then as $\delta\to \infty$,
\begin{equation}
\frac{1}{m}\hat{G}-\frac{1}{m}\| {\bf H}\hat{\bf x}_{\rm ZF}-\sqrt{\rho}{\bf s}\|^2\stackrel{P}{\rightarrow} 0.\label{eq:G}
\end{equation} 
\end{enumerate}
\noindent To prove (i), we use  concentration results of random matrices \cite{wainwright_2019} to show that as $\delta\to\infty$ the spectral norm $\|\frac{n}{m}({\bf H}^{T}{\bf H})-{\bf I}_p\|$ converges in probability to zero.  Hence, the precoding vector $\hat{\bf x}_{\rm ZF}$ can be approximated by $\tilde{\bf x}_{\rm ZF}:=\frac{n}{m}\sqrt{\rho}{\bf H}^{T}{\bf s}$ in the sense that $\|\hat{\bf x}_{\rm ZF}-\tilde{\bf x}_{\rm ZF}\|$ converges to zero in probability. Denote by $\tilde{\bf h}_i$ the $i$-th row of ${\bf H}^{T}$. Then, the $i$-th entry of $\tilde{\bf x}_{\rm ZF}$ can be bounded as
\begin{equation*}
\left|\left[\tilde{\bf x}_{\rm ZF}\right]_i\right|\leq \sqrt{\rho}\sqrt{\frac{n}{m}} \frac{\sqrt{n}\tilde{\bf h}_i^{T}{\bf s}}{\sqrt{m}}.
\end{equation*}
Since $\frac{\sqrt{n}\max_{1\leq i\leq n}|\tilde{\bf h}_i^{T}{\bf s}|}{\sqrt{m}}$ is stochastically bounded, we conclude that with probability approaching $1$, 
all entries of $\hat{\bf x}_{\rm ZF}$ are less than any fixed positive threshold $\epsilon$. Letting $\epsilon <\sqrt{P}$ establishes the first statement (\romannum{1}). 

\noindent To validate (\romannum{2}), we use the following pair of inequalities 
\begin{equation}
\frac{\| {\bf H}\hat{\bf x}_{\rm ZF}-\sqrt{\rho}{\bf s}\|^2}{m} \leq \frac{\hat{G}}{m} \leq \min_{\substack{-\sqrt{P}\leq x_i \leq \sqrt{P}\\ i=1,\cdots,n}} \frac{\| {\bf Hx}-\sqrt{\rho}{\bf s}\|^2}{m} +\frac{\lambda Pn}{m}. \label{eq:rfe}
\end{equation}
From (i), the  absolute value of the entries of the ZF precoder are less than ${\sqrt{P}}$. Hence, with probability approaching $1$, 
\begin{equation*}
\min_{-\sqrt{P}\leq x_i \leq \sqrt{P}} \frac{\| {\bf Hx}-\sqrt{\rho}{\bf s}\|^2}{m}=\frac{1}{m} \| {\bf H}\hat{\bf x}_{\rm ZF}-\sqrt{\rho}{\bf s}\|^2.
\end{equation*}
To continue, we substitute $\hat{\bf x}_{\rm ZF}$ by its expression into $\left\|{\bf H}\hat{\bf x}_{\rm ZF}-\sqrt{\rho}{\bf s}\right\|^2$, then we apply the quadratic forms convergence results \cite{bai2009spectral} to conclude that
\begin{equation*}
\frac{1}{m}\|{\bf H}\hat{\bf x}_{\rm ZF}-\sqrt{\rho}{\bf s}\|^2-\rho\frac{m-n}{m}\stackrel{P}{\rightarrow} 0. 
\end{equation*}
Plugging the above convergence into \eqref{eq:rfe}, we can deduce that as $\delta\to \infty$, the left-hand side and the right-hand side of \eqref{eq:rfe} are asymptotically equivalent to $\frac{1}{m}\|{\bf H}\hat{\bf x}_{\rm ZF}-\sqrt{\rho}{\bf s}\|^2$, which proves the desired convergence in \eqref{eq:G}.

\noindent{\bf Numerical illustration.} Figure \ref{fig4} plots the theoretical values for all the studied performance metrics versus $\delta$.  As expected from Theorem \ref{th:as_1} and Theorem \ref{th:as_2}, the per-antenna power goes to zero as $\delta$ tends to zero or $\delta$ tends to infinity. Moreover, the performance in terms of bit error probability and SINAD becomes the best when $\delta$ is close to zero due to the excess in the number of spatial degrees of freedom.

\begin{figure}  
	\centering
	\includegraphics[width=3 in]{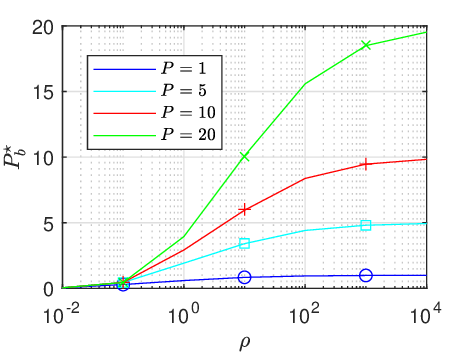}
	\caption{$P_b^\star$ versus $\rho$, for different $P$ values. $n=256$, $\delta=1$, $\lambda=0.01$ and $\sigma=0.05$.  (Markers show the simulated results
averaged over 50 realizations of random quantities $\mathbf{H}$, $\mathbf{s}$ and $\mathbf{z}$.)}
	\label{fig1}
\end{figure} 

\begin{theorem}[Power control parameter $\rho\to 0$] For a given value of $\rho$, denote by $\tau^\star(\rho)$ and $\beta^\star(\rho)$ the solutions to the max-min problem in \eqref{eq:sec}. 
	Then the following statements hold true:
	\begin{enumerate}
	\item Assume $\lambda>0$. Then, as $\rho\to 0$, the following convergences hold
	\begin{align}
&	\lim_{\rho\to 0}\frac{\tau^\star(\rho)}{\frac{\sqrt{\rho}}{\sqrt{\delta-\frac{1}{(1+\lambda s^\star)^2}}}}=1, \label{eq:tau_rho}\\
&  \lim_{\rho\to 0}\frac{\beta^\star(\rho)}{\frac{2\sqrt{\rho}}{s^\star\sqrt{\delta-\frac{1}{(1+\lambda s^\star)^2}}}}=1, \label{eq:beta_rho}
	\end{align}
	where $s^\star$ is the same as in \eqref{s_star}.
	Particularly, in this regime,  the asymptotic values for the per-antenna power, the distortion power, the SINAD lower bound $\overline{\rm SINAD}_{\rm lb}^{\star}$, and the bit error probability converge to:
	\begin{align}
	&\lim_{\rho\to 0}\frac{P_b^\star}{\frac{\rho}{
	\delta(1+\lambda s^\star)^2-1}}=1,\label{r0_pb}\\
	&\lim_{\rho\to 0}\frac{P_d^\star}{\frac{\rho}{(s^\star)^2\delta(
\delta-\frac{1}{(1+\lambda s^\star)^2})}}=1,\label{r0_pd}\\
	&\lim_{\rho\to 0}\frac{\overline{\rm SINAD}_{{\rm lb}}^\star}{\frac{\rho}{\sigma^2}}=1,
	\label{r0_snr}\\
	&\lim_{\rho\to 0}\frac{P_e^\star}{\frac{1}{2}-\frac{1}{\sqrt{2\pi}}\frac{\sqrt{\rho}(\delta s^\star-1)}{\sqrt{\sigma^2\delta^2(s^\star)^2+\frac{\rho}{(1+\lambda s^\star)^2-1}}}}=1.
\end{align}
\item Assume $\lambda=0$ and $\delta >1$.  Then, as $\rho\to 0$, the following convergences hold:\begin{align}
	&	\lim_{\rho\to 0}\frac{\tau^\star(\rho)}{\frac{\sqrt{\rho}}{\sqrt{\delta-1}}}=1, \label{eq:tau_rho1}\\
	&  \lim_{\rho\to 0}\frac{\beta^\star(\rho)}{2\sqrt{\rho}\sqrt{\delta-1}}=1. \label{eq:beta_rho1}
\end{align}
Particularly, in this regime,  the asymptotic values for the per-antenna power, the distortion power, the SINAD lower bound $\overline{\rm SINAD}_{\rm lb}^{\star}$, and the bit error probability converge to:
\begin{align}
&	\lim_{\rho\to 0}\frac{P_b^\star}{\frac{\rho}{\delta-1}}=1,\\
&\lim_{\rho\to 0}\frac{P_d^\star}{\frac{\rho(\delta-1)}{\delta}}=1,\\
&\lim_{\rho\to 0} \frac{\overline{\rm SINAD}_{{\rm lb}}^\star}{\frac{\rho}{\sigma^2}}=1,\\
& \lim_{\rho\to 0}\frac{P_e^\star}{\frac{1}{2}-\frac{1}{\sqrt{2\pi}}\frac{\sqrt{\rho}}{\sqrt{\rho(\delta-1)+\sigma^2\delta^2}}}=1.
\end{align}
\end{enumerate}
\label{th:power_control_zero}
\end{theorem}
\begin{proof}
See Section \ref{Proof_rho0}.
\end{proof}

\begin{figure*}
	\centering
	\includegraphics[width=6 in]{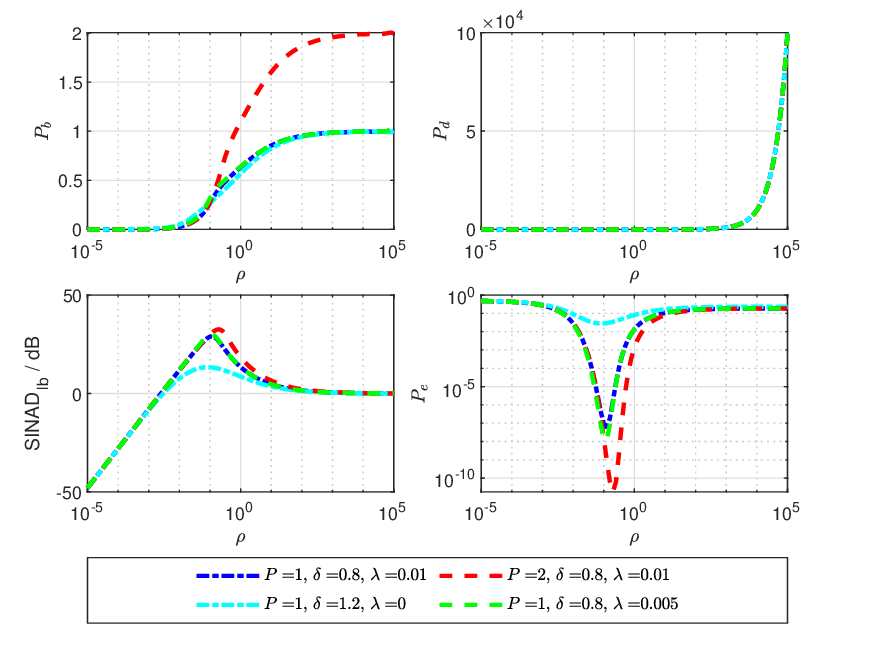}
	\caption{The per-antenna power, distortion power, SINAD lower bound and bit error probability versus $\rho$ for $\sigma=0.05$ and different  parameter combinations.}
	\label{fig7}
\end{figure*}

\begin{theorem}[Power control parameter $\rho\to\infty$]
For a given value of $\rho$, denote by $\tau^\star(\rho)$ and $\beta^\star(\rho)$ the solutions to the max-min problem in \eqref{eq:sec}. Then, as $\rho\to \infty$, the following approximations hold true:
\begin{align}
\tau^\star(\rho)&=\sqrt{\frac{\rho}{\delta}}+\frac{P}{2\sqrt{\delta\rho}}+O\left(\frac{1}{\rho}\right),\\
\beta^\star(\rho)&=2\sqrt{\rho\delta}-2\frac{\sqrt{2P}}{\sqrt{\pi}}+O\left(\frac{1}{\sqrt{\rho}}\right).
\end{align}
Moreover, in this regime, the asymptotic values for the per-antenna power, the distortion power, the SINAD lower bound $\overline{\rm SINAD}_{{\rm lb}}^{\star}$, and the bit error probability can be approximated as:
\begin{align}
	P_b^\star&= P + O\left(\frac{1}{\sqrt{\rho}}\right),\label{eq:P_b_rho_inf}\\
 P_d^\star&= \rho-2\sqrt{\frac{2P\rho}{\pi\delta}}+O(1), \\ 
 \overline{\rm SINAD}_{\rm lb}^\star&=1+2\sqrt{\frac{2P}{\pi\delta\rho}}+O\left(\frac{1}{\rho}\right),\\P_e^\star &= Q\left(\sqrt{\frac{2P}{\pi\delta(P+\sigma^2)}}\right) +O\left(\frac{1}{\sqrt{\rho}}\right).
  \label{eq:SINR_rho_inf}
\end{align}
\label{th:limit_rho}
\end{theorem}
\begin{proof}
	See Section \ref{app:limit_rho}.
\end{proof}

Theorem \ref{th:power_control_zero} and Theorem \ref{th:limit_rho} allow us to shed light on the behavior of the limited PAPR  precoder when the control parameter $\rho$ goes to either zero or infinity. As an interesting remark, we note that in the case where $\rho\to 0$, the performance becomes independent of $P$ but dependent on the regularization parameter $\lambda$. In this case, we claim that the limited PAPR  precoder becomes close to the RZF. This is because as $\rho\to 0$, the entries of the RZF precoder tend to zero and thus the RZF becomes feasible with respect to the optimization problem in \eqref{eq:nonlse}.  On the other hand, when $\rho\to \infty$, the performance depends on $P$ but not on the regularization parameter. To explain such  behavior, we argue that, in this case, the limited PAPR  precoder becomes close to the non-linear least squares (LS) precoder given by:
\begin{equation}
\hat{\bf x}_{{\rm LS}}=\arg\min_{-\sqrt{P}\leq x_i \leq \sqrt{P}}\|{\bf Hx}-\sqrt{\rho}{\bf s}\|^2. \label{eq:LS}
\end{equation}
To see this, it suffices to note that for all ${\bf x}$, the term $ \|{\bf Hx}-\sqrt{\rho}{\bf s}\|^2$ becomes dominant in the minimization of \eqref{eq:nonlse} since for all feasible ${\bf x}$,
\begin{equation*}
 \|{\bf Hx}-\sqrt{\rho}{\bf s}\|^2\!\geq\! \|{\bf H}\hat{\bf x}_{\rm ZF}-\sqrt{\rho}{\bf s}\|^2\!=\!\rho\left\|({\bf I}-{\bf H}({\bf H}^{T}{\bf H})^{-1}{\bf H}^{T}){\bf s}\right\|^2
\end{equation*}
while the second term $\lambda\|{\bf x}\|^2$ remains bounded by $Pn$ as $\rho$ grows large.

\begin{figure*} 
	\centering
	\includegraphics[width=6 in]{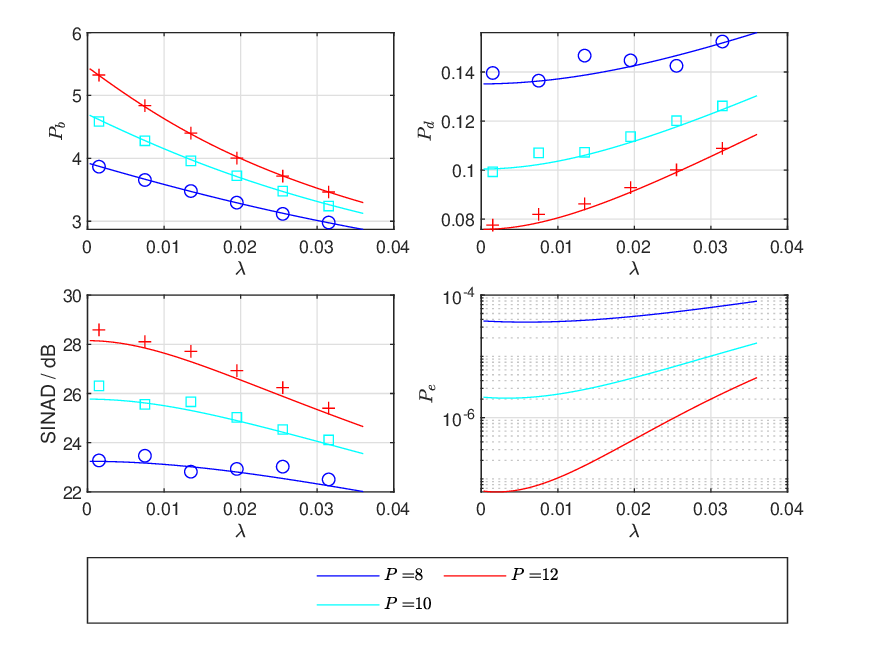}
	\caption{Impact of the regularization parameter $\lambda$ for different $P$ values. $n=512$, $\delta=0.84$, $\rho=2$ and $\sigma=0.05$. (Markers show the simulated results
averaged over 50 realizations of random quantities $\mathbf{H}$, $\mathbf{s}$ and $\mathbf{z}$. We do not show simulated results for $P_e$ because they require too many runs to simulate such a low $P_e$. )}
	\label{fig2}
\end{figure*}

By combining the observations in both regimes ($\rho\to 0$ and $\rho\to\infty)$, we get a more precise idea of the role of the control parameter $\rho$ on the per-antenna power $P_b^\star$.
 Setting $\rho$ to small values makes the per-antenna power close to zero  while using  large values for $\rho$ leads the precoder to use the maximum allowed power at each antenna. Such behavior is illustrated in Figure \ref{fig1}, which plots $P_b^\star$ against $\rho$ for several values of $P$. As can be seen, by varying $\rho$, the  per-antenna  power varies accordingly, becoming small for small $\rho$ values and close to the maximum allowed power for very large $\rho$ values. Note that, unlike RZF and ZF, the value of $\rho$ that achieves a fixed asymptotic per-antenna power $P_b^\star$ can not be determined in an explicit form. In this respect, when it comes to comparing precoders, it is necessary to  require the same $P_b^\star$ value. This can be done for each precoder by using the value of $\rho$ that achieves the target $P_b^\star$. A plot like the one in Figure \ref{fig1} can be used to determine numerically the corresponding values of the power control parameter.

  By expanding  $\|{\bf Hx}-\sqrt{\rho}{\bf s}\|^2$ in \eqref{eq:LS} and neglecting the quantities independent of $\rho$ or ${\bf x}$, we may claim that for large $\rho$ values,  the precoder in \eqref{eq:LS} would be  close to the one bit-precoding $\tilde{\bf x}:=\sqrt{P}{\rm sign}({\bf H}^{T}{\bf s})$. Such a finding, although making sense, calls into question the main interest behind solving the optimization problem in \eqref{eq:nonlse} to obtain the limited PAPR precoder. If for $\rho\to\infty$, its behavior would be equivalent to the precoder $\tilde{\bf x}$ which uses the maximum allowed power, one can rightly think that it should be less complex and more efficient to use $\tilde{\bf x}$ rather than solving the involved problem in \eqref{eq:nonlse}. Such a conclusion would be correct if more power necessarily implies better performance. As evidenced later in the simulation section, it is possible for a precoder using a lower per-antenna power to perform better than the one-bit precoding scheme given by $\tilde{\bf x}$ (See  Figure \ref{fig8} and Figure \ref{fig9}). 
  
\noindent{\bf Numerical illustration.} Figure \ref{fig7} plots the theoretical values for all studied specifications and performance metrics against $\rho$.  In agreement with the results of Figure \ref{fig1}, the per-antenna power is an increasing function of $\rho$, approaching $P$ when $\rho$ tends to infinity. However, when $\rho$ becomes very large, the distortion power increases, resulting in the saturation of the SINAD and the bit error probability. Interestingly, there is an optimal finite $\rho$, and hence an optimal $P_b^\star$ for which the performances in terms of SINAD and bit error probability are maximized. 

\section{Numerical simulations}
\label{sec4}
In this section, we numerically investigate the performance of the limited PAPR  precoders under different settings. We study the following specifications and performance metrics: the average per-user SINAD defined in \eqref{eq:SINR_average}, the average per-antenna power, the average per-user distortion  power, and the bit error probability. We compare the results with the theoretical predictions derived in Section \ref{sec3}. 
In all the figures below,  solid lines  represent the theoretical predictions, while markers  show the simulated results averaged over $50$ realizations of random quantities $\mathbf{H}$, $\mathbf{s}$ and $\mathbf{z}$. 

\subsection{Impact of the regularization parameter $\lambda$}

Figure \ref{fig2} illustrates the behavior of various specifications and performance metrics with a varying regularization parameter for $P=8,10$ and $12$. We note that the per-antenna power decreases with $\lambda$, since a large $\lambda$ value would penalize the term $\|{\bf x}\|^2$ in \eqref{eq:nonlse} more. However, as it can be seen from the plot of the bit error probability, setting $\lambda$ to smaller values does not  always translate into better performance. Indeed, a non-zero optimal value of $\lambda$ exists, which is small for a large $P$, but becomes larger as $P$ decreases. This shows that regularization is more important when $P$ is small to compensate for the bias caused by restricting the per-antenna power of the precoded vector. In a second experiment, we investigate whether this behavior still holds when all precoders have the same average per-antenna power. For that, we fix $P$ at $50$, and then use our asymptotic results to tune the power control factor $\rho$ to the value achieving the target $P_b^\star$. Figure \ref{fig6} illustrates the obtained performances in terms of  distortion power, the average per-user SINAD, and the bit error probability. Similar to Figure \ref{fig2}, we note that the optimal regularization parameter for low $P_b^\star$ is larger than that for higher $P_b^\star$.

\begin{figure} 
	\centering
	\includegraphics[width=3 in]{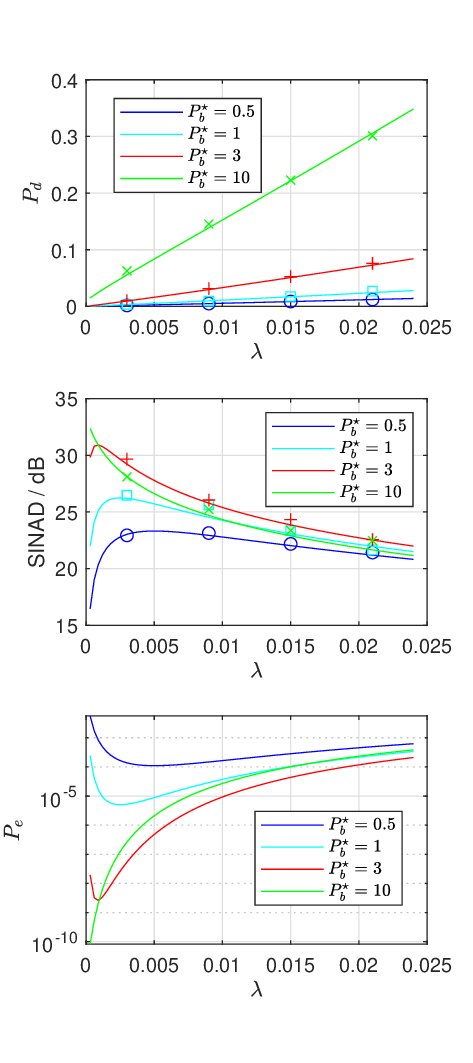}
	\caption{Impact of the regularization parameter $\lambda$ for different $P_b^\star$ values. $P=50$, $n=256$, $\delta=1$ and $\sigma=0.05$. Herein, the value of $\rho$ is tuned for each setting to achieve the target $P_b^\star$. We do not show simulated results for $P_e$ because they require too many runs to simulate such a low $P_e$.}
	\label{fig6}
\end{figure}

\subsection{Impact of the number of users to the number of antennas ratio ($\delta$)}

In Figure \ref{fig5}, we investigate the impact of the number of users to the number of antennas ratio $\delta$ on the performance of the limited PAPR precoder. As in Figure \ref{fig6}, for each plot, we leverage our asymptotic analysis to set the power control parameter $\rho$ at the value  ensuring the target asymptotic per-antenna power $P_b^\star$. As expected, we note that the power distortion increases with $\delta$. This is because a higher $\delta$ translates into serving more users and thus causes higher distortion error levels. However, it is curious to note that the distortion error reaches very high levels  as $P_b^\star$ becomes of an order of magnitude of $P$. 
To explain this, we refer to the findings of Theorem \ref{th:limit_rho} and Figure \ref{fig1}, which suggest that a higher value of $\rho$ is required to reach higher values of $P_b^\star$. But, when $\rho$ is large, the distortion error automatically increases as it becomes difficult to approximate $\sqrt{\rho}{\bf s}$ by ${\bf Hx}$ when ${\bf x}$ is constrained to a compact set. An important consequence of this behavior is that the SINAD performance and the bit error probability do not always improve by increasing the average per-antenna power $P_b^\star$. As shown in Figure \ref{fig5}, the performance is worse for $P_b^\star=10$ than for $P_b^\star=3$. This is because,  for $P_b^\star=10$, the higher transmit power could not compensate for the higher distortion caused by using a higher value for $\rho$. 

\begin{figure} 
	\centering
	\includegraphics[width=3 in]{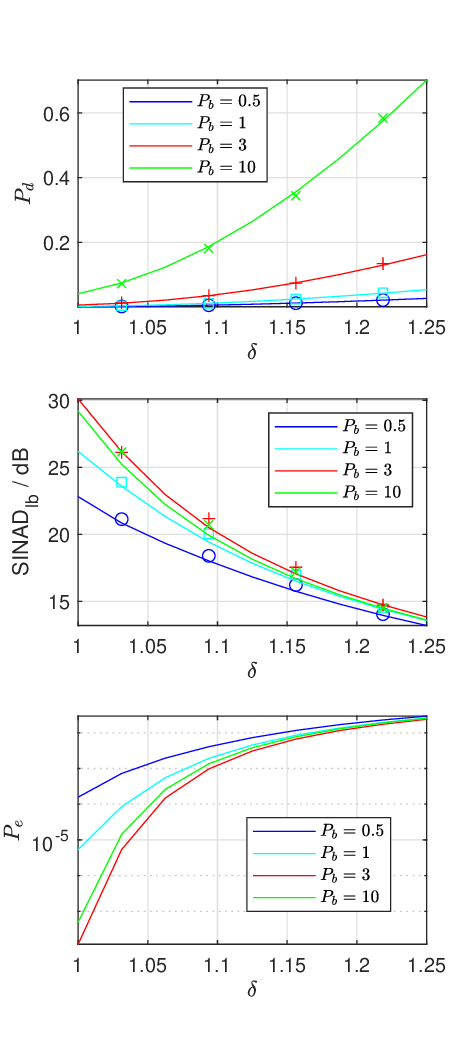}
	\caption{Impact of $\delta$ for different $P_b^\star$ values. $P=50$ $n=256$, $\lambda=0.002$ and $\sigma=0.05$. The value of $\rho$ is tuned for each plot to achieve the target $P_b^\star$. We do not show simulated results for $P_e$ because they require too many runs to simulate such a low $P_e$.}
	\label{fig5}
\end{figure}

\begin{figure} 
	\centering
	\includegraphics[width=3 in]{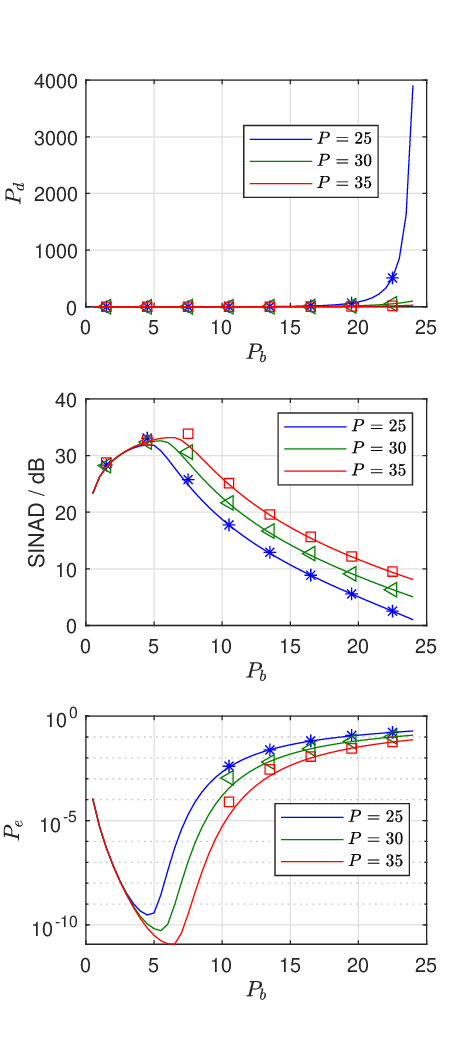}
	\caption{The distortion power, SINAD and bit error probability versus $P_b$, for different  $P$ values. $n=256$, $\delta=1$ and $\sigma=0.05$. The control parameter $\rho$ is tuned to achieve the target $P_b^\star$ and the regularization parameter $\lambda$ is set to the value that maximizes the SINAD.}
	\label{fig8}
\end{figure}

\subsection{Comparison between precoders with optimal regularization}

Figure \ref{fig8} demonstrates the performance variation with $P_b$ for $P=25, 30$, and $35$, when the regularization parameter $\lambda$ is set to the value optimizing the SINAD. As in Figure \ref{fig6} and Figure \ref{fig5}, $\rho$ is tuned to achieve the target $P_b^\star$. As an important remark, we note that there exists a $P_b^\star$ for which the performances in terms of bit error probability and SINAD are optimal. This value is below $P$. Indeed, as $P_b^\star$ approaches $P$, the distortion power significantly increases, resulting in a large performance deterioration. In a final experiment, we compare in Figure \ref{fig9}  the bit error probability performances of the limited PAPR precoder using optimal regularization with the one-bit precoding $\tilde{\bf x}_{P_b}=\sqrt{P_b^\star}{\rm sign}({\bf H}^{T}{\bf s})$. 
Complexitywise, the one-bit precoding possesses an explicit form and thus is more computationally efficient than the limited PAPR precoder which is based on solving a convex optimization problem. However, when it comes to bit error probability performance, we can easily see that the limited PAPR precoder is  more efficient for all values of $P_b^\star$. As expected from Theorem \ref{th:limit_rho}, the performance gap is small when $P_b^\star$ approaches $P$ but becomes much more pronounced when the limited PAPR precoder uses the optimal value of $P_b^\star$.
 
\begin{figure} 
	\centering
	\includegraphics[width=3 in]{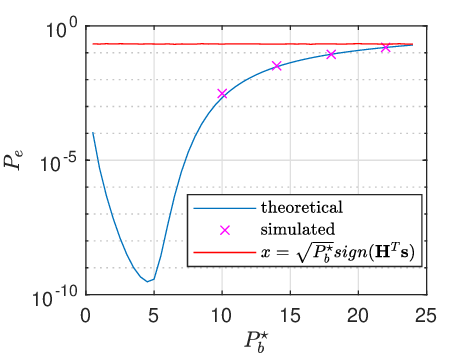}
	\caption{The bit error probability versus $P_b$ given $P=25$ ($P_b$ increased by increasing the underlying $\rho$ ). Here we assign $\lambda$ values which make optimal SINAD's, $n=256$, $\delta=1$ and $\sigma=0.05$.}
	\label{fig9}
\end{figure}

\section{Conclusion}
\label{sec9}
In this paper, we studied the asymptotic behavior of the limited PAPR precoder for multi-user communication systems in the regime in which the number of antennas and that of users grow large at the same pace. Contrary to the previous studies in \cite{papr} and \cite{glse}, we rely on the CGMT framework and present  approximations for other important performance metrics including the bit error probability and the average per user SINAD. To get more insights, we particularized our results to specific regimes in which the number of antennas is much larger than that of users, or the power control parameter takes very small or very high values. As a major outcome, our analysis demonstrates the existence of an optimal transmit power that maximizes the SINAD, and the bit error probability performances.

\section{Proof of main results}
\label{sec5}
\subsection{The CGMT framework}
In this section, we prove the main results in section~\ref{sec3}. The main technical ingredient is the CGMT.  Before delving into the technical details of the proof, we provide a brief overview of the CGMT tool. 

The CGMT is a mathematical framework that allows us to study the asymptotic behavior of high-dimensional optimization problems that can be written in the form of 
\begin{equation}
	\Phi(\mathbf{G}):=\underset{\mathbf{w}\in\mathcal{S}_{\mathbf{w}}}{\min}\underset{\mathbf{u}\in\mathcal{S}_{\mathbf{u}}}{\max}\mathbf{u}^T\mathbf{Gw}+\psi(\mathbf{w},\mathbf{u}),\label{poao_po}
\end{equation}
where $\mathbf{G}\in\mathbb{R}^{m\times n}$ is a standard Gaussian matrix, $\psi$ is a real-valued  function possibly random but independent of ${\bf G}$, and $\mathcal{S}_{\bf w}$ and $\mathcal{S}_{\bf u}$ are two  compact sets. The problem defined in \eqref{poao_po} is known as the primary optimization problem (PO). The CGMT infers the behavior of the PO by considering the following associated auxiliary optimization problem (AO):
 
\begin{equation}
	\phi(\mathbf{g},\mathbf{h}):=\underset{\mathbf{w}\in\mathcal{S}_{\mathbf{w}}}{\min}\underset{\mathbf{u}\in\mathcal{S}_{\mathbf{u}}}{\max}\|\mathbf{w}\|\mathbf{g}^T\mathbf{u}-\|\mathbf{u}\|\mathbf{h}^T\mathbf{w}+\psi(\mathbf{w},\mathbf{u}),\label{poao_ao}
\end{equation}
where  $\mathbf{g}\in\mathbb{R}^m$ and $\mathbf{h}\in\mathbb{R}^n$ two standard Gaussian vectors. More formally the CGMT is stated as follows: 

\begin{theorem}[CGMT]
	\label{CGMT}
	Consider the optimization problems in \eqref{poao_po} and \eqref{poao_ao}
	The following statements hold true:
	\begin{itemize}
		\item For all $t\in\mathbb{R}$,
		\begin{equation}
			\mathbb{P}\left[\Phi(\mathbf{G})\leq t\right]\leq 2\mathbb{P}\left[\phi({\bf g},{\bf h})\leq t\right]. \label{eq:ineq1}
		\end{equation}
		\item If additionally $\mathcal{S}_{\bf w}$ and $\mathcal{S}_{\bf u}$ are convex and $\psi$ is convex-concave, then for all $t\in\mathbb{R}$,
		\begin{equation}
			\mathbb{P}\left[\Phi(\mathbf{G})\geq t\right]\leq 2\mathbb{P}\left[\phi({\bf g},{\bf h})\geq t\right]. \label{eq:ineq2}
		\end{equation}
		Particularly, for any $\nu\geq 0$ it holds that:
		\begin{equation*}
		\mathbb{P}\left[\left|\Phi(\mathbf{G})-\nu\right|\geq t\right]\leq 2\mathbb{P}\left[\left|\phi({\bf g},{\bf h})-\nu\right|\geq t\right].
		\end{equation*}
	\end{itemize}
\end{theorem}
According to the first statement in Theorem \ref{CGMT}, 
\begin{equation*}
\mathbb{P}\left[\Phi(\mathbf{G})\leq t\right]\leq 2\mathbb{P}\left[\phi({\bf g},{\bf h})\leq t\right].
\end{equation*} 
Equivalently stated, this implies that a high-probability lower bound of the AO cost is also a high probability lower bound of the PO. Such a result holds even when the sets or the function $\psi$ are not convex. 

However, the main interest in the CGMT lies in the second statement of Theorem \ref{CGMT}, which affirms that under convexity conditions of the PO, the AO can be used to infer properties on the PO's asymptotic cost. More precisely, if for some $\nu$, the AO cost concentrates around $\nu$, so does cost of the PO.  Moreover, as shall be shown next, under appropriate strong-convexity conditions with respect to the solutions of the AO, the CGMT shows that  concentration of Lipschitz functions of the solution of the AO implies concentration of that of the PO.  

In the sequel, we make use of the CGMT framework to analyze the performance of the PAPR precoding scheme. As a first step, we express the PAPR precoding problem as a PO problem. 

\subsection{ Relating the PAPR precoding  problem to POs  }

\noindent{\bf Formulation of the POs.}
For $\mathbb{X}=[-\sqrt{P},\sqrt{P}]$, the solution of the regularized least squares problem is given by
\begin{equation}
\hat{\bf x}=\arg\min_{x_i^2\leq P} \frac{1}{n}\|{\bf Hx}-\sqrt{\rho}{\bf s}\|^2+\frac{\lambda}{n} \|{\bf x}\|^2, \label{eq:op}
\end{equation}
where compared to \eqref{eq:nonlse}, we normalized the optimization cost by $\frac{1}{n}$. 
Using the following identity:
\begin{equation*}
\|{\bf z}\|^2= \max_{{\bf u}\in \mathbb{R}^{m}} {\bf u}^{T}{\bf z}-\frac{\|{\bf u}\|^2}{4},
\end{equation*}
which holds for any vector ${\bf z}\in\mathbb{R}^{m}$, 
we can write the optimization problem in \eqref{eq:op} as
\begin{equation}
	\underset{x_i^2\le P}{\min}\underset{\mathbf{u}}{\max}\frac{\sqrt{n}\mathbf{u}^T\mathbf{Hx}}{n}-\frac{\sqrt{\rho}\mathbf{u}^T\mathbf{s}}{\sqrt{n}}-\frac{\|\mathbf{u}\|^2}{4}+\frac{\lambda\|\mathbf{x}\|^2}{n}.
	\label{afterC0}
\end{equation}
The above problem is in the form of the PO, except that the constraint set over ${\bf u}$ is not bounded. From first-order optimality conditions, we can easily check that the optimal ${\bf u}$ is given by
\begin{equation*}
{\bf u}^\star=2\left(\frac{1}{\sqrt{n}}{\bf Hx}-\frac{\sqrt{\rho}{\bf s}}{\sqrt{n}}\right).
\end{equation*}
Hence, 
\begin{equation*}
\|{\bf u}^\star\|\leq 2\sqrt{P}\left\|{\bf H}\right\| +\frac{\sqrt{\rho}\sqrt{m}}{\sqrt{n}}.
\end{equation*}
Also, using standard inequalities of the spectral norm of Gaussian matrices, we can prove that  $\|{\bf H}\|\leq \mathcal{B}$ with probability approaching $1$ for some positive constant $\mathcal{B}$. All this shows that
$
\|{\bf u}^\star\|
$ is bounded with probability approaching $1$. Thus the analysis would not thus change if we instead consider the following problem:
\begin{equation}
		\underset{x_i^2\le P}{\min}\underset{\mathbf{u}\in\mathcal{S}_{\mathbf{u}}}{\max}\frac{\sqrt{n}\mathbf{u}^T\mathbf{Hx}}{n}-\frac{\sqrt{\rho}\mathbf{u}^T\mathbf{s}}{\sqrt{n}}-\frac{\|\mathbf{u}\|^2}{4}+\frac{\lambda\|\mathbf{x}\|^2}{n},
		\label{afterC1}
\end{equation}
where $\mathcal{S}_{\bf u}=\left\{{\bf u}\in\mathbb{R}^m, \ \ \|{\bf u}\|\leq \mathcal{B}\right\}$ for some $\mathcal{B}>0$ is a high-probability upper bound on $2\|{\bf u}^\star\|$. 
Our interest is to characterize the asymptotic behavior of the solutions in ${\bf x}$ and ${\bf u}$ to \eqref{afterC1}, which  perfectly agrees with the conditions required by the CGMT. For that, 
we introduce the following cost functions:
\begin{align}
	\mathcal{C}_{\lambda,\rho}({\bf x})= \max_{{\bf u}\in\mathcal{S}_{{\bf u}}} \frac{\sqrt{n}{\bf u}^{T}{\bf Hx}}{n}  
	-\frac{\sqrt{\rho}{\bf u}^{T}{\bf s}}{\sqrt{n}}- \frac{\|{\bf u}\|^2}{4}+\frac{\lambda\|{\bf x}\|^2 }{n},\\
	\mathcal{V}_{\lambda,\rho}({\bf u})=\min_{x_i^2\leq P} \frac{\sqrt{n}{\bf u}^{T}{\bf Hx}}{n}  
	-\frac{\sqrt{\rho}{\bf u}^{T}{\bf s}}{\sqrt{n}}- \frac{\|{\bf u}\|^2}{4}+\frac{\lambda\|{\bf x}\|^2 }{n},
\end{align}
and consider the following primary problems:
\begin{align}
	\Phi_{\lambda,\rho}({\bf H})&:=\min_{x_i^2\leq P} \mathcal{C}_{\lambda,\rho}({\bf x}) \label{eq:PO1},\\
	\tilde{\Phi}_{\lambda,\rho}({\bf H})&:= \max_{{\bf u}\in\mathcal{S}_{\bf u}} \mathcal{V}_{\lambda,\rho}({\bf u}). \label{eq:PO2}
	\end{align}

Since the objective function in \eqref{afterC1} is convex in ${\bf x}$ and concave in ${\bf u}$, then
\begin{equation}
	\Phi_{\lambda,\rho}({\bf H})=\tilde{\Phi}_{\lambda,\rho}({\bf H}), \label{eq:equality}
\end{equation}
and
 the solutions $\hat{\bf x}_{\rm PO}$ and $\hat{\bf u}_{\rm PO}$  to \eqref{afterC1} are given by \footnote{Note that the objective in \eqref{afterC1} is strictly convex in ${\bf x}$ and strictly concave in ${\bf u}$. Hence, the solutions ${\bf x}$ and ${\bf u}$ are unique.}:
\begin{align}
	\hat{\bf x}_{\rm PO}:=\arg\min_{x_i^2\leq P} \mathcal{C}_{\lambda,\rho}({\bf x}) \label{eq:PO_x},\\
	\hat{\bf u}_{\rm PO}:=\arg\max_{{\bf u}\in\mathcal{S}_{\bf u}} \mathcal{V}_{\lambda,\rho}({\bf u}) \label{eq:PO_u}.
\end{align}
\noindent{\bf Formulation of the AOs.}
With the PO problems in \eqref{eq:PO1} and \eqref{eq:PO2}, we associate the following AO problems:
\begin{align}
	\phi_{\lambda,\rho}({\bf g},{\bf h})&:=\min_{x_i^2\leq P} \mathcal{L}_{\lambda,\rho}({\bf x}) \label{eq:AO1},\\
	\tilde{\phi}_{\lambda,\rho} ({\bf g},{\bf h})&=\max_{{\bf u}\in\mathcal{S}_{{\bf u}}} \mathcal{F}_{\lambda,\rho}({\bf u}), \label{eq:AO2}
	\end{align}
where $\mathcal{L}_{\lambda,\rho}({\bf x})$ and $ \mathcal{F}_{\lambda,\rho}({\bf u})$ are given by
\begin{align}
	\mathcal{L}_{\lambda,\rho}({\bf x}):=\underset{\mathbf{u}\in\mathcal{S}_{\mathbf{u}}}{\max}\frac{1}{n}\|\mathbf{x}\|\mathbf{g}^T\mathbf{u}-\frac{1}{n}\|\mathbf{u}\|\mathbf{h}^T\mathbf{x}
	-\frac{\sqrt{\rho}\mathbf{u}^T\mathbf{s}}{\sqrt{n}}-\frac{\|\mathbf{u}\|^2}{4}+\frac{\lambda\|\mathbf{x}\|^2}{n} \label{eq:L},\\
	\mathcal{F}_{\lambda,\rho}({\bf u}):=\min_{x_i^2\leq P} \frac{1}{n}\|\mathbf{x}\|\mathbf{g}^T\mathbf{u}-\frac{1}{n}\|\mathbf{u}\|\mathbf{h}^T\mathbf{x} 
	-\frac{\sqrt{\rho}\mathbf{u}^T\mathbf{s}}{\sqrt{n}}-\frac{\|\mathbf{u}\|^2}{4}+\frac{\lambda\|\mathbf{x}\|^2}{n}. \label{eq:F}
\end{align}

Similarly, we define the solutions $\hat{\bf x}^{\rm AO}$ and $\hat{\bf u}^{\rm AO}$ as
\begin{align}
	\hat{\bf x}^{\rm AO}&:= \arg\min_{\substack{{\bf x}\\ x_i^2\leq P}} \mathcal{L}_{\lambda,\rho}({\bf x}), \label{eq:AO_x}\\
	\hat{\bf u}^{\rm AO}&:= \arg\max_{{\bf u}\in\mathcal{S}_{\mathbf{u}}} \mathcal{F}_{\lambda,\rho} ({\bf u}). \label{eq:AO_u}
\end{align}
The objective of the CGMT is to prove that the properties of the solutions of the PO defined in \eqref{eq:PO_x} and \eqref{eq:PO_u} can be transferred to the solutions of the AO defined in \eqref{eq:AO_x} and \eqref{eq:AO_u}.  This can be performed by using the following inequalities which directly follow as a direct application of Theorem \ref{CGMT}. 

\begin{itemize}
	\item For all $t\in\mathbb{R}$ and any compact sets $\tilde{\mathcal{S}}_{{\bf x}}$ and $\tilde{\mathcal{S}}_{\bf u}$, 
	\begin{align}
		&\mathbb{P}\left[\min_{{\bf x}\in\tilde{\mathcal{S}}_{{\bf x}}} \mathcal{C}_{\lambda,\rho}({\bf x})\leq t\right]\leq  2\mathbb{P}\left[\min_{{\bf x}\in\tilde{\mathcal{S}}_{{\bf x}}} \mathcal{L}_{\lambda,\rho}({\bf x})\leq t\right], \label{eq:1}\\
		& \mathbb{P}\left[\max_{{\bf u}\in\tilde{\mathcal{S}}_{{\bf u}}} \mathcal{V}_{\lambda,\rho}({\bf u})\geq t\right]\leq  2\mathbb{P}\left[\max_{{\bf u}\in\tilde{\mathcal{S}}_{{\bf u}}} \mathcal{F}_{\lambda,\rho}({\bf u})\geq t\right]. \label{eq:2}
	\end{align}
	\item If  the sets $\mathcal{S}_{{\bf x}}$ and $\mathcal{S}_{\bf u}$ are convex, we have for all $t\in\mathbb{R}$, 
	\begin{align}
		&\mathbb{P}\left[\min_{{\bf x}\in\mathcal{S}_{{\bf x}}} \mathcal{C}_{\lambda,\rho}({\bf x})\geq t\right]\leq  2\mathbb{P}\left[\min_{{\bf x}\in\mathcal{S}_{{\bf x}}} \mathcal{L}_{\lambda,\rho}({\bf x})\geq t\right], \label{eq:3}\\
		& \mathbb{P}\left[\max_{{\bf u}\in\mathcal{S}_{{\bf u}}} \mathcal{V}_{\lambda,\rho}({\bf u})\leq t\right]\leq  2\mathbb{P}\left[\max_{{\bf u}\in\mathcal{S}_{{\bf u}}} \mathcal{F}_{\lambda,\rho}({\bf u})\leq t\right]. \label{eq:4}
	\end{align}
\end{itemize}

Particularly, the above inequalities can be used to prove concentrations of the optimal cost of the PO. Interestingly, they can also lead to transfer concentrations of the optimal solution of the AO to that of the PO, under some strong-convexity properties on the AO problem. 
It is thus a key step before delving into the technical proofs of Theorem \ref{th:previous} and Theorem \ref{th:distortion} to analyze the behavior of the AO problems in \eqref{eq:AO1} and \eqref{eq:AO2}. This is  the main purpose of the following Lemmas, the proof of which is deferred to Section \ref{sec6} and \ref{sec7} to avoid disrupting the flow of the proof. 
Particularly, Lemma \ref{eq:lemma_asym} establishes the uniqueness and the existence of the solutions to the asymptotic AO problem introduced in Theorem \ref{th:previous}. Lemma \ref{lem:technical_1} and Lemma \ref{lem:technical_2} provide the technical ingredients to study the asymptotic behavior of the  solutions $\hat{\bf x}^{\rm AO}$ and $\hat{\bf u}^{\rm AO}$ in \eqref{eq:AO_x} and \eqref{eq:AO_u}.

\begin{lemma}[Behavior of the asymptotic optimization problem]
Define $\overline{\phi}$ as the following deterministic max-min problem:
\begin{equation}
	\overline{\phi}:=\max_{\beta\geq 0}\min_{\tau\geq 0} \ \ \mathcal{D}(\beta,\tau):=\frac{\tau\beta\delta}{2}+\frac{\rho\beta}{2\tau}-\frac{\beta^2}{4} +Y(\beta,\tau), 
	\label{eq:asde1}
\end{equation}
where $Y(\beta,\tau)$ is given by
\begin{align}
Y(\beta,\tau)&:=\beta\sqrt{P}\mathbb{E}\left[\left(\sqrt{P}\alpha-2H\right){\bf 1}_{\{H\geq \sqrt{P}\alpha\}}\right]-\frac{\beta}{2\alpha} \mathbb{E}\left[H^2{\bf 1}_{\{-\sqrt{P}\alpha\leq H\leq \sqrt{P}\alpha\}}\right] \label{eq:rep1}\\
&=\frac{\beta}{\alpha}\left(\mathbb{E}\left[(H-\sqrt{P}\alpha)^2{\bf 1}_{\{H\geq \sqrt{P}\alpha\}}\right]-\frac{1}{2}\right) \label{eq:rr}
\end{align}
and \footnote{The compact expression  in \eqref{eq:rr} is easily obtained by noticing that $\mathbb{E}\left[H^2{\bf 1}_{\{-\sqrt{P}\alpha\leq H\leq \sqrt{P}\alpha\}}\right]=1-2\mathbb{E}\left[H^2{\bf 1}_{\{H\geq \sqrt{P}\alpha\}}\right]$.} $\alpha=1/\tau+2\lambda/\beta$ . Then the following statements hold:
\begin{enumerate}
	\item The function $\beta\mapsto\mathcal{D}(\beta,\tau)$ is strictly concave. 
\item The above optimization problem possesses a unique finite saddle point $(\beta^\star,\tau^\star)$ if and only if $\lambda>0$ or $\lambda=0$ and $\delta>1$. Moreover, in both cases $\beta^\star>0$. 
\item The solution $\tau^\star$ satisfies the following  equation:
\begin{align}
	(\tau^\star)^2\delta-\rho &=2P \mathbb{P}\left[H\geq \sqrt{P}\alpha^\star\right]+\frac{1}{(\alpha^\star)^2} \mathbb{E}\left[H^2{\bf 1}_{\{-\sqrt{P}\alpha^\star\leq H\leq \sqrt{P}\alpha^\star\}}\right]. \label{eq:tau_r}
\end{align}
\item Assume $\lambda=0$ and $\delta>1$. Then 
the saddle point $(\beta^\star,\tau^\star)$ of \eqref{eq:asde1} is given by:
\begin{align}
	\tau^\star&=\arg\min_{\tau\geq 0} \left(\frac{\tau\delta}{2}+\frac{\rho}{2\tau}+\tilde{Y}(\tau)\right)_{+}^2,\\
	\beta^\star&= \left({\tau^\star\delta}+\frac{\rho}{\tau^\star}+2\tilde{Y}(\tau^\star)\right)_{+},
\end{align}
and
\eqref{eq:asde1} reduces to:
\begin{equation*}
\overline{\phi}:= \min_{\tau\geq 0} \left(\frac{\tau\delta}{2}+\frac{\rho}{2\tau}+\tilde{Y}(\tau)\right)_{+}^2.
\end{equation*}
where
\begin{align}
\tilde{Y}(\tau)&=\sqrt{P}\mathbb{E}\left[\left(\frac{\sqrt{P}}{\tau}-2H\right){\bf 1}_{\{H\geq \frac{\sqrt{P}}{\tau}\}}\right]-\frac{\tau}{2}\mathbb{E}\left[H^2{\bf 1}_{\{-\frac{\sqrt{P}}{\tau}\leq H\leq \frac{\sqrt{P}}{\tau}\}}\right] \label{eq:Y_tilde_tau}\\
&= \tau \left(\mathbb{E}\left[(H-\frac{\sqrt{P}}{\tau})^2{\bf 1}_{\{H\geq \frac{\sqrt{P}}{\tau}\}}\right]-\frac{1}{2}\right).
\end{align}

Moreover, $\tau^\star$ is the unique solution to the following equation:
\begin{equation}
\tau^2\delta-\rho=2P\mathbb{P}\left[H\geq\frac{\sqrt{P}}{\tau}\right] +\tau^2 \mathbb{E}\left[H^2{\bf 1}_{\{-\frac{\sqrt{P}}{\tau}\leq H\leq \frac{\sqrt{P}}{\tau}\}}\right]. \label{eq:tau_sol}
\end{equation}
\end{enumerate}
\label{eq:lemma_asym}
\end{lemma}
\begin{proof}
	See Section \ref{sec6}.
\end{proof}
\begin{lemma}
	Let $\overline{\phi}$ be the deterministic max-min problem defined in \eqref{eq:asde1}, and denote by $(\beta^\star,\tau^\star)$ its associated saddle point. Define $\overline{\bf x}^{\rm AO}$ as follows:
	\begin{equation*}
	\left[\overline{\bf x}_{\rm AO}\right]_{i}:=\left\{ 
	\begin{array}{ll}
		-\sqrt{P}, &  h_i\leq -\sqrt{P}\alpha^\star \\
		\frac{h_i}{\alpha^\star}, & -\sqrt{P} \alpha^\star\leq h_i \leq \sqrt{P} \alpha^\star\\
		\sqrt{P}, & h_i \geq \sqrt{P}\alpha^\star
	\end{array},
	\right.
	\end{equation*}
	where $\alpha^\star=1/\tau^\star+2\lambda/\beta^\star$ , $h_i$'s are independent standard Gaussian random variables.
	Then, the following statements hold true:
	\begin{enumerate}
\item  There exists a function $\hat{\mathcal{L}}_{\lambda,\rho}({\bf x})$ such that
\begin{itemize}
	\item The function ${\bf x}\mapsto \hat{\mathcal{L}}_{\lambda,\rho}({\bf x})$ is $\frac{\lambda}{n}$-strongly convex  for $\lambda>0$ and locally strongly convex in a neighborhood of $\overline{\bf x}^{\rm AO}$ when $\lambda=0$ and $\delta>1$. 
	\item The following convergences holds true:
	\begin{equation*}
	\sup_{\substack{{\bf x}\\ x_i^2\leq P}} \left|\hat{\mathcal{L}}_{\lambda,\rho}({\bf x})-\mathcal{L}_{\lambda,\rho}({\bf x})\right|\stackrel{P}{\rightarrow} 0,
	\end{equation*}
	and
	\begin{equation*}
	\min_{\substack{{\bf x}\\ x_i^2\leq P}} \hat{\mathcal{L}}_{\lambda,\rho}({\bf x})-\overline{\phi}\stackrel{P}{\rightarrow} 0. 
	\end{equation*}
\end{itemize}
\item Let $\tilde{\bf x}^{\rm AO}$ be a minimizer of $\hat{\mathcal{L}}_{\lambda,\rho}$. Then, for any $\epsilon>0$, with probability approaching $1$,
\begin{equation*}
\frac{1}{n}\|\tilde{\bf x}^{\rm AO}-\overline{\bf x}^{\rm AO}\|\leq \epsilon,
\end{equation*}
and
\begin{equation}
\hat{\mathcal{L}}_{\lambda,\rho}(\overline{\bf x}^{\rm AO}) \leq \min_{\substack{{\bf x}\\ x_i^2\leq P}}\hat{\mathcal{L}}_{\lambda,\rho}({\bf x}) +\epsilon. \label{eq:toprove}
\end{equation}
\end{enumerate}
\label{lem:technical_1}
\end{lemma}
\begin{proof}
	See Section \ref{pl2}.
\end{proof}
\begin{lemma}
	Consider the setting of Lemma \ref{lem:technical_1}. 
	Then, the following statements hold true:
	\begin{enumerate}
		\item With probability approaching $1$, we have
		\begin{equation}
		\left|\tilde{\phi}_{\lambda,\rho}({\bf g},{\bf h}) -\overline{\phi}\right|\stackrel{P}{\rightarrow} 0. \label{eq:con}
		\end{equation}
	\item Define $\tilde{\mathcal{F}}_{\lambda,\rho}({\bf u})$ as
	\begin{align}
	\tilde{\mathcal{F}}_{\lambda,\rho}({\bf u})&= \frac{1}{n} \|\overline{\bf x}^{\rm AO}\| {\bf g}^{T}{\bf u} -\frac{1}{n}\|{\bf u}\|{\bf h}^{T}\overline{\bf x}^{\rm AO} -\frac{\sqrt{\rho}}{\sqrt{n}}{\bf u}^{T}{\bf s} -\frac{\|{\bf u}\|^2}{4}+\frac{\lambda}{n}\|\overline{\bf x}^{\rm AO}\|^2,
	\end{align}
then $\tilde{\mathcal{F}}_{\lambda,\rho}$ is $\frac{1}{2}$-strongly-concave, and 
\begin{equation}
\left|\max_{{\bf u}} \mathcal{\tilde{F}}_{\lambda,\rho}({\bf u})-\overline{\phi}\right|\stackrel{P}{\rightarrow} 0. \label{eq:conv0}
\end{equation}

		\item Denote by $\tilde{\bf u}^{\rm AO}$ the maximizer of $\tilde{\mathcal{F}}_{\lambda,\rho}$:
		\begin{equation*}
		\tilde{\bf u}^{\rm AO}=\arg\max_{{\bf u}} \tilde{\mathcal{F}}_{\lambda,\rho}({\bf u}),
		\end{equation*}
		let $\overline{\bf u}^{\rm AO}$ be given by
		\begin{equation*}
		\overline{\bf u}^{\rm AO}=\frac{\beta^\star(\sqrt{(\tau_\star)^2\delta-\rho}{\bf g}-\sqrt{\rho}{\bf s})}{\tau^\star\delta\sqrt{n}},
		\end{equation*}
		then for any $\epsilon>0$, with probability approaching $1$, 
		\begin{equation}
		\|\overline{\bf u}^{\rm AO}-\tilde{\bf u}^{\rm AO}\|\leq \epsilon.
	\label{eq:conv3}	
	\end{equation}
		\end{enumerate}
	\label{lem:technical_2}
\end{lemma}
\begin{proof}
	See Section \ref{pl3}.
\end{proof}
\subsection{Proof of Theorem \ref{th:previous}}
\label{proof_th:previous}
With Lemma \ref{lem:technical_1} at hand, we are now ready to develop Theorem \ref{th:previous}. 
Let $f$ be a Lipschitz function. For ${\bf x}\in\mathbb{R}^{n}$ we define
\begin{equation*}
F({\bf x})=\frac{1}{n}\sum_{i=1}^n f({\bf x}_i)
\end{equation*} 
and $\kappa= \mathbb{E}_H\left[f(\theta(H))\right]$.
For any $\epsilon>0$, define the set
\begin{equation*}
\mathcal{S} = \left\{{\bf x} \in\mathbb{R}^n, x_i^2\leq P | \left|F({\bf x})-\kappa\right|\geq 2\epsilon\right\}.
\end{equation*}
Consider the following 'perturbed' PO and AO problems
\begin{equation*}
\Phi_{\mathcal{S}}({\bf H}) = \min_{{\bf x}\in\mathcal{S}} \mathcal{C}_{\lambda,\rho}({\bf x}),
\end{equation*}
and 
\begin{equation*}
\phi_{\mathcal{S}}({\bf g},{\bf h}) = \min_{{\bf x}\in\mathcal{S}} \mathcal{L}_{\lambda,\rho}({\bf x}).
\end{equation*}
 To prove the third statement of Theorem \ref{th:previous}, it suffices to show that with probability approaching $1$, 
\begin{equation}
\Phi_{S}({\bf H}) > \Phi_{\lambda,\rho}({\bf H}). \label{eq:desired}
\end{equation}
Using the CGMT, and as shown in \cite{mesti}, it suffices to find constants $\overline{\phi}_{S}$ and $\eta$ such that with probability approaching $1$, the following statements hold true:
\begin{enumerate}
	\item $\overline{\phi}_{S} \geq \overline{\phi} +3\eta$,
	\item ${\phi}_{\lambda,\rho}({\bf g},{\bf h})\leq \overline{\phi}+\eta$,
	\item $\phi_{\mathcal{S}}({\bf g},{\bf h})\geq \overline{\phi}_{S}-\eta$.
	\end{enumerate}
Indeed, if the three statements above are satisfied, then from \eqref{eq:1},
\begin{equation*}
\mathbb{P}\left[\Phi_{S}({\bf H})\leq \overline{\phi}_{S}-\eta\right] \leq 2\mathbb{P}\left[\phi_{S}({\bf g},{\bf h})\leq \overline{\phi}_{S}-\eta\right],
\end{equation*}
while from statement $3)$
\begin{equation*}
\mathbb{P}\left[\phi_{S}({\bf g},{\bf h})\leq \overline{\phi}_{S}-\eta\right]\to 0.
\end{equation*}
Hence, with probability approaching $1$,
\begin{equation*}
\Phi_{S}({\bf H})\geq \overline{\phi}_{S}-\eta,
\end{equation*}
and thus, from Statement $1)$, 
\begin{equation}
\Phi_{S}({\bf H}) \geq \overline{\phi}+2\eta .
\label{eq:re1}
\end{equation}
On the other hand, we can in a similar way use \eqref{eq:3} together with Statement $2)$, to show that with probability approaching $1$,
\begin{equation}
\Phi_{\lambda,\rho}({\bf H})\leq \overline{\phi}+\eta. \label{eq:re2}
\end{equation}
Combining \eqref{eq:re1} with \eqref{eq:re2} yields the desired relation in \eqref{eq:desired}.

In the remaining, we consider finding $\overline{\phi}_{S}$ and $\eta$ values for which the aforementioned statements hold. 

\noindent{\underline{Proof of Statement $2)$.}} From Lemma \ref{lem:technical_1}, it holds that with probability approaching 1,
\begin{equation*}
\phi_{\lambda,\rho}({\bf g},{\bf h}) \leq \overline{\phi}+\eta,
\end{equation*}
which shows that Statement $2)$ holds for any $\eta>0$. 

\noindent{\underline{Proof of Statement $3)$.}} 
As a first step, we show that Statement $3)$ holds if 
for an appropriate choice of $\tilde{\epsilon}$, it holds that
\begin{equation}
\forall {\bf x}\in\mathcal{S},  \ \ \frac{a}{8n}\left\|{\bf x}-\overline{\bf x}^{\rm AO}\right\|^2\geq \tilde{\epsilon}.  \label{eq:reff}
\end{equation}
 To see this, recall that from Lemma \ref{lem:technical_1}, the function $\hat{\mathcal{L}}$ is $\frac{a}{n}$-strongly convex in a neighborhood of $\overline{\bf x}^{\rm AO}$.
Moreover, for some $\epsilon$ sufficiently small, we have proved in \eqref{eq:toprove} that:
\begin{equation*}
\hat{\mathcal{L}}_{\lambda,\rho}(\overline{\bf x}^{\rm AO}) \leq \min_{{\bf x}}\hat{\mathcal{L}}_{\lambda,\rho}({\bf x})+\epsilon.
\end{equation*}
Hence, from Lemma \ref{app:technical}, 
\begin{equation*}
\forall {\bf x},  \frac{a}{8n}\|{\bf x}-\overline{\bf x}^{\rm AO}\|^2\geq \epsilon \Longrightarrow \hat{\mathcal{L}}_{\lambda, \rho}({\bf x})\geq \min_{x_i^2\leq P} \hat{\mathcal{L}}_{\lambda,\rho}({\bf x})+\epsilon.
\end{equation*}
Consequently, if for every ${\bf x}\in\mathcal{S}$, \eqref{eq:reff} holds, then
\begin{equation*}
\forall {\bf x}\in\mathcal{S}, \ \ \hat{\mathcal{L}}_{\lambda,\rho}({\bf x})\geq \min_{x_i^2\leq P}\hat{\mathcal{L}}_{\lambda,\rho}({\bf x})+\tilde{\epsilon}\geq \phi_{\lambda,\rho}({\bf g},{\bf h})+\frac{\tilde{\epsilon}}{2}.
\end{equation*}
Hence,
\begin{equation*}
\min_{{\bf x}\in\mathcal{S}} \mathcal{L}_{\lambda,\rho}({\bf x}) \geq \phi_{\lambda,\rho}({\bf g},{\bf h})+\frac{\tilde{\epsilon}}{2}.
\end{equation*}
With this inequality at hand,  we use the fact that for any $\eta>0$, we have with probability approaching $1$
\begin{equation*}
\phi({\bf g},{\bf h})\geq \overline{\phi} -\eta.
\end{equation*}
Hence, for any $\eta\leq \frac{\tilde{\epsilon}}{4}$, we obtain
\begin{equation*}
\min_{{\bf x}\in\mathcal{S}} \mathcal{L}_{\lambda,\rho}({\bf x}) \geq \overline{\phi}+\frac{1}{2}\tilde{\epsilon}-\eta,
\end{equation*}
and hence, Statement $3)$ holds with $\overline{\phi}_S= \overline{\phi}+\frac{1}{2}\tilde{\epsilon}$. 

\noindent{\underline{Proof of \eqref{eq:reff}}.} In what follows, we will thus consider proving \eqref{eq:reff}. As a first step,
we use  the weak law of large numbers to prove the following convergence:
\begin{equation*}
\left|\frac{1}{n}\sum_{i=1}^n f(\overline{\bf x}_i^{\rm AO})-\kappa\right|\stackrel{P}{\rightarrow} 0.
\end{equation*}
Equivalently, for any $\hat{\epsilon}>0$, with probability approaching $1$, we have
\begin{equation}
\left|F(\overline{\bf x}^{\rm AO})-\kappa\right| \leq\hat{\epsilon}. \label{eq:re454}
\end{equation}
To continue, let ${\bf x}\in\mathcal{S}$, then, 
\begin{equation*}
\left|F({\bf x})-\kappa\right|\geq 2\epsilon.
\end{equation*}
Hence, using \eqref{eq:re454} yields
\begin{equation}
\left|F({\bf x})-F(\overline{\bf x}^{\rm AO})\right|\geq \left|F({\bf x})-\kappa\right| - \left|F(\overline{\bf x}^{\rm AO})-\kappa\right| \geq 2\epsilon-\hat{\epsilon}.\label{eq:rest}
\end{equation}
Since $f$ is pseudo-Lipschitz of order $k$, there exists a constant $C$ such that:
\begin{align}
&\left|F({\bf x})-F(\overline{\bf x}^{\rm AO})\right| \nonumber\\
&\leq \frac{C}{n}\sum_{i=1}^n (1+|{\bf x}_i|^{k-1}+\left|\overline{\bf x}^{\rm AO}_i\right|^{k-1})\left|{\bf x}_i-\overline{\bf x}^{\rm AO}_i\right|\\
&\leq \frac{C}{\sqrt{n}}\|{\bf x}-\overline{\bf x}^{\rm AO}\|\left(1+\sqrt{\frac{1}{n}\sum_{i=1}^n |{\bf x}_i|^{2(k-1)}}+\sqrt{\frac{1}{n}\sum_{i=1}^n |\overline{\bf x}_i^{\rm AO}|^{2(k-1)}}\right).
\end{align}
Since the absolute values of elements of ${\bf x}$ and $\overline{\bf x}$ are bounded by $\sqrt{P}$, then
\begin{equation*}
\max\left(\sqrt{\frac{1}{n}\sum_{i=1}^n |{\bf x}_i|^{2(k-1)}},\sqrt{\frac{1}{n}\sum_{i=1}^n |\overline{\bf x}_i^{\rm AO}|^{2(k-1)}}\right)\leq P^{\frac{k-1}{2}}.
\end{equation*}
Hence, 
\begin{equation*}
\left|F({\bf x})-F(\overline{\bf x}^{\rm AO})\right|\leq \frac{C}{\sqrt{n}} \|{\bf x}-\overline{\bf x}^{\rm AO}\|\left(1+2P^{\frac{k-1}{2}}\right).
\end{equation*}
When combined with \eqref{eq:rest}, this shows that
\begin{equation*}
\frac{a}{8n}\left\|{\bf x}-\overline{\bf x}^{\rm AO}\right\|^2\geq \frac{a}{8C^2}\frac{(2\epsilon-\hat{\epsilon})^2}{(1+2P^{\frac{k-1}{2}})^2}.
\end{equation*}
which proves \eqref{eq:reff} for $\tilde{\epsilon} = \frac{a}{8C^2}\frac{(2\epsilon-\hat{\epsilon})^2}{(1+2P^{\frac{k-1}{2}})^2}$. 

\noindent{\underline{Proof of Statement $1)$}}.
Recall that from the proof of Statement~$3)$ we have
\begin{equation*}
\overline{\phi}_{S}=\overline{\phi}+\frac{1}{2}\tilde{\epsilon}.
\end{equation*}
Hence, with $\eta\leq \frac{1}{6}\tilde{\epsilon}$, we have
\begin{equation*}
3{\eta} \leq \frac{1}{2}\tilde{\epsilon}\Rightarrow \overline{\phi} +3\eta \leq \overline{\phi}+\frac{\tilde{\epsilon}}{2}=\overline{\phi}_{S}.
\end{equation*}
\subsection{Proof of Theorem \ref{th:distortion}}
\label{proof_th:distortion}

Let $\tilde{f}:\mathbb{R}^2\to\mathbb{R}$ be a Lipschitz function. For $({\bf u},{\bf s})\in \mathbb{R}^{m}\times \mathbb{R}^m$, we define the function $\tilde{F}$ as
\begin{equation*}
\tilde{F}({\bf u},{\bf s})= \frac{1}{m}\sum_{i=1}^m f(\frac{\sqrt{n}}{2}{\bf u}_i,{\bf s}_i).
\end{equation*}
For $H\sim\mathcal{N}(0,1)$ and ${S}$ a discrete variable taking $1$ or $-1$ with equal probabilities, we define $X(H,S)$ as
\begin{equation*}
X(H,S)= \frac{\beta^\star(\sqrt{\delta (\tau^\star)^2-\rho}H-\sqrt{\rho}S)}{2\tau^\star\delta}.
\end{equation*}
Fix $\epsilon>0$ and define the set $\tilde{\mathcal{S}}$ as
\begin{equation*}
\tilde{\mathcal{S}}=\left\{{\bf u}\in\mathcal{S}_{{\bf u}} | \left|\tilde{F}({\bf u},{\bf s})-\mathbb{E}_{H,S}\left[f(X(H,S),S)\right]\right|\geq 2\epsilon\right\}.
\end{equation*}
With this definition, it is easily observed that the proof of Theorem \ref{th:distortion} reduces to showing that with probability approaching $1$, the solution $\hat{\bf u}^{\rm PO}\notin \tilde{S}$. To prove the desired result, we introduce the following 'perturbed' PO:
\begin{align}
\tilde{\Phi}_{\tilde{\mathcal{S}}}({\bf H})= \max_{{\bf u}\in\tilde{\mathcal{S}}} \mathcal{V}_{\lambda,\rho}({\bf u})
\end{align}
and its associated AO:
\begin{align}
	\tilde{\phi}_{\tilde{\mathcal{S}}}({\bf g},{\bf h})= \max_{{\bf u}\in\tilde{\mathcal{S}}} \mathcal{F}_{\lambda,\rho}({\bf u}).
\end{align}
Obviously, the proof of Theorem \ref{th:distortion} is equivalent to showing that with probability approaching $1$, 
\begin{equation}
\tilde{\Phi}_{\tilde{\mathcal{S}}}({\bf H}) < \tilde{\Phi}_{\lambda,\rho}({\bf H}). \label{eq:des}
\end{equation}
As in the proof of Theorem \ref{th:previous}, it suffices to find $\tilde{\phi}_{\mathcal{S}}$ and $\eta$ for which the following statements hold:
\begin{enumerate}[label={\arabic*$'$)}]
	\item $\tilde{\phi}_{\mathcal{S}} \leq \overline{\phi}-3\eta$,
	\item $\tilde{\phi}_{\lambda,\rho}({\bf g},{\bf h})\geq \overline{\phi}-\eta$,
	\item $\tilde{\phi}_{\tilde{\mathcal{S}}}({\bf g},{\bf h}) \leq \tilde{\phi}_{\mathcal{S}}+\eta$.
	\end{enumerate}
Indeed, if the above statements hold then Statement $2')$ and \eqref{eq:4} imply that with probability approaching $1$, 
\begin{equation}
\tilde{\Phi}_{\lambda,\rho}({\bf H})\geq \overline{\phi}-\eta.
\label{eq:s1}
\end{equation}
Similarly, using \eqref{eq:2} and Statement $3')$,
\begin{equation}
\tilde{\Phi}_{\tilde{\mathcal{S}}}({\bf H})\leq \tilde{\phi}_{\mathcal{S}}+\eta \label{eq:s2}.
\end{equation}
Combining \eqref{eq:s1} and \eqref{eq:s2} with Statement $1')$ yields \eqref{eq:des}. In what follows, we exploit Lemma \ref{lem:technical_2} to show that the above statements are true. The proof shares similarities with that of Theorem \ref{th:previous}. However, for the sake of completeness, we provide all the details. 

\noindent{\underline{Proof of Statement $2')$}}
Recall that $\tilde{\phi}_{\lambda,\rho}({\bf g},{\bf h})$ is given by
\begin{equation*}
\tilde{\phi}_{\lambda,\rho}({\bf g},{\bf h})= \max_{{\bf u}\in\mathcal{S}_{\bf u}} \mathcal{F}_{\lambda,\rho}({\bf u}).
\end{equation*}
It follows from  Lemma \ref{lem:technical_2} that
\begin{equation*}
\left|\max_{{\bf u}\in\mathcal{S}_{\bf u}} \mathcal{F}_{\lambda,\rho}({\bf u})-\overline{\phi}\right|\stackrel{P}{\rightarrow} 0.
\end{equation*}
Hence, for any $\eta>0$, with probability approaching $1$,
\begin{equation*}
\tilde{\phi}_{\lambda,\rho}({\bf g},{\bf h})\geq \overline{\phi}-\eta,
\end{equation*}
which shows that Statement $2')$ is true. 

\noindent{\underline{Proof of Statement $3')$}} To begin with, we show that to prove Statement $3')$, it suffices to show that for an appropriate choice of $\tilde{\epsilon}$,
\begin{equation}
\forall {\bf u} \in \mathcal{\tilde{S}},  \ \ \ \|{\bf u}-\tilde{\bf u}^{\rm AO}\|^2\geq 2\tilde{\epsilon} .\label{eq:prop}
\end{equation}
Indeed if \eqref{eq:prop} holds true, then we can exploit the fact that $\tilde{\mathcal{F}}_{\lambda,\rho}$ is $\frac{1}{2}-$ strongly concave to obtain
\begin{equation}
\frac{1}{2}\|{\bf u}-\tilde{\bf u}^{\rm AO}\|^2\geq\tilde{\epsilon} \Longrightarrow \tilde{\mathcal{F}}_{\lambda,\rho}({\bf u})\leq \max_{{\bf u}} \tilde{\mathcal{F}}_{\lambda,\rho}({{\bf u}})-\tilde{\epsilon}. \label{eq:inviewof}
\end{equation}
Hence, if \eqref{eq:prop} holds true then, in view of \eqref{eq:inviewof},
\begin{equation*}
\forall {\bf u}\in\tilde{\mathcal{S}},  \ \ \mathcal{\tilde{F}}_{\lambda,\rho}({\bf u})\leq \max_{{\bf u}} \tilde{\mathcal{F}}_{\lambda,\rho}({{\bf u}})-\tilde{\epsilon}.
\end{equation*}
Now, using the fact that $\mathcal{\tilde{F}}_{\lambda,\rho}({\bf u})\geq \mathcal{F}_{\lambda,\rho}({\bf u})$, we have 
\begin{equation*}
\forall {\bf u}\in\tilde{\mathcal{S}}\ \ \mathcal{F}_{\lambda,\rho}({\bf u}) \leq \max_{{\bf u}}\tilde{\mathcal{F}}_{\lambda,\rho}({{\bf u}})-\tilde{\epsilon},
\end{equation*}
and hence,
\begin{equation*}
\tilde{\phi}_{\tilde{\mathcal{S}}}({\bf g},{\bf h}) \leq \max_{{\bf u}}\tilde{\mathcal{F}}_{\lambda,\rho}({{\bf u}})-\tilde{\epsilon}.
\end{equation*}
With this inequality at hand, we use the fact that for any $\eta>0$, with probability approaching $1$, it holds that
\begin{equation*}
\max_{{\bf u}}\tilde{\mathcal{F}}_{\lambda,\rho}({{\bf u}})\leq \overline{\phi}+\eta.
\end{equation*}
We thus obtain
\begin{equation*}
\tilde{\phi}_{\tilde{\mathcal{S}}}({\bf g},{\bf h}) \leq \overline{\phi}+\eta-\tilde{\epsilon},
\end{equation*}
and hence Statement $3')$ holds with $\tilde{\phi}_{{\mathcal{S}}}=\overline{\phi}-\tilde{\epsilon}$. 

\noindent{\underline{Proof \eqref{eq:prop}}.} In a first step, we show that for any $\hat{\epsilon}$, with probability approaching $1$,
\begin{equation}
\left|\tilde{F}(\tilde{\bf u}^{\rm AO},{\bf s})-\tilde{\kappa}\right|\leq 2\hat{\epsilon} ,\label{eq:tilde_F}
\end{equation}
where 
\begin{equation*}
\tilde{\kappa}:= \mathbb{E}\left[f(X(H,S),S)\right].
\end{equation*}
 Now we invoke the weak law of large numbers to prove that with a probability approaching $1$, 
\begin{equation}
\left|\tilde{F}(\overline{\bf u}^{\rm AO})-\tilde{\kappa}\right|\leq \hat{\epsilon}. \label{eq:refdw}
\end{equation}
Since $f$ is pseudo-Lipschitz of order $2$, there exists a constant $C$ such that
\begin{align}
&\frac{1}{m} \sum_{i=1}^m \left|f(\frac{\sqrt{n}}{2}\tilde{\bf u}_i^{\rm AO},{\bf s}_i)-f(\frac{\sqrt{n}}{2}\overline{\bf u}_i^{\rm AO},{\bf s}_i)\right|\\
&\leq \frac{C\sqrt{n}}{m}\sum_{i=1}^m \left(1+\|\big[\frac{\sqrt{n}}{2}\tilde{\bf u}_i^{\rm AO},{\bf s}_i\big]^{T}\|_2+\|\big[\frac{\sqrt{n}}{2}\overline{\bf u}_i^{\rm AO},{\bf s}_i\big]^{T}\|_2\right) |\tilde{\bf u}_i^{\rm AO}-\overline{\bf u}_i^{\rm AO}|
\\
&\leq C\left(\frac{\sqrt{n}}{\sqrt{m}}+\frac{n}{2m}\|\tilde{\bf u}^{\rm AO}\|+\frac{n}{2m}\|\overline{\bf u}^{\rm AO}\|+2\frac{\sqrt{n}}{m}\|{\bf s}\|\right) \|\overline{\bf u}^{\rm AO}-\tilde{\bf u}^{\rm AO}\|.\label{eq:r}
\end{align}
From Lemma \ref{lem:technical_2}, it follows that with probability approaching $1$, 
\begin{equation*}
\|\overline{\bf u}^{\rm AO}-\tilde{\bf u}^{\rm AO}\|\leq \frac{\hat{\epsilon}}{C}.
\end{equation*}
As $\|\tilde{\bf u}^{\rm AO}\|$ and $\|\overline{\bf u}^{\rm AO}\|$ are bounded in probability and $\sqrt{n}{m}\|{\bf s}\|$ is bounded, the following inequality holds with probability approaching $1$,
\begin{equation}
\frac{1}{m} \sum_{i=1}^m \left|f(\frac{\sqrt{n}}{2}\tilde{\bf u}_i^{\rm AO},{\bf s}_i)-f(\frac{\sqrt{n}}{2}\overline{\bf u}_i^{\rm AO},{\bf s}_i)\right|\leq \hat{\epsilon} .\label{eq:f}
\end{equation}
With this, we can prove \eqref{eq:tilde_F} by using the following inequality
\begin{equation*}
\left|F(\tilde{\bf u}^{\rm AO},{\bf s})-\tilde{\kappa}\right|\leq \left|F(\tilde{\bf u}^{\rm AO},{\bf s})-F(\overline{\bf u}^{\rm AO},{\bf s})\right|+ \left|F(\overline{\bf u}^{\rm AO},{\bf s})-\tilde{\kappa}\right|
\end{equation*}
along with \eqref{eq:refdw} and \eqref{eq:f}. 

Having proven \eqref{eq:tilde_F}, we are now ready to complete the proof of \eqref{eq:prop}. Take ${\bf u}\in\tilde{\mathcal{S}}$. Then, 
\begin{equation*}
\left|\tilde{F}({\bf u},{\bf s})-\tilde{\kappa}\right|\geq 2\epsilon.
\end{equation*}
We can thus lower bound $\tilde{F}({\bf u},{\bf s})-\tilde{F}(\overline{\bf u}^{\rm AO},{\bf s})$ as
\begin{align}
\left|\tilde{F}({\bf u},{\bf s})-\tilde{F}(\overline{\bf u}^{\rm AO},{\bf s})\right| \geq \left| \tilde{F}({\bf u},{\bf s})-\tilde{\kappa}\right| -\left|\tilde{F}(\overline{\bf u}^{\rm AO},{\bf s})-\tilde{\kappa} \right|
\geq 2\epsilon-\hat{\epsilon}. \label{eq:lower_bound}
\end{align}
On the other hand, following the same arguments used to obtain \eqref{eq:r}, there exists a constant $K$ such that $\tilde{F}({\bf u},{\bf s})-\tilde{F}(\overline{\bf u}^{\rm AO},{\bf s})$ is upper bounded as
\begin{align}
	\left|\tilde{F}({\bf u},{\bf s})-\tilde{F}(\overline{\bf u}^{\rm AO},{\bf s})\right| \leq K \|{\bf u}-\overline{\bf u}^{\rm AO}\|
\leq K \|{\bf u}-\tilde{\bf u}^{\rm AO}\|+K\hat{\epsilon},
\label{eq:upper_bound}
\end{align}
where in the last inequality, we exploited Lemma \ref{lem:technical_2} to use the fact that  $\|\tilde{\bf u}^{\rm AO}-\overline{\bf u}^{\rm AO}\|\leq \hat{\epsilon}$ with probability approaching $1$.  
Combining the inequalities in \eqref{eq:lower_bound} and \eqref{eq:upper_bound}, we have for  all ${\bf u}\in\tilde{\mathcal{S}}$
\begin{equation*}
\|{\bf u}-\tilde{\bf u}^{\rm AO}\|^2 \geq \left(\frac{2\epsilon-\hat{\epsilon}}{K}-\hat{\epsilon}\right)^2.
\end{equation*}
Hence \eqref{eq:prop} holds true, with $\tilde{\epsilon}=\frac{\left(\frac{2\epsilon-\hat{\epsilon}}{K}-\hat{\epsilon}\right)^2}{2}$. 

\noindent{\underline{Proof of Statement $1')$}}
Recall that from the proof of Statement $3')$ we have 
\begin{equation*}
	\tilde{\phi}_{\mathcal{S}}=\bar{\phi}-\tilde{\epsilon}.
\end{equation*}
Hence, with $\eta<\frac{\tilde{\epsilon}}{3}$, we have
\begin{equation*}
	\tilde{\phi}_{\mathcal{S}}<\bar{\phi}-3\eta.
\end{equation*}

\section{Proof of Lemma \ref{eq:lemma_asym}}
\label{sec6}
\begin{enumerate}
	\item	The function $\beta\mapsto -\frac{\beta^2}{4}$ is strictly concave. We can readily check that $\beta\mapsto \min_{\tau\geq 0 }\frac{\tau\beta\delta}{2}+\frac{\rho\beta}{2\tau}+Y(\beta,\tau)$ is concave. Hence, $\beta\mapsto \mathcal{D}(\beta,\tau)$ is strictly concave. 
	\item Noting that $Y(\beta,\tau)\leq 0$, it can be easily seen that:
	\begin{align}
		\lim_{\beta\to\infty}\min_{\tau\geq 0} \mathcal{D}(\beta,\tau) &\leq \lim_{\beta\to\infty}-\frac{\beta^2}{4} + \min_{\tau\geq 0}\frac{\tau\beta\delta}{2}+\frac{\rho\beta}{2\tau}.
	\end{align}
	Using the fact that $\min_{\tau\geq 0} \frac{\tau\beta\delta}{2}+\frac{\rho\beta}{2\tau}=\beta\sqrt{\rho\delta}$, we thus get:
\begin{equation*}
	\lim_{\beta\to\infty}\min_{\tau\geq 0} \mathcal{D}(\beta,\tau)=-\infty.
\end{equation*}
	This implies that ${\beta}^\star$ is bounded. To prove that this maximum is bounded below zero, we exploit the fact that:
	\begin{equation}
		\lim_{\beta\to 0} \inf_{\tau\geq 0} \frac{\partial \mathcal{D}(\beta,\tau)}{\partial \beta} =\inf_{\tau\geq 0}\frac{\tau\delta}{2}+\frac{\rho}{2\tau}>0,
		\label{eq:true}
	\end{equation}
	which follows by noticing that $\lim_{\beta\to 0}\inf_{\tau\geq 0}\frac{\partial Y(\beta,\tau)}{\partial \beta}=0$. 
	Indeed, due to \eqref{eq:true}, for any $\beta$ in the vicinity of zero, and all $\tau>0$, 
\begin{equation*}
	\beta \frac{\partial \mathcal{D}(\beta,\tau)}{\partial \beta} \geq 0.
\end{equation*}
	Hence $\beta=0$ does not satisfy the first order optimality condition and hence could not be the one that maximizes $\inf_{\tau\geq 0} \frac{\partial \mathcal{D}}{\partial \beta}$.  
	
	Next, we prove that $\tau^\star$ is bounded. For that, we consider separately the cases $\lambda>0$ and $\left\{\lambda=0,\delta>0\right\}$. 

	\noindent{\underline{Case 1: $\lambda>0$}}. Note that since $\mathcal{D}(\beta,\tau)$ is concave in $\beta$ and convex in $\tau$, 
	\begin{equation*}
	{\tau}^\star=\arg\min_{\tau\geq 0} \mathcal{D}({\beta}^\star,\tau).
	\end{equation*}
	As ${\beta}^\star\neq 0$, we can easily check that for $\lambda>0$, 
	\begin{equation*}
	\lim_{\tau\to \infty} \mathcal{D}({\beta}^\star,\tau) =\infty.
	\end{equation*}
	Hence, necessarily ${\tau}^\star$ is bounded.
 
	\noindent{\underline{Case 2: $\lambda=0$ and $\delta > 1$ }}
	In this case, $\alpha=\frac{1}{\tau}$ and $Y(\beta,\tau)$ simplifies to:
	\begin{align}
		Y(\beta,\tau)&=\beta\sqrt{P} \mathbb{E}\left[\left(\frac{\sqrt{P}}{\tau}-H\right){\bf 1}_{\{H\geq \frac{\sqrt{P}}{\tau}\}}\right] -\frac{\beta\tau}{2} \mathbb{E}\left[H^2{\bf 1}_{\{-\frac{\sqrt{P}}{\tau}\leq H\leq \frac{\sqrt{P}}{\tau}\}}\right].
	\end{align}
	If $\delta>1$, we can obviously see that:
	\begin{equation}
\lim_{\tau\to\infty} \mathcal{D}(\beta^\star,\tau)= \lim_{\tau\to\infty} \frac{\beta^\star\tau}{2}\left(\delta- \mathbb{E}\left[H^2{\bf 1}_{\{-\frac{\sqrt{P}}{\tau}\leq H\leq \frac{\sqrt{P}}{\tau}\}}\right] \right)=\infty.
	\end{equation}
	which again shows that $\tau^\star$ is bounded. 
	\item Equations \eqref{eq:tau_r} and \eqref{eq:tau_sol} follow directly by writing the first order optimality condition for the variable $\tau$.
	\item When $\lambda=0$ and $\delta>1$, the asymptotic deterministic max-min problem in \eqref{eq:asde1} simplifies to:
\begin{equation*}
	\overline{\phi}=\min_{\tau\geq 0}\max_{\beta\geq 0} \frac{\tau\beta\delta}{2}+\frac{\rho\beta}{2\tau}-\frac{\beta^2}{4}+\beta\tilde{Y}(\tau),
\end{equation*}
	where $\tilde{Y}(\tau)$ is given by \eqref{eq:Y_tilde_tau}. Taking the maximum over $\beta$ yields the desired. 
\end{enumerate}
\section{Proof of Lemma \ref{lem:technical_1} and Lemma \ref{lem:technical_2}}
\label{sec7}
\subsection{Proof of Lemma \ref{lem:technical_1}}
\label{pl2}

\noindent{\bf Asymptotic equivalent for the AO.} Recall the function $\mathcal{L}_{\lambda,\rho}({\bf x})$ defined in \eqref{eq:L}.  We note that the variable ${\bf u}$ appears in the objective of \eqref{eq:L} through a linear term and through its magnitude, which suggests that one can first optimize over its direction for fixed amplitude. Formally, we proceed as follows. Let $\beta=\|{\bf u}\|_2$. The contribution of the terms involving ${\bf u}$ in \eqref{eq:L} can be expressed as:
\begin{equation*}
{\bf u}^{T}\left(\frac{1}{n}\|{\bf x}\|{\bf g}-\sqrt{\frac{\rho}{n}}{\bf s}\right)-\frac{\|{\bf u}\|^2}{4}-\frac{1}{n}\|\mathbf{u}\|\mathbf{h}^T\mathbf{x}.
\end{equation*}
Obviously, the direction of ${\bf u}$ that optimizes the above expression is the one that aligns with $\left(\frac{1}{n}\|{\bf x}\|{\bf g}-\sqrt{\frac{\rho}{n}}{\bf s}\right)$. Hence, \eqref{eq:L} simplifies as:
\begin{align}
	\mathcal{L}_{\lambda,\rho}({\bf x})&=\max_{C_\beta\geq\beta\geq 0}\ell_{\lambda,\rho}({\bf x},\beta)
	\end{align}  
with $\ell_{\lambda,\rho}({\bf x},\beta)$ defined as:
\begin{align}
\ell_{\lambda,\rho}({\bf x},\beta)&=	\frac{\beta\|{\bf g}\|}{\sqrt{n}}  \left\|\frac{\|{\bf x}\|{\bf g}}{\sqrt{n}\|{\bf g}\|}-\frac{\sqrt{\rho}{\bf s}}{\|{\bf g}\|}\right\| -\frac{\beta}{n}{\bf h}^{T}{\bf x}-\frac{\beta^2}{4}+\frac{\lambda\|{\bf x}\|^2}{n}.
\end{align}
To continue, we consider proving that 
\begin{equation}
\sup_{\substack{{\bf x}\\ x_i^2\leq P}} \left|\left\|\frac{\|{\bf x}\|{\bf g}}{\sqrt{n}\|{\bf g}\|}-\frac{\sqrt{\rho}{\bf s}}{\|{\bf g}\|}\right\|-\sqrt{\frac{\|{\bf x}\|^2}{n}+\frac{\rho m}{\|{\bf g}\|^2}}\right|\to 0.\label{eq:redc}
\end{equation}
For that, we use the relation $|\sqrt{a}-\sqrt{b}|=\frac{|a-b|}{\sqrt{a}+\sqrt{b}}$ which holds for any  positive $a$ and $b$ with $\min(a,b)>0$ to get:
\begin{align}
\left|\left\|\frac{\|{\bf x}\|{\bf g}}{\sqrt{n}\|{\bf g}\|}-\frac{\sqrt{\rho}{\bf s}}{\|{\bf g}\|}\right\|-\sqrt{\frac{\|{\bf x}\|^2}{n}+\frac{\rho m}{\|{\bf g}\|^2}}\right| 
= \frac{\left|\left\|\frac{\|{\bf x}\|{\bf g}}{\sqrt{n}\|{\bf g}\|}-\frac{\sqrt{\rho}{\bf s}}{\|{\bf g}\|}\right\|^2-\left(\frac{\|{\bf x}\|^2}{n}+\frac{\rho m }{\|{\bf g}\|^2}\right)\right|}{\left\|\frac{\|{\bf x}\|{\bf g}}{\sqrt{n}\|{\bf g}\|}-\frac{\sqrt{\rho}{\bf s}}{\|{\bf g}\|}\right\|+\sqrt{\frac{\|{\bf x}\|^2}{n}+\frac{\rho m}{\|{\bf g}\|^2}}}
\leq \frac{\left|\frac{2\sqrt{\rho}{\bf s}^{T}{\bf g}\|{\bf x}\|}{\sqrt{n}\|{\bf g}\|^2}\right|}{\frac{\sqrt{\rho m}}{\|{\bf g}\|}}.\label{eq:der}
\end{align}
It follows from the weak law of large numbers that:
\begin{equation*}
\frac{{\bf s}^{T}{\bf g}}{\|{\bf g}\|^2}\to 0.
\end{equation*}
Using the above convergence together with \eqref{eq:der}, we thus prove \eqref{eq:redc}.

Define function $x\mapsto \hat{\mathcal{L}}_{\lambda,\rho}({\bf x})$ as:
\begin{align}
	\hat{\mathcal{L}}_{\lambda,\rho}({\bf x})&= \max_{C_\beta\geq\beta\geq 0} 
	\hat{\ell}_{\lambda,\rho}({\bf x},\beta)\label{eq:L_hat_def}
	\end{align} 
with $\hat{\ell}_{\lambda,\rho}({\bf x},\beta)$ given by:
\begin{align}
\hat{\ell}_{\lambda,\rho}({\bf x},\beta):=&	\frac{\beta\|{\bf g}\|}{\sqrt{n}}  \sqrt{\frac{\|{\bf x}\|^2}{n}+\frac{\rho m}{\|{\bf g}\|^2}} -\frac{\beta}{n}{\bf h}^{T}{\bf x}-\frac{\beta^2}{4}+\frac{\lambda\|{\bf x}\|^2}{n}. \label{eq:l_def}
\end{align}
It follows from \eqref{eq:redc} that:
\begin{equation*}
\sup_{0\leq \beta\leq C_\beta}\sup_{\substack{{\bf x}\\ x_i^2\leq P}} \left|\ell_{\lambda,\rho}({\bf x},\beta)-\hat{\ell}_{\lambda,\rho}({\bf x},\beta)\right|\to 0. 
\end{equation*}
Hence, for any $\epsilon>0$, with probability approaching $1$, for all feasible ${\bf x}$ and $\beta$, it holds that:
\begin{equation*}
\hat{\ell}_{\lambda,\rho}({\bf x},\beta)-\epsilon \leq \ell_{\lambda,\rho}({\bf x},\beta)\leq \hat{\ell}_{\lambda,\rho}({\bf x},\beta)+\epsilon.
\end{equation*}
Taking the supremum in the above inequalities yields,
\begin{equation*}
\hat{\mathcal{L}}_{\lambda,\rho}({\bf x}) -\epsilon\leq \mathcal{L}_{\lambda,\rho}({\bf x})\leq \hat{\mathcal{L}}_{\lambda,\rho}({\bf x}) +\epsilon.
\end{equation*}
Since $\epsilon$ is taken independently of ${\bf x}$ and $\beta$, we thus obtain:
\begin{equation}
\sup_{\substack{{\bf x}\\ x_i^2\leq P}} \left|\hat{\mathcal{L}}_{\lambda,\rho}({\bf x})-\mathcal{L}_{\lambda,\rho}({\bf x})\right|\to 0. \label{eq:uniform}
\end{equation}
Associated with $\hat{\mathcal{L}}_{\lambda,\rho}({\bf x})$, we define the following asymptotic equivalent AO  given by:
\begin{equation*}
\hat{\phi}_{\lambda,\rho}({\bf g},{\bf h})=\min_{\substack{{\bf x}\\ x_i^2\leq P}} \hat{\mathcal{L}}_{\lambda,\rho}({\bf x}).
\end{equation*}
It follows from the uniform convergence in \eqref{eq:uniform} that:
\begin{equation}
\phi_{\lambda,\rho}({\bf g},{\bf h})-\hat{\phi}_{\lambda,\rho}({\bf g},{\bf h})\to 0. \label{eq:equivalent}
\end{equation}

\noindent{\bf Simplification  of $\hat{\phi}_{\lambda,\rho}({\bf g},{\bf h})$.} 
 To further simplify $\hat{\mathcal{L}}_{\lambda,\rho}({\bf x})$, we use the following variational expression for the square-root term $\sqrt{\frac{\|{\bf x}\|^2}{n}+\frac{\rho m}{\|{\bf g}\|^2}} $ ,
\begin{equation}
\sqrt{\frac{\|{\bf x}\|^2}{n}+\frac{\rho m}{\|{\bf g}\|^2}}=\min_{\tau\geq 0} \frac{\tau}{2}+ \frac{\frac{\|{\bf x}\|^2}{n}+\frac{\rho m}{\|{\bf g}\|^2}}{2\tau}. \label{eq:last}
\end{equation}
At optimum, the optimal $\tau$ satisfies: $\overline{\tau}=\sqrt{\frac{\|{\bf x}\|^2}{n}+\frac{\rho m}{\|{\bf g}\|^2}}$. Since  with probability approaching $1$, $\frac{m}{\|{\bf g}\|^2}=1$, the value of the optimal $\overline{\tau}$ is larger than  $\frac{\sqrt{\rho}}{2}$. Similarly, as $\frac{\|{\bf x}\|^2}{n}\leq P$, the value of $\overline{\tau}$ is less than $\sqrt{2P+2\rho}$. Hence, nothing changes in \eqref{eq:last}, if we further constrain $\tau$ to lie in the interval $[\frac{\sqrt{\rho}}{2},C_\tau]$ where $C_\tau$ which can be set to any value fixed value larger than $\sqrt{2P+2\rho}$. 
This leads to the following equivalent formulation for $\hat{\phi}_{\lambda,\rho}({\bf g},{\bf h})$:
\begin{align}
\hat{\phi}_{\lambda,\rho}({\bf g},{\bf h})=\min_{\substack{{\bf x}\\ x_i^2\leq P}}\max_{0\leq \beta\leq C_\beta}\min_{\frac{\sqrt{\rho}}{2}\leq \tau\leq C_\tau} \frac{\tau\beta\|{\bf g}\|}{2\sqrt{n}} +\frac{\beta \|{\bf g}\|}{2\tau\sqrt{n}}\left(
\frac{\|{\bf x}\|^2}{n}+\frac{\rho m}{\|{\bf g}\|^2}\right) -\frac{\beta}{n}{\bf h}^{T}{\bf x} -\frac{\beta^2}{4}+\lambda\frac{\|{\bf x}\|^2}{n}. \label{eq:AO_eq0}
\end{align}
For convenience, we perform the change of variable $\tau\leftrightarrow \tau\sqrt{\delta}$, which leads to:
\begin{align}
&\hat{\phi}_{\lambda,\rho}({\bf g},{\bf h})=\min_{\substack{{\bf x}\\ x_i^2\leq P}}\max_{0\leq \beta\leq C_\beta}\min_{\frac{\sqrt{\rho}}{2\sqrt{\delta}}\leq \tau\leq \frac{C_\tau}{\sqrt{\delta}}} \frac{\tau\beta\sqrt{\delta}\|{\bf g}\|}{2\sqrt{ n}} \nonumber\\
&+\frac{\beta \|{\bf g}\|}{2\tau\sqrt{\delta n}}\left(
\frac{\|{\bf x}\|^2}{n}+\frac{\rho m}{\|{\bf g}\|^2}\right) -\frac{\beta}{n}{\bf h}^{T}{\bf x} -\frac{\beta^2}{4}+\lambda\frac{\|{\bf x}\|^2}{n}. \label{eq:AO_eq}
\end{align}
It can be checked by studying the hessian matrix  that the objective function in \eqref{eq:AO_eq} is jointly-convex in $({\bf x},\tau)$ and concave in $\beta$. Hence, we may use the Sion's min-max theorem to permute the min-max and find 
\begin{align}
&	\hat{\phi}_{\lambda,\rho}({\bf g},{\bf h})=\max_{0\leq \beta\leq C_\beta}\min_{\frac{\sqrt{\rho}}{2\sqrt{\delta}}\leq \tau\leq \frac{C_\tau}{\sqrt{\delta}}} \min_{\substack{{\bf x}\\ x_i^2\leq P}} \frac{\tau\beta\sqrt{\delta}\|{\bf g}\|}{2\sqrt{n}}\nonumber\\
	&+\frac{\beta \|{\bf g}\|}{2\tau\sqrt{\delta n}}\left(
	\frac{\|{\bf x}\|^2}{n}+\frac{\rho m}{\|{\bf g}\|^2}\right) -\frac{\beta}{n}{\bf h}^{T}{\bf x} -\frac{\beta^2}{4}+\lambda\frac{\|{\bf x}\|^2}{n}.
\end{align}
For fixed $\beta$ and $\tau$, the minimization over ${\bf x}$ reduces to solving the following separable optimization problem:
\begin{equation}
\min_{\substack{{\bf x}\\ -\sqrt{P}\leq x_i\leq \sqrt{P}}} \frac{1}{n}\sum_{i=1}^n \left(\frac{\beta\|{\bf g}\|}{2\tau\sqrt{\delta n}}+\lambda\right)x_i^2-\beta h_ix_i . \label{eq:to_minimize}
\end{equation}
Solving \eqref{eq:to_minimize} for fixed $\beta$ and $\tau$, the elements of the optimal ${\bf x}'={\bf x}'(\tau,\beta)$ are given by:
\begin{equation*}
x_i^{'}=\left\{
\begin{array}{ll}
	-\sqrt{P} & \ \ \mathrm{if } \ \ h_i<-\sqrt{P}\tilde{\alpha}\\
	\frac{h_i}{\tilde{\alpha}} & \ \ \mathrm{if } -\sqrt{P}\tilde{\alpha}\leq h_i\leq \sqrt{P}\tilde{\alpha}\\
	\sqrt{P} & \ \ \mathrm{if } \ \ h_i\geq \sqrt{P}\tilde{\alpha}
\end{array},
\right.
\end{equation*}
where $\tilde{\alpha}=\frac{\|{\bf g}\|}{\tau\sqrt{\delta n}}+\frac{2\lambda}{\beta}$.
Replacing ${\bf x}$ by its optimal value in \eqref{eq:to_minimize} yields:
\begin{equation*}
\min_{\substack{{\bf x}\\ -\sqrt{P}\leq x_i\leq \sqrt{P}}} \frac{1}{n}\sum_{i=1}^n \left(\frac{\beta\|{\bf g}\|}{2\tau\sqrt{\delta n}}+\lambda\right)x_i^2-\beta h_ix_i=\frac{1}{n}\sum_{i=1}^n \hat{v}_i,
\end{equation*}
where $\hat{v}_i$ is given by:
\begin{equation*}
\hat{v}_i:=\left\{
\begin{array}{ll}
\frac{\beta P \|{\bf g}\|}{2\tau\sqrt{\delta n}} +\lambda P +\beta h_i\sqrt{P} & \mathrm{if } h_i\leq -\sqrt{P}\tilde{\alpha}\\
-\frac{\beta h_i^2}{2\tilde{\alpha}} & \mathrm{if } -\sqrt{P}\tilde{\alpha} \leq h_i\leq \sqrt{P}\tilde{\alpha}\\
\frac{\beta P\|{\bf g}\|}{2\tau\sqrt{\delta n}}+\lambda P -\beta h_i\sqrt{P} & \mathrm{if } h_i\geq \sqrt{P}\tilde{\alpha}
\end{array}.
\right.
\end{equation*}
With this, we can express $\hat{\phi}_{\lambda,\rho}({\bf g},{\bf h})$ as:
\begin{align}
\hat{\phi}_{\lambda,\rho}({\bf g},{\bf h})&=\max_{0\leq \beta\leq C_\beta}\min_{\frac{\sqrt{\rho}}{2\sqrt{\delta}}\leq \tau \leq \frac{C_\tau}{\sqrt{\delta}}} \frac{\tau\beta\sqrt{\delta}\|{\bf g}\|}{2\sqrt{n}} +\frac{\beta\rho m}{2\tau\sqrt{\delta n}\|{\bf g}\|}-\frac{\beta^2}{4}+\frac{1}{n}\sum_{i=1}^n v_i. \label{eq:sim}
\end{align}
Let $(\hat{\beta}_n,\hat{\tau}_n)$ be the saddle point in the above optimization problem. Since, the objective function is strictly convex in $\tau$ and strictly concave in $\beta$ over the domain $\left\{(\beta,\tau)\in [0,C_\beta]\times [\frac{\sqrt{\rho}}{2\sqrt{\delta}},\frac{C_\tau}{\sqrt{\delta}}]\right\}$, this saddle point is unique. Specifically, this implies that function $\hat{\mathcal{L}}$ admits a unique minimizer $\tilde{\bf x}^{\rm AO}$ which is given by:
\begin{equation*}
\tilde{\bf x}_i^{\rm AO}= \left\{ \begin{array}{ll}
-\sqrt{P} & \mathrm{if } h_i\leq -\hat{\alpha}\sqrt{P}\\
\frac{h_i}{{\hat{\alpha }}}  &\mathrm{if }  -\hat{\alpha}\sqrt{P}\leq h_i\leq \hat{\alpha}\sqrt{P} \\
\sqrt{P} & \mathrm{if } h_i\geq \hat{\alpha}\sqrt{P}
\end{array},
\right.
\end{equation*}
where $\hat{\alpha}=\frac{\|{\bf g}\|}{\hat{\tau}_n\sqrt{\delta n}} + \frac{2\lambda}{\hat{\beta}_n}$.

\noindent{\bf Asymptotic limits of $(\hat{\beta}_n,\hat{\tau}_n)$ and $\hat{\phi}_{\lambda,\rho}({\bf g},{\bf h})$.}

Call $\mathcal{R}(\beta,\tau)$ the objective in \eqref{eq:sim}. 
From the weak law of large numbers, it can be readily seen that:
\begin{equation}
\left\{\frac{\tau\beta\sqrt{\delta}\|{\bf g}\|}{2\sqrt{n}}+\frac{\beta\rho m}{2\tau\sqrt{\delta n}\|{\bf g}\|}\right\} -\frac{\beta^2}{4} \to  \frac{\tau\beta\delta}{2}+\frac{\beta \rho}{2\tau}-\frac{\beta^2}{4}.\label{eq:conv1}
\end{equation}
Using the same argument, it can be easily checked that:
\begin{align}
\frac{1}{n}\sum_{i=1}^n \hat{v_i} &\to Y(\beta,\tau):=\beta \sqrt{P}\mathbb{E}\left[\left(\sqrt{P}\alpha-2H\right){\bf 1}_{\{H\geq\sqrt{P}\alpha\}}\right] -\frac{\beta}{2\alpha}\mathbb{E}\left[H^2{\bf 1}_{\{-\sqrt{P}\alpha\leq H\leq \sqrt{P}\alpha\}}\right] ,\label{eq:conv2}
\end{align}
where $\alpha=\frac{1}{\tau}+\frac{2\lambda}{\beta}$ and the expectation is taken with respect to the distribution of the standard normal variable $H$. 

Putting \eqref{eq:conv1} and \eqref{eq:conv2} together, the objective $\mathcal{R}(\beta,\tau)$ converges pointwise to 
\begin{equation*}
\mathcal{D}(\beta,\tau):= \frac{\tau\beta \delta}{2} +\frac{\rho\beta}{2\tau}-\frac{\beta^2}{4}+Y(\beta,\tau).
\end{equation*}
The objective $\mathcal{R}(\beta,\tau)$ is concave in $\beta$ and convex in $\tau$. The convergence is thus uniform over compacts and we thus obtain:
\begin{equation}
\hat{\phi}_{\lambda,\rho}({\bf g},{\bf h})\to \max_{0\leq \beta\leq C_\beta}\min_{\frac{\sqrt{\rho}}{2\sqrt{\delta}}\leq \tau \leq \frac{C_\tau}{\sqrt{\delta}}}\mathcal{D}(\beta,\tau). \label{eq:asym}
\end{equation}
Furthermore, using the strict convexity and concavity of $\mathcal{D}(\beta,\tau)$ on $\left\{(\beta,\tau)\in\mathbb{R}_{+}\times\mathbb{R}_{+} \ |\ 0\leq \beta\leq C_\beta, \frac{\sqrt{\rho}}{2\sqrt{\delta}}\leq \tau\leq \frac{C_\tau}{\sqrt{\delta}}  \right\}$, we also have
\begin{equation*}
(\hat{\beta}_n,\hat{\tau}_n) \to (\overline{\beta},\overline{\tau}):=\arg\max_{0\leq \beta \leq C_\beta} \min_{\frac{\sqrt{\rho}}{2\sqrt{\delta}}\leq \tau \leq \frac{C_\tau}{\sqrt{\delta}}} \mathcal{D}(\beta,\tau).
\end{equation*}
\noindent{\bf Unbounded asymptotic AO.}
Consider the unbounded asymptotic optimization problem:
\begin{equation*}
\overline{\phi}:=\max_{\beta\geq 0}\min_{\tau\geq 0} \mathcal{D}(\beta,\tau).
\end{equation*}
From Lemma \ref{eq:lemma_asym}, the above optimization problem possesses a unique finite saddle point $(\beta^\star,\tau^\star)$. For all $\beta\geq 0$, it takes no much effort to check that $\tau\mapsto \mathcal{D}(\tau,\beta)$ is decreasing on the interval $[0,\sqrt{\frac{\rho}{\delta}}]$. Hence, $\tau^\star\geq \sqrt{\frac{\rho}{\delta}}$. From this, we conclude that by choosing $C_\beta\geq 2\beta^\star$ and $C_\tau\geq 2\sqrt{\delta}\tau^\star$, we have $\overline{\beta}=\beta^\star$ and $\overline{\tau}=\tau^\star$, and consequently,
\begin{equation*}
\overline{\phi}=\max_{C_\beta\geq \beta\geq 0}\min_{\frac{\sqrt{\rho}}{2\sqrt{\delta}}\leq \tau\leq \frac{C_\tau}{\sqrt{\delta}}} \mathcal{D}(\beta,\tau).
\end{equation*}
Hence, the convergence in \eqref{eq:asym} can be equivalently expressed as:
\begin{equation}
\hat{\phi}_{\lambda,\rho}({\bf g},{\bf h})\to \overline{\phi}. \label{eq:uniform_conv}
\end{equation}

\noindent{\bf On the behavior of function $\hat{\mathcal{L}}({\bf x})$.}
For $\lambda>0$, it can be readily seen that function $\hat{\mathcal{L}}$ is $\frac{\lambda}{n}-$strongly convex. 
The objective here is to prove that when $\lambda=0$ and $\delta>1$, function $\hat{\mathcal{L}}$ is $\frac{a}{n}-$ strongly convex on a neighborhood of $\overline{\bf x}^{\rm AO}$ for some $a>0$.

Starting from \eqref{eq:L_hat_def} and \eqref{eq:l_def}, for $\lambda=0$, $\hat{\mathcal{L}}$ simplifies to:
\begin{equation*}
\hat{\mathcal{L}}({\bf x})=\max_{0\leq \beta\leq C_\beta} \beta a({\bf x})-\frac{\beta^2}{4}
\end{equation*}
where $a({\bf x})$ is given by:
\begin{equation}
a({\bf x}):=\frac{\|{\bf g}\|}{\sqrt{n}} \sqrt{\frac{\|{\bf x}\|^2}{n}+\frac{\rho m}{\|{\bf g}\|^2}} -\frac{1}{n}{\bf h}^{T}{\bf x}. \label{eq:ax}
\end{equation}
To prove that $\hat{\mathcal{L}}$ is strongly convex on a neighborhood of $\overline{\bf x}^{\rm AO}$, it suffices to check that there exists $c>0$ such that with probability approaching $1$,
\begin{align}
a(\overline{\bf x}^{\rm AO})>c. \label{eq:req}
\end{align}
Indeed, if \eqref{eq:req} holds true, then we can find a ball $\mathcal{B}$ centered at $\overline{\bf x}^{\rm AO}$ with radius $r\sqrt{n}$ such that
\begin{equation*}
a({\bf x})\geq \frac{c}{2}, \ \ \forall {\bf x}\in\mathcal{B}.
\end{equation*}
Since $a({\bf x})$ is positive for ${\bf x}\in\mathcal{B}$, optimizing $\hat{\mathcal{L}}$ with respect to  $\beta\in[0,C_\beta]$ yields \footnote{Here we used the fact that $a({\bf x})$ is bounded in probability on $\mathcal{B}$ so that $C_\beta$ can be chosen larger than a fixed upper bound of $\left\{a({\bf x}), {\bf x}\in\mathcal{B}\right\}$ }:
\begin{equation*}
\hat{\mathcal{L}}({\bf x})=\left(a({\bf x})\right)^2, \ \ \forall {\bf x}\in\mathcal{B}.
\end{equation*}
By Lemma F.14 in \cite{Miolane}, function $x\mapsto a({\bf x})$ is $\frac{\gamma}{n}$-strongly convex on the ball $\mathcal{B}$ for some $\gamma>0$. Particularly the hessian of ${\bf a}$ satisfies:
\begin{equation*}
\nabla^2{\bf a}\geq \frac{\gamma}{n} {\bf I}_{n}.
\end{equation*}
Computing the Hessian of $\hat{\mathcal{L}}$, we obtain:
\begin{equation*}
\nabla^2a({\bf x})=2a({\bf x})\nabla^2a({\bf x})+\nabla a({\bf x})\nabla a({\bf x}) \geq c\frac{\gamma}{n}{\bf I}_n, \ \ \forall {\bf x}\in\mathcal{B}. 
\end{equation*}
This proves that $\hat{\mathcal{L}}$ is $\frac{c\gamma}{n}$-strongly convex on a neighborhood of $\overline{\bf x}^{\rm AO}$. It remains thus to show \eqref{eq:req}. 

\noindent{\underline{Proof of \eqref{eq:req}}.} From the weak law of large numbers, the following convergences hold true:
\begin{equation*}
\frac{1}{n}\|\overline{\bf x}^{\rm AO}\|^2\to 2P \mathbb{P}\Big[H\geq\frac{ \sqrt{P}}{\tau^\star}\Big] + (\tau^\star)^2\mathbb{E}\left[H^2{\bf 1}_{\{-\frac{\sqrt{P}}{\tau^\star}\leq H\leq \frac{\sqrt{P}}{\tau^\star}\}}\right]
\end{equation*}
and 
\begin{equation*}
\frac{1}{n}{\bf h}^{T}\overline{\bf x}^{\rm AO}\to 2\sqrt{P}\mathbb{E}\left[H{\bf 1}_{\{H\geq \frac{\sqrt{P}}{\tau^\star}\}}\right] + \tau^\star\mathbb{E}\left[H^2{\bf 1}_{\{-\frac{\sqrt{P}}{\tau^\star}\leq H\leq \frac{\sqrt{P}}{\tau^\star}\}}\right].
\end{equation*}
Using the first order optimality condition in \eqref{eq:tau_sol}, we can write the first convergence as:
\begin{equation*}
\frac{1}{n}\|\overline{\bf x}^{\rm AO}\|^2 \to \delta(\tau^\star)^2-\rho. 
\end{equation*}
Next, we use the facts that $\frac{\|{\bf g}\|}{\sqrt{n}}\to \sqrt{\delta}$ and $\frac{\rho m}{\|{\bf g}\|^2}\to \rho$ to obtain:
\begin{equation}
\frac{\|{\bf g}\|}{\sqrt{n}}\sqrt{\frac{\|\overline{\bf x}^{\rm AO}\|^2}{n}+\frac{\rho m}{\|{\bf g}\|^2}}\to \delta\tau^\star. \label{eq:ref2}
\end{equation}
Using again \eqref{eq:tau_sol}, we can easily check the following equality
\begin{equation}
\delta\tau^\star=\frac{\rho}{\tau^\star} +\frac{2P}{\tau^\star} \mathbb{P}\left[H\geq \frac{\sqrt{P}}{\tau^\star}\right]+\tau^\star\mathbb{E}\left[H^2{\bf 1}_{\{-\frac{\sqrt{P}}{\tau^\star}\leq H\leq \frac{\sqrt{P}}{\tau^\star}\}}\right]. \label{eq:ref1}
\end{equation}
Using \eqref{eq:ref1} and \eqref{eq:ref2}, we thus obtain:
\begin{equation*}
a(\overline{\bf x}^{\rm AO})\to \overline{a},
\end{equation*}
where
\begin{equation}
\overline{a}=\frac{\rho}{\tau^\star}+\frac{2P}{\tau^\star} \mathbb{P}\left[H\geq \frac{\sqrt{P}}{\tau^\star}\right]-2\sqrt{P}\mathbb{E}\left[H{\bf 1}_{\{H\geq \frac{\sqrt{P}}{\tau^\star}\}}\right]. \label{eq:overline_a}
\end{equation}
It follows from \eqref{eq:ref1} that:
\begin{equation*}
\frac{\rho}{\tau^\star}=\delta\tau^\star-\frac{2P}{\tau^\star} \mathbb{P}\left[H\geq \frac{\sqrt{P}}{\tau^\star}\right]-\tau^\star\mathbb{E}\left[H^2{\bf 1}_{\{-\frac{\sqrt{P}}{\tau^\star}\leq H\leq \frac{\sqrt{P}}{\tau^\star}\}}\right].
\end{equation*}
Plugging the above relation into \eqref{eq:overline_a} yields:
\begin{equation*}
\overline{a}=\delta \tau^\star - \tau^\star\mathbb{E}\left[H^2{\bf 1}_{\{-\frac{\sqrt{P}}{\tau^\star}\leq H\leq \frac{\sqrt{P}}{\tau^\star}\}}\right]-2\sqrt{P}\mathbb{E}\left[H{\bf 1}_{\{H\geq \frac{\sqrt{P}}{\tau^\star}\}}\right].
\end{equation*}
To continue, we use the fact that
\begin{equation*}
{\bf 1}_{\{-\frac{\sqrt{P}}{\tau^\star}\leq H\leq \frac{\sqrt{P}}{\tau^\star}\}}= 1- \left({\bf 1}_{\{H\geq \frac{\sqrt{P}}{\tau^\star}\}}+{\bf 1}_{\{H\leq -\frac{\sqrt{P}}{\tau^\star}\}}\right)
\end{equation*}
to obtain:
\begin{equation*}
\overline{a}=\delta\tau^\star-\tau^\star+2\tau^\star \mathbb{E}\left[H^2{\bf 1}_{\{H\geq \frac{\sqrt{P}}{\tau^\star}\}}\right]-2\sqrt{P}\mathbb{E}\left[H {\bf 1}_{\{H\geq \frac{\sqrt{P}}{\tau^\star}\}}\right].
\end{equation*}
Finally, using the fact that
\begin{equation*}
H^2{\bf 1}_{\{H\geq \frac{\sqrt{P}}{\tau^\star}\}}\geq \frac{\sqrt{P}H}{\tau^\star}{\bf 1}_{\{H\geq \frac{\sqrt{P}}{\tau^\star}\}}
\end{equation*}
we can easily check that:
\begin{equation*}
2\tau^\star \mathbb{E}\left[H^2{\bf 1}_{\{H\geq \frac{\sqrt{P}}{\tau^\star}\}}\right]-2\sqrt{P}\mathbb{E}\left[H {\bf 1}_{\{H\geq \frac{\sqrt{P}}{\tau^\star}\}}\right]\geq 0.
\end{equation*}
Since $\delta>1$, $\delta\tau^\star-\tau^\star \geq 0$. We thus conclude that $\overline{a}> 0$, which shows \eqref{eq:req}. 

\noindent{\bf Putting all things together.} The proof of the first item of Lemma \ref{lem:technical_1} follows directly from the above analysis. Indeed, considering function $\hat{\mathcal{L}}_{\lambda,\rho}$ in \eqref{eq:L_hat_def}, we proved in  \eqref{eq:uniform} that:
\begin{equation*}
\sup_{{\bf x}} \left|\hat{\mathcal{L}}_{\lambda,\rho}({\bf x}) - \mathcal{L}_{\lambda,\rho}({\bf x})\right|\to 0. 
\end{equation*}
Based on \eqref{eq:uniform_conv}, we have:
\begin{equation*}
\min_{\substack{{\bf x}\\ x_i^2\leq P}} \hat{\mathcal{L}}_{\lambda,\rho}({\bf x}) \to\overline{\phi}.
\end{equation*}
To prove the second item in Lemma \ref{lem:technical_1}, we let $\hat{t}=\min(\sqrt{P}\alpha^\star,\sqrt{P}\hat{\alpha})$ and use the fact that:
\begin{align*}
&\left|\left[\tilde{\bf x}^{\rm AO}\right]_i-\left[\overline{\bf x}^{\rm AO}\right]_i\right|\\ \leq \max&\left(|h_i|\left|\left(\frac{1}{\hat{\alpha}}-\frac{1}{\alpha^\star}\right){\bf 1}_{\{-t\leq h_i\leq t\}}\right|,\right. \left|\frac{h_i}{\alpha^\star}-\sqrt{P}\right|{\bf 1}_{\{\sqrt{P}\hat{\alpha}\leq h_i\leq \sqrt{P}\alpha^\star\}},  \left.\left|\frac{h_i}{\hat{\alpha}}+\sqrt{P}\right|{\bf 1}_{\{-\sqrt{P}{\hat{\alpha}}\leq h_i\leq -\sqrt{P}\alpha^\star\}}\right).
\end{align*}
Since $\hat{\alpha}-\alpha^\star$ converges to zero in probability, for any $\epsilon>0$, with probability approaching $1$, 
\begin{equation*}
\left|\hat{\alpha}-\alpha^\star\right|\leq \epsilon  
\end{equation*}
and 
\begin{equation*}
\left|\frac{1}{\hat{\alpha}}-\frac{1}{\alpha^\star}\right|\leq \epsilon. 
\end{equation*}
Hence, 
\begin{align*}
	\left|\left[\tilde{\bf x}^{\rm AO}\right]_i-\left[\overline{\bf x}^{\rm AO}\right]_i\right| \leq \max\left(\frac{\epsilon \hat{t}}{\hat{\alpha}\alpha^\star},\frac{\sqrt{P}\epsilon}{\alpha^\star},\frac{\sqrt{P}\epsilon}{\hat{\alpha}}\right)
	\leq \sqrt{P}\epsilon \left(1+\frac{\epsilon}{\alpha^\star}\right)\left(\frac{1}{\alpha^\star}+\epsilon\right).
\end{align*}
Choosing $\epsilon\leq \min(\frac{1}{2}\alpha^\star,1)$, we thus obtain
\begin{equation*}
\left|\left[\tilde{\bf x}^{\rm AO}\right]_i-\left[\overline{\bf x}^{\rm AO}\right]_i\right|\leq \frac{3}{2}\sqrt{P}(\frac{1}{\alpha^\star}+1)\epsilon .
\end{equation*}
This shows that:
\begin{equation}
\frac{1}{n}\left\|\tilde{\bf x}^{\rm AO}-\overline{\bf x}^{\rm AO}\right\|^2 \to 0. \label{eq:conv}
\end{equation}
Next, to prove \eqref{eq:toprove}, we recall that $\hat{\mathcal{L}}_{\lambda,\rho}$ writes as:
\begin{equation*}
\min_{{\bf x}}\hat{\mathcal{L}}_{\lambda,\rho}({\bf x})=\max_{C_{\beta}\geq \beta\geq 0} \beta {a}(\tilde{\bf x}^{\rm AO}) -\frac{\beta^2}{4}+\lambda \|\tilde{\bf x}^{\rm AO}\|^2,
\end{equation*}
where ${\bf x}\mapsto a({\bf x})$ is defined in \eqref{eq:ax}.
Let $$\hat{\hat{\beta}}=\arg\max_{C_\beta\geq \beta\geq 0} \beta a(\overline{\bf x}^{\rm AO})-\frac{\beta^2}{4}+\lambda\|\overline{\bf x}^{\rm AO}\|^2,$$
then, 
\begin{equation}
\min_{{\bf x}}\hat{\mathcal{L}}_{\lambda,\rho}({\bf x}) \geq \hat{\hat{\beta}} a(\tilde{\bf x}^{\rm AO}) -\hat{\hat{\beta}}+\lambda\|\tilde{\bf x}^{\rm AO}\|^2.\label{eq:red}
\end{equation}
With this we can use \eqref{eq:conv} to show that:
\begin{align}
\hat{\hat{\beta}} a(\tilde{\bf x}^{\rm AO}) -\hat{\hat{\beta}}+\lambda\|\tilde{\bf x}^{\rm AO}\|^2\geq \hat{\hat{\beta}} a(\overline{\bf x}^{\rm AO}) -\hat{\hat{\beta}}+\lambda\|\overline{\bf x}^{\rm AO}\|^2 -\epsilon=\hat{\mathcal{L}}_{\lambda,\rho}(\overline{\bf x}^{\rm AO})-\epsilon, 
\end{align}
and hence, in view of  \eqref{eq:red}, we obtain:
\begin{equation*}
\min_{{\bf x}}\hat{\mathcal{L}}_{\lambda,\rho}({\bf x}) \geq \hat{\mathcal{L}}_{\lambda,\rho}(\overline{\bf x}^{\rm AO})-\epsilon
\end{equation*}
which shows \eqref{eq:toprove}.
 
\subsection{Proof of Lemma \ref{lem:technical_2}}
\label{pl3}
\subsubsection{Proof of \eqref{eq:con}}
We start by proving that:
\begin{equation}
\tilde{\phi}_{\lambda,\rho}({\bf g},{\bf h})-\overline{\phi}\to 0. \label{eq:d}
\end{equation}
Recall that $\tilde{\phi}_{\lambda,\rho}({\bf g},{\bf h})$ is given by:
\begin{align}
\tilde{\phi}_{\lambda,\rho}({\bf g},{\bf h})&=\max_{{\bf u}\in\mathcal{S}_{{\bf u}}} \min_{x_i^2\leq P} \frac{1}{n}\|{\bf x}\|{\bf g}^{T}{\bf u}-\frac{1}{n}\|{\bf u}\|{\bf h}^{T}{\bf x} -\frac{\sqrt{\rho}{\bf u}^{T}{\bf s}}{\sqrt{n}} -\frac{\|{\bf u}\|^2}{4} +\frac{\lambda \|{\bf x}\|^2}{n},
\end{align}
we first note that:
\begin{equation}
\tilde{\phi}_{\lambda,\rho}({\bf g},{\bf h})\leq \phi_{\lambda,\rho}({\bf g},{\bf h}). \label{eq:1re}
\end{equation}
Since $\phi({\bf g},{\bf h})$ converges in probability to $\overline{\phi}$, for any $\eta>0$ with probability approaching $1$, 
\begin{equation*}
\phi_{\lambda,\rho}({\bf g},{\bf h})\leq \overline{\phi}+\eta
\end{equation*}
and thus in view of \eqref{eq:1re},
\begin{equation*}
\tilde{\phi}_{\lambda,\rho}({\bf g},{\bf h})\leq \phi_{\lambda,\rho}({\bf g},{\bf h})\leq \overline{\phi}+\eta.
\end{equation*}
To prove \eqref{eq:d}, it suffices to show that with probability approaching $1$
\begin{equation}
\tilde{\phi}_{\lambda,\rho}({\bf g},{\bf h})\geq \overline{\phi}-\eta.\label{eq:refq}
\end{equation}
For that, we note that necessarily at optimum ${\bf h}^{T}{\bf x}\geq 0$ because otherwise $-{\bf x}$ would lead to a higher cost. Hence, nothing would change if we write $\tilde{\phi}_{\lambda,\rho}({\bf g},{\bf h})$ as
\begin{align}
	\tilde{\phi}_{\lambda,\rho}({\bf g},{\bf h})&=\max_{{\bf u}\in\mathcal{S}_{\bf u}} \min_{\substack{x_i^2\leq P\\ {\bf h}^{T}{\bf x}\geq 0}} \frac{1}{n}\|{\bf x}\|{\bf g}^{T}{\bf u}-\frac{1}{n}\|{\bf u}\|{\bf h}^{T}{\bf x} -\frac{\sqrt{\rho}{\bf u}^{T}{\bf s}}{\sqrt{n}} -\frac{\|{\bf u}\|^2}{4} +\frac{\lambda \|{\bf x}\|^2}{n}. \label{eq:l}
\end{align}
Stating from \eqref{eq:l}, we can lower bound $\tilde{\phi}({\bf g},{\bf h})$ as
\begin{align}
\tilde{\phi}_{\lambda,\rho}({\bf g},{\bf h})&\geq \max_{\substack{{\bf u}\in\mathcal{S}_{{\bf u}}\\ {\bf g}^{T}{\bf u}\geq 0}}\min_{\substack{x_i^2\leq P\\ {\bf h}^{T}{\bf x}\geq 0}} \frac{1}{n}\|{\bf x}\|{\bf g}^{T}{\bf u}-\frac{1}{n}\|{\bf u}\|{\bf h}^{T}{\bf x} -\frac{\sqrt{\rho}{\bf u}^{T}{\bf s}}{\sqrt{n}} -\frac{\|{\bf u}\|^2}{4} +\frac{\lambda \|{\bf x}\|^2}{n}.
\end{align}
We can easily see that the objective function of the above problem is convex in ${\bf x}$ for any ${\bf u}$ satisfying ${\bf g}^{T}{\bf u}\geq 0$ and concave in ${\bf u}$ for any ${\bf x}$ such that ${\bf x}^{T}{\bf h}\geq 0$. We can thus flip the order of the max-min to find:
\begin{align}
	\tilde{\phi}({\bf g},{\bf h})&\geq \min_{\substack{x_i^2\leq P\\ {\bf h}^{T}{\bf x}\geq 0}} \max_{\substack{{\bf u}\in\mathcal{S}_{{\bf u}}\\ {\bf g}^{T}{\bf u}\geq 0}}\frac{1}{n}\|{\bf x}\|{\bf g}^{T}{\bf u}-\frac{1}{n}\|{\bf u}\|{\bf h}^{T}{\bf x} -\frac{\sqrt{\rho}{\bf u}^{T}{\bf s}}{\sqrt{n}} -\frac{\|{\bf u}\|^2}{4} +\frac{\lambda \|{\bf x}\|^2}{n}. \label{eq:phi_tilde}
\end{align}
If we discard the constraint ${\bf g}^{T}{\bf u}\geq 0$ and fixing the magnitude of ${\bf u}$ at $\beta\geq 0$, we can easily see that the optimal ${\bf u}$ is given by 
$$
{\bf u}^{\star}=\beta\frac{\frac{1}{n}\|{\bf x}\|{\bf g}-\sqrt{\frac{\rho}{n}}{\bf s}}{\|\frac{1}{n}\|{\bf x}\|{\bf g}-\sqrt{\frac{\rho}{n}}{\bf s}\|}
$$
and satisfies with probability approaching $1$ ${\bf g}^{T}{\bf u}^\star\geq 0$. Replacing ${\bf u}$ by ${\bf u}^\star$ in \eqref{eq:phi_tilde},
\begin{equation}
\tilde{\phi}({\bf g},{\bf h})\geq  \min_{\substack{x_i^2\leq P}}\mathcal{L}_{\lambda,\rho}({\bf x}).\label{eq:lr}
\end{equation}
It follows from Lemma \ref{lem:technical_1} that
$$
\min_{\substack{x_i^2\leq P}}\mathcal{L}_{\lambda,\rho}({\bf x})\to \overline{\phi}
$$
or equivalently for any $\eta>0$, with probability approaching $1$,
$$
\min_{\substack{x_i^2\leq P}}\mathcal{L}_{\lambda,\rho}({\bf x})\geq \overline{\phi}-\eta,
$$
which in view of \eqref{eq:lr} implies \eqref{eq:refq}. This proves the convergence in \eqref{eq:d}.
\subsubsection{Proof of \eqref{eq:conv0}}
It can be easily checked that ${\bf h}^{T}\overline{\bf x}^{\rm AO}\geq 0$. Hence, function $\tilde{\mathcal{F}}_{\lambda,\rho}$ is $\frac{1}{2}$-strongly concave. Moreover, we may use the same calculations as in Lemma \ref{lem:technical_1} to optimize over the direction and magnitude of ${\bf u}$. By doing so, we find:
\begin{align}
\max_{{\bf u}\in\mathcal{S}_{{\bf u}}}\tilde{\mathcal{F}}_{\lambda,\rho}({\bf u})&=\max_{0\leq \beta\leq C_\beta}\beta \left\|\frac{1}{n}\|\overline{\bf x}^{\rm AO}\|{\bf g}-\sqrt{\frac{\rho}{n}}{\bf s}\right\| -\beta \frac{1}{n}{\bf h}^{T}\overline{\bf x}^{\rm AO} -\frac{\beta^2}{4} +\frac{\lambda}{n}\|\overline{\bf x}^{\rm AO}\|^2.
\end{align}
To continue, we use  the fact that:
$$
\frac{1}{n}\|\tilde{\bf x}^{\rm AO}-\overline{\bf x}^{\rm AO}\|\to 0,
$$
to show that 
\begin{equation}
\max_{{\bf u}} \tilde{\mathcal{F}}_{\lambda,\rho}({\bf u}) - \tilde{\mathcal{L}}_{\lambda,\rho}(\tilde{\bf x}^{\rm AO}) \to 0. \label{eq:pre}
\end{equation}
From Lemma \ref{lem:technical_1},
$$
\tilde{\mathcal{L}}_{\lambda,\rho}(\tilde{\bf x}^{\rm AO})-\overline{\phi}\to 0,
$$
which thus in view of \eqref{eq:pre} implies the convergence in \eqref{eq:conv0}. 

\subsubsection{Proof of \eqref{eq:conv3}}

Clearly, the maximizer of $\tilde{\mathcal{F}}_{\lambda,\rho}$ is given by 
$$
\tilde{\bf u}=\tilde{\beta}\frac{\frac{1}{n}\|\overline{\bf x}^{\rm AO}\|{\bf g}-\sqrt{\frac{\rho}{n}}{\bf s}}{\|\frac{1}{n}\|\overline{\bf x}^{\rm AO}\|{\bf g}-\sqrt{\frac{\rho}{n}}{\bf s}\|},
$$
where $\tilde{\beta}$ is given by 
\begin{align}
\tilde{\beta}=\arg\max_{0\leq \beta\leq C_\beta} &\beta \left\|\frac{1}{n}\|\overline{\bf x}^{\rm AO}\|{\bf g}-\sqrt{\frac{\rho}{n}}{\bf s}\right\| -\beta \frac{1}{n}{\bf h}^{T}\overline{\bf x}^{\rm AO} -\frac{\beta^2}{4} +\frac{\lambda}{n}\|\overline{\bf x}^{\rm AO}\|^2. \label{eq:ab}
\end{align}
We can easily see that $\tilde{\beta}-\beta^\star\to 0$. This comes indeed from the facts that (\romannum{1}) $\frac{1}{n}\|\tilde{\bf x}^{\rm AO}-\overline{\bf x}^{\rm AO}\|\to 0$, which implies that the optimal cost in \eqref{eq:ab} converges to $\overline{\phi}$ and (\romannum{2}) the strict concavity of the asymptotic AO in \eqref{eq:asde1} with respect to $\beta$. Using the fact that  $\tilde{\beta}-\beta^\star\to 0$, we can thus easily see that:
\begin{equation}
\|\tilde{\bf u}^{\rm AO}-\breve{\bf u}^{\rm AO}\|\to 0, \label{eq:tuse}
\end{equation}
where
$$
\breve{\bf u}^{\rm AO}= \beta^\star\frac{\frac{1}{n}\|\overline{\bf x}^{\rm AO}\|{\bf g}-\sqrt{\frac{\rho}{n}}{\bf s}}{\|\frac{1}{n}\|\overline{\bf x}^{\rm AO}\|{\bf g}-\sqrt{\frac{\rho}{n}}{\bf s}\|}.
$$
To continue, we use the weak law of large numbers along with the fixed point equation in \eqref{eq:tau_r} to show that 
$$
\frac{1}{\sqrt{n}}\|\overline{\bf x}^{\rm AO}\|\to \sqrt{\delta(\tau^\star)^2-\rho}
$$
and 
$$
\left\|\frac{1}{n}\|\overline{\bf  x}^{\rm AO}\|{\bf g}-\sqrt{\frac{\rho}{n}}{\bf s}\right\|^2\to (\tau^\star)^2\delta^2.
$$
Using the above convergence, it takes no much effort to check that
$$
\|\breve{\bf u}^{\rm AO}-\overline{\bf u}^{\rm AO}\|\to 0,
$$
and thus in view of \eqref{eq:tuse}, we have 
$$
\|\tilde{\bf u}^{\rm AO}-\overline{\bf u}^{\rm AO}\|\to 0.
$$

\section{Proof of the results for the special cases}
\subsection{Proof of Theorem \ref{Theo_rzf}}
\label{Proof_rzf}
	We can prove that as $P\to\infty$, $\beta^\star(P)$ and $\tau^\star(P)$ converge to the solutions of the following max-min problem:
	\begin{equation}
	(\beta^\star,\tau^\star):=\arg\max_{\beta\geq 0}\min_{\tau\geq 0}  \frac{\tau\beta\delta}{2}+\frac{\rho\beta}{2\tau}-\frac{\beta^2}{4}-\frac{1}{2}\frac{\beta}{\frac{1}{\tau}+\frac{2\lambda}{\beta}}.
	\label{1st_condition}
	\end{equation}
	From the first-order optimality conditions, $(\beta_{\rm RZF}^\star,\tau_{\rm RZF}^\star)$ are solutions to the following system of equations:
	\begin{align}
		\tau\delta+\frac{\rho}{\tau}-\beta-\frac{1}{\frac{1}{\tau}+\frac{2\lambda}{\beta}}-\frac{2\lambda}{\beta(\frac{1}{\tau}+\frac{2\lambda}{\beta})^2}&=0, \label{eq:1v}\\
		\delta-\frac{\rho}{\tau^2}-\frac{1}{(1+\frac{2\lambda\tau}{\beta})^2}&=0. \label{eq:2v}
	\end{align}
Let $s=\frac{2\tau}{\beta}$. Then, it follows from \eqref{eq:2v} that 
$$
\frac{\rho}{\tau^2}=\delta-\frac{1}{(1+\lambda s)^2}.
$$
Plugging this relation into \eqref{eq:1v}, we may express \eqref{eq:1v} as 
$$
2\delta-\frac{2}{s}-\frac{1}{(1+\lambda s)^2}-\frac{1}{1+\lambda s}-\frac{\lambda s}{(1+\lambda s)^2}=0,
$$
or equivalently 
$$
\delta-\frac{1}{s}-\frac{1}{1+\lambda s}=0.
$$
The above equation admits a unique positive solution $s^\star$ given by:
\begin{equation}
s^{\star}=\frac{\sqrt{(\delta-\lambda -1)^2+4\delta\lambda}-\delta+\lambda+1}{2\delta \lambda}, \label{eq:star}
\end{equation}
and 
\begin{align}
		\beta_{\rm RZF}^\star&=\frac{2\sqrt{\rho}}{s^\star\sqrt{\delta-\frac{1}{(1+\lambda s^\star)^2}}},\\
		\tau_{\rm RZF}^\star&=\frac{\sqrt{\rho}}{\sqrt{\delta-\frac{1}{(1+\lambda s^\star)^2}}}.
	\end{align}
Finally, plugging the above expressions into the expressions of $P_b^\star$, $P_d^\star$, $P_e^\star$ and $\rm{\overline{SINAD}}_{lb}^\star$ yields the convergences in \eqref{rzf_pb}-\eqref{rzf_pe}.

\section{Proof of the results for the limiting cases}
\label{sec8}
\subsection{Proof of Theorem \ref{th:as_1}}
\label{ass:as_1}
To begin with, we perform the change of variables $\tau'=\tau\sqrt{\delta}$ and $\beta'= \frac{\beta}{\sqrt{\delta}}$ and write $\overline{\phi}$ as 
\begin{equation}
\overline{\phi}=\max_{\beta'\geq 0}\min_{\tau'\geq 0} \frac{\tau'\beta' \delta}{2}+\frac{\rho\beta'\delta}{2\tau'}-\frac{\beta'{}^{2}\delta}{4}+Y(\beta'\sqrt{\delta},\frac{\tau'}{\sqrt{\delta}}) .\label{eq:as}
\end{equation}
Obviously, the saddle point of $\overline{\phi}$ remain the same if we divide the cost by $\delta$. We will thus consider the normalized cost $\overline{\phi}_{\delta}$ given by 
\begin{equation}
\overline{\phi}_{\delta}= \max_{\beta'\geq 0}\min_{\tau'\geq 0} \frac{\tau'\beta'}{2}+\frac{\rho\beta'}{2\tau'}-\frac{\beta'{}^{2}}{4}+\frac{1}{\delta}Y(\beta'\sqrt{\delta},\frac{\tau'}{\sqrt{\delta}}) .\label{eq:ine1}
\end{equation}
 Since $\tau'\mapsto Y(\beta'\sqrt{\delta},\frac{\tau'}{\sqrt{\delta}})$ is decreasing, 
$$
\overline{\phi}_{\delta}\geq \max_{\beta'\geq 0}\min_{\tau'\geq 0} \frac{\tau'\beta'}{2}+\frac{\rho\beta'}{2\tau'}-\frac{\beta'{}^{2}}{4}+\lim_{\tau'\to\infty} \frac{1}{\delta}Y(\beta'\sqrt{\delta},\frac{\tau'}{\sqrt{\delta}}).
$$
Moreover, we can easily check that 
\begin{align}
&\lim_{\tau'\to\infty} \frac{1}{\delta}Y(\beta'\sqrt{\delta},\frac{\tau'}{\sqrt{\delta}})\nonumber\\
&=\frac{(\beta')^2}{2\lambda} \left\{\mathbb{E}\left[(H-\sqrt{P}\frac{2\lambda}{\beta'\sqrt{\delta}})^2{\bf 1}_{\{H\geq \sqrt{P}\frac{2\lambda}{\beta'\sqrt{\delta}}\}}\right]-\frac{1}{2}\right\}\\
&\geq -\frac{(\beta')^2}{4\lambda}.
\end{align}
This together with \eqref{eq:ine1} yields 
\begin{equation}
\overline{\phi}_{\delta}\geq \max_{\beta'\geq 0}\min_{\tau'\geq 0} \frac{\tau'\beta'}{2}+\frac{\rho\beta'}{2\tau'}-\frac{\beta'{}^{2}}{4}-\frac{(\beta')^2}{4\lambda}.\label{eq:lb1}
\end{equation}
On the other hand, we have
\begin{align}
\overline{\phi}_{\delta}&\leq \max_{\beta \geq 0} \left[\frac{\tau'\beta'}{2}+\frac{\rho\beta'}{2\tau'}-\frac{\beta'{}^{2}}{4}+\frac{1}{\delta}Y(\beta'\sqrt{\delta},\frac{\tau'}{\sqrt{\delta}})\right]_{\tau'=\sqrt{\rho}}\\
&= \max_{\beta \geq 0} \min_{\tau'\geq 0} \frac{\tau'\beta'}{2}+\frac{\rho\beta'}{2\tau'}-\frac{\beta'{}^{2}}{4}+\frac{1}{\delta}Y(\beta'\sqrt{\delta},\frac{\sqrt{\rho}}{\sqrt{\delta}}), \label{eq:ref}
\end{align}
where the last equality follows by noticing that function $\frac{\tau'}{2}+\frac{\rho}{2\tau'}$ takes its minimum when $\tau'=\sqrt{\rho}$. 
Function $\beta'\mapsto \min_{\tau'\geq 0} \frac{\tau'\beta'}{2}+\frac{\rho\beta'}{2\tau'}-\frac{\beta'{}^{2}}{4}+\frac{1}{\delta}Y(\beta'\sqrt{\delta},\frac{\sqrt{\rho}}{\sqrt{\delta}})$ is concave and tends to $-\infty$ as $\beta\to\infty$. Moreover, as $\delta$ tends to zero, 
$$
\lim_{\delta\to 0} \frac{1}{\delta}Y(\beta'\sqrt{\delta},\frac{\sqrt{\rho}}{\sqrt{\delta}})= -\frac{(\beta')^2}{4\lambda}.
$$
Hence, using Lemma 10 in \cite{mesti}, we thus have 
\begin{align*}
&\lim_{\delta\to 0} \max_{\beta' \geq 0} \min_{\tau'\geq 0} \frac{\tau'\beta'}{2}+\frac{\rho\beta'}{2\tau'}-\frac{\beta'{}^{2}}{4}+\frac{1}{\delta}Y(\beta'\sqrt{\delta},\frac{\sqrt{\rho}}{\sqrt{\delta}})\nonumber\\
& =  \max_{\beta' \geq 0} \min_{\tau'\geq 0} \frac{\tau'\beta'}{2}+\frac{\rho\beta'}{2\tau'}-\frac{(\beta')^2}{4}-\frac{(\beta')^2}{4\lambda}.
\end{align*}
Combining this with \eqref{eq:ref} and \eqref{eq:lb1} we obtain 
$$
\lim_{\delta\to 0} \overline{\phi}_{\delta}=  \max_{\beta' \geq 0} \min_{\tau'\geq 0} \frac{\tau'\beta'}{2}+\frac{\rho\beta'}{2\tau'}-\frac{(\beta')^2}{4}-\frac{(\beta')^2}{4\lambda}.
$$
The saddle point of the limiting max-min problem in the above equation is unique and is given by: $(\overline{\beta}'=\frac{2\sqrt{\rho}}{1+\frac{1}{\lambda}},\overline{\tau}'=\sqrt{\rho})$. Therefore, letting $\hat{\beta}'(\delta)$ and $\hat{\tau}'(\delta)$ be the saddle point of the max-min problem in \eqref{eq:as}, we have 
$$
\lim_{\delta\to 0}\hat{\beta'}(\delta)= \frac{2\sqrt{\rho}}{1+\frac{1}{\lambda}}
$$
and 
$$
\lim_{\delta\to 0}\hat{\tau}'(\delta)= \sqrt{\rho}.
$$
Finally, going back to the original variables $\tau=\frac{\tau'}{\sqrt{\delta}}$ and $\beta=\sqrt{\delta}\beta'$ yields the convergences in \eqref{tau_delta_zero} and \eqref{beta_delta_zero}. Using these convergences in the asymptotic expressions of $P_b^\star$, $P_d^\star$, $P_e^\star$ and $\overline{\rm SINAD}_{\rm lb}^\star$, we can directly obtain \eqref{eq:P_b_delta_zero}-\eqref{eq:SINR_delta_zero}.

\subsection{Proof of Theorem \ref{th:as_2}}
\label{ass:as_2}
Similar to the proof of Theorem~\ref{th:as_1}, we work with the change of variable $\tau'=\tau\sqrt{\delta}$ and $\beta'=\frac{\beta}{\sqrt{\delta}}$ and consider the max-min problem in \eqref{eq:ine1}. As a first step, we prove that the optimal $\tau'$ lies in the bounded interval $[0,2\sqrt{\rho+1}]$ for sufficiently large $\delta$. To see this, we begin by  rewriting $\overline{\phi}_\delta$ as 
\begin{equation}
\overline{\phi}_\delta=\max_{\beta'\geq 0} \beta' \left(\min_{\tau'\geq 0} \frac{\tau'}{2}+\frac{\rho}{2\tau'}+\frac{1}{\delta}\overline{Y}(\beta',\tau')\right)-\frac{(\beta')^2}{4} ,\label{eq:er}
\end{equation}
where:
$$
\overline{Y}(\beta',\tau')=\frac{\sqrt{\delta}}{\alpha_\delta} \left(\mathbb{E}\left[(H-\sqrt{P}\alpha_\delta)^2{\bf 1}_{\{H\geq \sqrt{P}\alpha_\delta\}}\right]-\frac{1}{2}\right),
$$
where $\alpha_\delta=\frac{\sqrt{\delta}}{\tau'}+\frac{2\lambda}{\sqrt{\delta}\beta'}$. Next, we define $\overline{\phi}_{\delta,1}(\beta)$ and $\overline{\phi}_{\delta,2}(\beta)$ as the optimal costs of the minimization problem in \eqref{eq:er} when $\tau$ is constrained in the interval $[0,2\sqrt{\rho+1}]$ and in the interval $[2\sqrt{\rho+1},\infty]$, respectively, namely:
\begin{align}
	\overline{\phi}_{\delta,1}(\beta)&=  \min_{0\leq\tau'\leq 2\sqrt{\rho+1}} \frac{\tau'}{2}+\frac{\rho}{2\tau'}+\frac{1}{\delta}\overline{Y}(\beta',\tau'),\\
	\overline{\phi}_{\delta,2}(\beta)&= \min_{\tau'\geq 2\sqrt{\rho+1}} \frac{\tau'}{2}+\frac{\rho}{2\tau'}+\frac{1}{\delta}\overline{Y}(\beta',\tau').
\end{align}
Since $\frac{1}{\delta}\overline{Y}(\beta',\tau')\leq 0$, for any fixed $\beta'$,
\begin{align}
\overline{\phi}_{\delta,1}(\beta)\leq \min_{0\leq\tau'\leq 2\sqrt{\rho+1}} \frac{\tau'}{2}+\frac{\rho}{2\tau'}
=\sqrt{\rho}. \label{eq:refd}
\end{align}
On the other hand, one can easily check that:
$$
\frac{1}{\delta}\overline{Y}(\beta',\tau')\geq -\frac{1}{\sqrt{\delta}\alpha_\delta}\geq -\frac{\tau'}{\delta}.
$$
Hence,
$$
\overline{\phi}_{\delta,2}(\beta)\geq \min_{\tau'\geq 2\sqrt{\rho+1}} \frac{\tau'}{2}+\frac{\rho}{2\tau'}-\frac{\tau'}{\delta}.
$$
Now, function $\tau'\mapsto \frac{\tau'}{2}+\frac{\rho}{2\tau'}-\frac{\tau'}{\delta}$ is non-decreasing on the interval $[\frac{\sqrt{\rho}}{\sqrt{1-\frac{1}{\delta}}},\infty)$. Hence, for $\delta\geq \frac{4}{3}$, this function is non-decreasing on the interval $[2\sqrt{\rho+1},\infty)$. Consequently, for sufficiently large $\delta$,
\begin{align}
\overline{\phi}_{\delta,2}(\beta)\geq \sqrt{\rho+1}\left(1-\frac{2}{\delta}\right)+\frac{\rho}{4\sqrt{\rho+1}}
=\sqrt{\rho+1}\left(\frac{5}{4}-\frac{2}{\delta}\right)-\frac{1}{4\sqrt{\rho+1}}.
\end{align}
Obviously, when $\delta\geq 8$, one can check that:
$$
\sqrt{\rho+1}\left(\frac{5}{4}-\frac{2}{\delta}\right)-\frac{1}{4\sqrt{\rho+1}}\geq \sqrt{\rho+1}-\frac{1}{4\sqrt{\rho+1}}\geq \sqrt{\rho},
$$
and thus:
\begin{equation}
\overline{\phi}_{\delta,2}(\beta)\geq \sqrt{\rho} \label{eq:refd1}.
\end{equation}
Combining \eqref{eq:refd} with \eqref{eq:refd1}, we thus conclude that for $\delta\geq 8$, the optimal $\tau'$ lies in the interval $[0,2\sqrt{\rho+1}]$. As a result, for sufficiently large $\delta$,  $\overline{\phi}_\delta$ writes as 
$$
\overline{\phi}_{\delta}=\max_{\beta'\geq 0} \beta' \left(\min_{0\leq\tau'\leq 2\sqrt{\rho+1}} \frac{\tau'}{2}+\frac{\rho}{2\tau'}+\frac{1}{\delta}\overline{Y}(\beta',\tau')\right)-\frac{(\beta')^2}{4}.
$$
With the new rewriting above of $\overline{\phi}_{\delta}$, we are now ready to prove the desired result. For that, similar to Lemma \ref{th:as_1}, we prove that
\begin{equation}
\lim_{\delta\to\infty}\overline{\phi}_{\delta}= \max_{\beta'\geq 0} \beta' \left(\min_{0\leq\tau'\leq 2\sqrt{\rho+1}} \frac{\tau'}{2}+\frac{\rho}{2\tau'}\right)-\frac{(\beta')^2}{4}.\label{eq:dd}
\end{equation}
Indeed, if \eqref{eq:dd} holds true, then one can easily check that the saddle point of the max-min problem is unique and is equal to $\overline{\beta}'=2\sqrt{\rho}$, $\overline{\tau}'=\sqrt{\rho}$. Hence, denoting by $\hat{\beta'}(\delta)$ and $\hat{\tau}'(\delta)$ the saddle point of the optimization problem in \eqref{eq:er}, we thus have 
\begin{align}
	\lim_{\delta\to\infty}\hat{\beta}'(\delta)&=2\sqrt{\rho}, \label{eq:r1}\\
		\lim_{\delta\to\infty}\hat{\tau}'(\delta)&=\sqrt{\rho}. \label{eq:r2}
\end{align}
Hence, going back to the original variables $\tau=\frac{\tau'}{\sqrt{\delta}}$ and $\beta=\sqrt{\delta}\beta'$ yields the sought-for result. 
To prove \eqref{eq:dd}, we use the fact that since the objective function tends to $-\infty$ as $\beta'$ grows to $\infty$, the optimal $\beta'$ is bounded by a constant $C_{\beta'}$. Based on this, it clearly  suffices to show that
\begin{equation}
\lim_{\delta\to\infty}\sup_{\substack{0\leq \beta'\leq C_{\beta'}\\ 0\leq \tau'\leq 2\sqrt{\rho+1}}} \left|\frac{1}{\delta}\overline{Y}(\beta',\tau')\right|\to 0  \label{eq:dr}
\end{equation}
to obtain \eqref{eq:dd} and thus the desired results in \eqref{eq:r1} and \eqref{eq:r2}. To show \eqref{eq:dr}, we use the fact that since $\overline{Y}(\beta',\tau')\leq 0$, we have
$$
\left|\frac{1}{\delta}\overline{Y}(\beta',\tau')\right|\leq \frac{1}{2\sqrt{\delta}\alpha_\delta}\leq \frac{\tau'}{\delta},
$$
and thus,
$$
\sup_{\substack{0\leq \beta'\leq C_{\beta'}\\ 0\leq \tau'\leq 2\sqrt{\rho+1}}} \left|\frac{1}{\delta}\overline{Y}(\beta',\tau')\right|\leq \frac{2\sqrt{\rho+1}}{\delta}\underset{\delta\to\infty}{\to} 0.
$$
This completes the proof of \eqref{eq:r1} and \eqref{eq:r2} and thus that of the convergences in \eqref{tau_delta_inf} and \eqref{beta_delta_inf}. Using  these convergences into the expressions of $P_b^\star$, $P_d^\star$, $P_e^\star$ and $\overline{\rm SINAD}_{\rm lb}^\star$, the convergences in \eqref{P_b_delta_inf}-\eqref{P_e_delta_inf} follow easily. 

\subsection{Proof of Theorem \ref{th:power_control_zero}}
\label{Proof_rho0}

To begin with, we perform the change of variables $\tau^{'}=\tau/\sqrt{\rho}$ and $\beta^{'}=\beta/\sqrt{\rho}$ to write $\bar{\phi}$ as 

\begin{equation}
\overline{\phi}=\max_{\beta'\geq 0}\min_{\tau'\geq 0} \frac{\tau'\beta' \delta\rho}{2}+\frac{\rho\beta'}{2\tau'}-\frac{\beta'^{2}\rho}{4}+Y(\beta'\sqrt{\rho},\tau'\sqrt{\rho}). \label{up_eq:as}
\end{equation}
Obviously, the saddle point of $\overline{\phi}$ remain the same if we divide the cost by $\rho$. We will thus consider the normalized cost $\overline{\phi}_{\rho}$ given by 
\begin{equation}
\overline{\phi}_{\rho}= \max_{\beta'\geq 0}\min_{\tau'\geq 0} \frac{\tau'\beta'\delta}{2}+\frac{\beta'}{2\tau'}-\frac{\beta'{}^{2}}{4}+\frac{1}{\rho}Y(\beta'\sqrt{\rho},\tau'\sqrt{\rho}). \label{up_eq:ine1}
\end{equation}
It takes no much effort to see that 
\begin{equation}
\begin{split}
&\frac{1}{\rho}Y(\beta'\sqrt{\rho},\tau'\sqrt{\rho}) 
 \geq-\frac{1}{2}\frac{\beta'}{\frac{1}{\tau'}+\frac{2\lambda}{\beta'}},\label{lb}
\end{split}
\end{equation}
which directly implies that 
\begin{equation}
\overline{\phi}_{\rho}\geq \overline{\phi}_{\rho,{\rm lb}}, \label{eq:phi_lb}
\end{equation}
where 
\begin{equation}
	 \overline{\phi}_{\rho,{\rm lb}}:=\max_{\beta'\geq 0}\min_{\tau'\geq 0} \frac{\tau'\beta'\delta}{2}+\frac{\beta'}{2\tau'}-\frac{\beta'{}^{2}}{4}-\frac{1}{2}\frac{\beta'}{\frac{1}{\tau'}+\frac{2\lambda}{\beta'}}.\label{phi_lb}
\end{equation}
Denote by $\overline{\tau}'$ and $\overline{\beta}'$  optimal solutions in $(\tau,\beta)$ to \eqref{phi_lb}. It is easy to check that the optimization problem in \eqref{phi_lb} is the same one as in \eqref{1st_condition} (Proof of Theorem \ref{Theo_rzf}) when $\rho=1$. Hence, based on the proof of Theorem \ref{Theo_rzf}, we conclude that $\overline{\tau}'$ and $\overline{\beta}'$ are unique and are given by 
\begin{align}
	\overline{\tau}'&=\frac{1}{\sqrt{\delta-\frac{1}{(1+\lambda s^\star)^2}}},\\
	\overline{\beta}'&=\frac{2}{s^\star\sqrt{\delta-\frac{1}{(1+\lambda s^\star)^2}}},
\end{align}
where $s^\star$ is defined in \eqref{eq:star}. To continue, we note that 
$$
\overline{\phi}_{\rho}\leq \overline{\phi}_{\rho,\rm{up}},
$$
where
\begin{equation}
	\overline{\phi}_{\rho,\rm{up}}:=\max_{\beta'\geq 0} \frac{\overline{\tau}'\beta'\delta}{2} +\frac{\beta'}{2\overline{\tau}'}-\frac{\beta'^2}{4}+\frac{1}{\rho}Y(\beta'\sqrt{\rho},\overline{\tau}'\sqrt{\rho}).\label{eq:phi_up}
\end{equation}
The objective function in \eqref{eq:phi_up} is concave in $\beta'$ and tends to $-\infty$ as $\beta'$ grows to infinty. From Lemma 10 in \cite{mesti}, 
$$
\lim_{\rho\to 0} \overline{\phi}_{\rho,\rm{up}} = \max_{\beta'\geq 0} \frac{\overline{\tau}'\beta'\delta}{2} +\frac{\beta'}{2\overline{\tau}'}-\frac{\beta'^2}{4}+\lim_{\rho\to 0}\frac{1}{\rho}Y(\beta'\sqrt{\rho},\overline{\tau}'\sqrt{\rho}).
$$
For fixed $\beta'$, it can be readily seen that 
$$
\lim_{\rho\to 0  }\frac{1}{\rho}Y(\beta'\sqrt{\rho},\overline{\tau}'\sqrt{\rho}) = -\frac{1}{2} \frac{\beta'}{\frac{1}{\overline{\tau}'}+\frac{2\lambda}{\beta'}}
$$
and thus 
$$
\lim_{\rho\to 0} \overline{\phi}_{\rho,\rm{up}} = \max_{\beta'\geq 0} \frac{\overline{\tau}'\beta'\delta}{2} +\frac{\beta'}{2\overline{\tau}'}-\frac{\beta'^2}{4}-\frac{1}{2}\frac{\beta'}{\frac{1}{\overline{\tau}'}+\frac{2\lambda}{\beta'}}.
$$
This in combination with \eqref{eq:phi_lb} yields:
\begin{equation}
\lim_{\rho\to 0}\overline{\phi}_\rho= \max_{\beta'\geq 0} \frac{\overline{\tau}'\beta'\delta}{2} +\frac{\beta'}{2\overline{\tau}'}-\frac{\beta'^2}{4}-\frac{1}{2}   \frac{\beta'}{\frac{1}{\overline{\tau}'}+\frac{2\lambda}{\beta'}}.\label{eq:final}
\end{equation}
Let  $\hat{\beta}'(\rho)$ and $\hat{\tau}'(\rho)$ be solutions to \eqref{up_eq:as}. Then, based on \eqref{eq:final} and using the uniqueness of the solutions $\overline{\tau}'$ and $\overline{\beta}'$, we have 
\begin{align}
	\lim_{\rho\to 0} \hat{\tau}'(\rho)& = \overline{\tau}',\\
	\lim_{\rho\to 0} \hat{\beta}'(\rho)&= \overline{\beta}'.
\end{align}
Hence, going back to the original variables $\tau^{'}=\tau/\sqrt{\rho}$ and $\beta^{'}=\beta/\sqrt{\rho}$ yields the convergences in \eqref{eq:tau_rho} and \eqref{eq:beta_rho}.

\subsection{Proof of Theorem \ref{th:limit_rho}}
\label{app:limit_rho}
The proof is organized in two parts. In the first part, we follow the same line steps as in the proof of Theorem \ref{th:as_2} to show that 
\begin{equation}
\lim_{\rho\to\infty} \frac{\tau^\star(\rho)}{\sqrt{\frac{\rho}{\delta}}} = 1 \ \ \mathrm{and} \ \  \lim_{\rho\to \infty} \frac{\beta^\star(\rho)}{2\sqrt{\rho\delta}}=1 . \label{eq:first_step}
\end{equation}
In a second step, we use the fixed point equations in \eqref{eq:tau_r} to refine the approximation in \eqref{eq:first_step}.

\noindent{\bf First part: Proof of \eqref{eq:first_step}.} We start by performing the change of variables $\tau'=\frac{\tau}{\sqrt{\rho}}$ and $\beta'=\frac{\beta}{\sqrt{\rho}}$ to write $\overline{\phi}$ as 
$$
\overline{\phi}=\max_{\beta'\geq 0} \min_{\tau'\geq 0} \frac{\tau'\beta'\rho\delta}{2} +\frac{\rho\beta'}{2\tau'} - \frac{\rho(\beta')^2}{4} +Y(\beta'\sqrt{\rho},\tau'\sqrt{\rho}).
$$
For convenience, we consider the normalized cost $\overline{\phi}_{\rho}$ given by 
\begin{align}
\overline{\phi}_\rho= \max_{\beta'\geq 0} \min_{\tau'\geq 0}& \beta'\left(\frac{\tau'\delta}{2} +\frac{1}{2\tau'}  +
\frac{1}{\sqrt{\rho}}\overline{Y}_{\rho,1}(\alpha')-\frac{1}{2\alpha'}\overline{Y}_{\rho,2}(\alpha')\right) - \frac{(\beta')^2}{4}, \label{eq:prho}
\end{align}
where
\begin{align}
	\overline{Y}_{\rho,1}(\alpha')&=\sqrt{P}\mathbb{E}\left[(\frac{\sqrt{P}}{\sqrt{\rho}}\alpha'-2H){\bf 1}_{\{H\geq \frac{\sqrt{P}}{\sqrt{\rho}}\alpha'\}}\right],\\
		\overline{Y}_{\rho,2}(\alpha')&=\mathbb{E}\left[H^2{\bf 1}_{\{-\frac{\sqrt{P}}{\sqrt{\rho}}\alpha'\leq H\leq \frac{\sqrt{P}\alpha'}{\sqrt{\rho}}\}}\right],
\end{align}
with $\alpha'=\frac{1}{\tau'}+\frac{2\lambda}{\beta'}$. Note that to find \eqref{eq:prho}, we used the expression of $Y(\beta,\tau)$ in \eqref{eq:rep1} instead of \eqref{eq:rr}. 

To continue, we need to show that
\begin{equation}
	\lim_{\rho\to\infty} \sup_{\substack{\beta'\geq 0\\ \tau'\geq 0 }} \frac{1}{\sqrt{\rho}} \left|\overline{Y}_{\rho,1}(\alpha')\right| \to 0, \label{eq:relation_1}
\end{equation}
and
\begin{equation}
	\lim_{\rho\to\infty} \sup_{\substack{\beta'\geq 0\\ \tau'\geq 0 }} \frac{1}{2\alpha'}\overline{Y}_{\rho,2}(\alpha') \to 0. \label{eq:relation_2}
\end{equation}
\noindent{\underline{{Proof of \eqref{eq:relation_1}}}}. To begin with, we note that 
	\begin{align}
	 \frac{1}{\sqrt{\rho}} \left|\overline{Y}_{\rho,1}(\alpha')\right|\leq \frac{\sqrt{P}}{\sqrt{\rho}}\mathbb{E}\left[|2H|{\bf 1}_{\{H\geq \frac{\sqrt{P}}{\sqrt{\rho}}\alpha'\}}\right] \leq \frac{\sqrt{P}}{\sqrt{\rho}}\mathbb{E}\left[|2H|\right].\label{eq:e1}
	\end{align}
Hence, 
$$
 \sup_{\substack{\beta'\geq 0\\ \tau'\geq 0 }} \frac{1}{\sqrt{\rho}} \left|\overline{Y}_{\rho,1}(\alpha')\right|\leq \frac{\sqrt{P}}{\sqrt{\rho}}\mathbb{E}\left[|2H|\right]\underset{\rho\to\infty }{\to}0 .
$$

\noindent{\underline{{Proof of \eqref{eq:relation_2}}}}.
Clearly, $\overline{Y}_{\rho,2}(\alpha')$ can be upper bounded as:
\begin{align}
|\overline{Y}_{\rho,2}(\alpha')|&\leq  \alpha'\frac{\sqrt{P}}{\sqrt{\rho}} \mathbb{E}\left[|H|{\bf 1}_{\{|H|\leq \frac{\sqrt{P}}{\sqrt{\rho}}\alpha'\}}\right]\\
&\leq \alpha'\frac{\sqrt{P}}{\sqrt{\rho}} \mathbb{E}[|H|],
\end{align}
which yields:
$$
\sup_{\substack{\beta'\geq 0\\ \tau'\geq 0} }\frac{1}{2\alpha'}|\overline{Y}_{\rho,2}(\alpha')|\leq \frac{1}{2}\frac{\sqrt{P}}{\sqrt{\rho}} \mathbb{E}[|H|]\underset{\rho\to\infty}\to 0.
$$
With \eqref{eq:relation_1} and \eqref{eq:relation_2} at hand, we obtain 
$$
\lim_{\rho\to\infty} \overline{\phi}_{\rho}=\max_{\beta'\geq 0}\min_{\tau'\geq 0} \beta'\left(\frac{\tau'\delta}{2}+\frac{1}{2\tau'}\right)-\frac{(\beta')^2}{4}.
$$
The above limiting optimization problem possesses a unique saddle point given by $(\beta'=2\sqrt{\delta}, \tau'=\frac{1}{\sqrt{\delta}})$. Denoting  by $\hat{\beta}'(\rho)$ and $\hat{\tau}'(\rho)$ the saddle point of the optimization problem in \eqref{eq:prho}, we thus obtain 
$$
\lim_{\rho\to\infty} \hat{\beta}'(\rho) = 2\sqrt{\delta} \ \ \mathrm{and}\ \ \lim_{\rho\to\infty} \hat{\tau}'(\rho)=\frac{1}{\sqrt{\delta}}.
$$
Hence, going back to the original variables $\tau=\sqrt{\rho}\tau'$ and $\beta=\sqrt{\rho}\beta'$ yields the convergence in \eqref{eq:first_step}. 

\noindent{{\bf Second part: {Approximation's refinements.}}}
Denoting by $\alpha^\star(\rho)=\frac{1}{\tau^\star(\rho)}+\frac{2\lambda}{\beta^\star(\rho)}$,  it follows from the convergence in \eqref{eq:first_step} that $\alpha^\star(\rho)$ goes to zero and verifies the following convergence:
$$
\lim_{\rho\to \infty} \frac{\alpha^\star(\rho)}{\sqrt{\frac{\delta}{\rho}}+\frac{\lambda}{\sqrt{\rho\delta}}}=1.
$$
To continue, we exploit the fixed point equation in \eqref{eq:tau_r} and rewrite it as 
\begin{equation}
(\tau^\star)^2\delta-\rho=2Pf_1(\sqrt{P}\alpha^\star)+\frac{1}{(\alpha^\star)^2}(1-2f_2(\sqrt{P}\alpha^\star)), \label{eq:fixed}
\end{equation}
where functions $f_1$ and $f_2$ are given by 
\begin{align}
f_1(x)&=\frac{1}{\sqrt{2\pi}}\int_{\sqrt{P}x}^{\infty} \exp(-\frac{t^2}{2})dt,\\
f_2(x)&=\frac{1}{\sqrt{2\pi}}\int_{\sqrt{P}x}^\infty t^2\exp(-\frac{t^2}{2})dt.
\end{align}
By using standard calculations, we may expand the Taylor expansion of $f_1(\sqrt{P}\alpha^\star)$ and $f_2(\sqrt{P}\alpha^\star)$ for $\alpha$ near zero as 
\begin{align}
	f_1(\sqrt{P}\alpha^\star)&= \frac{1}{2}-\frac{\sqrt{P}}{\sqrt{2\pi}} \alpha^\star +O((\alpha^\star)^2),\\
	f_2(\sqrt{P}\alpha^\star)&= \frac{1}{2}-\frac{P^3(\alpha^\star)^3}{3\sqrt{2\pi}}+O((\alpha^\star)^4).
\end{align}
Plugging the above approximations into \eqref{eq:fixed} yields:
\begin{align}
&(\tau^\star(\rho))^2\delta-\rho= P+\frac{2P^2}{\sqrt{2\pi}}\left(\frac{P}{3}-1\right)\alpha^\star +O((\alpha^\star)^2)\\
&=P+\frac{2P^2}{\sqrt{2\pi}}\left(\frac{P}{3}-1\right)\left(\sqrt{\frac{\delta}{\rho}}+\frac{\lambda}{\sqrt{\rho\delta}}\right) +O(\frac{1}{\rho})
\end{align}
and hence, we get after straightforward calculations:
\begin{equation}
\tau^\star(\rho)=\frac{\sqrt{\rho}}{\sqrt{\delta}}+\frac{P}{2\sqrt{\delta{\rho}}}+O(\frac{1}{\rho}). \label{eq:t}
\end{equation}
To find a similar approximation for $\beta^\star(\rho)$, we first take the derivative in \eqref{eq:asde1} with respect to $\beta$ to obtain the following relation
\begin{equation}
\beta^\star(\rho)=\tau^\star(\rho)\delta +\frac{\rho}{\tau^\star(\rho)} +2\frac{\partial Y}{\partial \beta}(\beta^\star,\tau^\star) .\label{eq:b}
\end{equation}
Simple calculations leads to 
$$
\frac{\partial Y}{\partial \beta}= -\frac{\sqrt{2P}}{\sqrt{\pi}}+O(\frac{1}{\sqrt{\rho}}).
$$
Plugging this together with \eqref{eq:t} into \eqref{eq:b} yields 
\begin{equation}
\beta^\star(\rho)=2\sqrt{\rho\delta}-2\frac{\sqrt{2P}}{\sqrt{\pi}}+O(\frac{1}{\sqrt{\rho}}). \label{eq:beta_equiv}
\end{equation}
Finally, plugging the asymptotic equivalences \eqref{eq:t} and \eqref{eq:beta_equiv} into the asymptotic expressions of $P_b^\star$, $P_d^\star$, $P_e^\star$ and $\overline{\rm SINAD}_{\rm lb}$, we obtain the convergences in \eqref{eq:P_b_rho_inf}-\eqref{eq:SINR_rho_inf}.

\section{A Useful technical Lemma}
 \begin{lemma}[Lemma B1 in \cite{Miolane}]
	\label{app:technical}
	Let $f$ be a convex function in $\mathbb{R}^n$. Let $\overline{\bf x}$ be in $\mathbb{R}^n$ and $r>0$. Assume that $f$ is $\gamma$-strongly convex on the ball $\mathcal{B}(\overline{\bf x},r)$ for some $\gamma>0$. Assume that 
	$$
	f(\overline{\bf x}) \leq \min_{{\bf x}\in\mathcal{B}(\overline{\bf x},r)} f({\bf x}) +\epsilon 
	$$
	for some $\epsilon<\frac{r^2\gamma}{8}$. Then, the following statements hold true:
	\begin{enumerate}
		\item $f$ admits a unique minimizer ${\bf x}^\star$ over $\mathbb{R}^n$. Moreover, ${\bf x}^{\star}\in\mathcal{B}(\overline{\bf x},r)$ and hence 
		$$
		\|{\bf x}^\star-\overline{\bf x}\|\leq \frac{2}{\gamma}\epsilon.
		$$ 
		\item For every ${\bf x}\in\mathbb{R}^n$,
		$$
		\|{\bf x}-\overline{\bf x}\|^2\geq \frac{8}{\gamma}\epsilon \Longrightarrow f({\bf x})\geq \min f+\epsilon.
		$$
	\end{enumerate}
\end{lemma}

\bibliographystyle{IEEEtran}
\bibliography{ref}

\vfill

\end{document}